\documentclass[]{article}

\usepackage[utf8]{inputenc}
\usepackage{amsmath, amssymb, amsthm}
\usepackage{stackengine,scalerel}
\usepackage{authblk}
\usepackage{graphicx}
\usepackage{ragged2e}
\usepackage{comment}
\usepackage[colorlinks = true,linkcolor = blue,urlcolor = blue,citecolor = blue,anchorcolor = blue]{hyperref}
\usepackage{url}
\usepackage{float}
\usepackage[table]{xcolor}
\usepackage{cite}

\usepackage{wasysym}
\usepackage{relsize}
\usepackage{subcaption}
\usepackage[braket, qm]{qcircuit}

\newcommand{\MCO}[3]{[#1]^{#2}_{#3}}
\newcommand{\sqO}[2]{[#1]_{#2}}
\newcommand{\floor}[1]{\lfloor #1 \rfloor}
\newcommand{\ceil}[1]{\lceil #1 \rceil}

\newcommand{\abs}[1]{| #1 |}
\newcommand{\rpicl}{\Pi}
\newcommand{\rpi}[1]{\rpicl(#1)}
\newcommand{\rpisub}[1]{\rpicl_{#1}}

\newcommand{\rv}[2]{R_{#2}(#1)}

\newcommand{\rx}{R_{\hat{x}}}

\newcommand{\rpigatesub}[1]{\gate{\mathrm{\rpisub{#1}}}}

\newcommand\eye{\ensurestackMath{\stackinset{c}{}{c}{-.3pt}%
  {\bullet}{\scriptstyle\bigcirc}}}  
\newcommand\smallereye{\scalerel*{\eye}{x}}
\newcommand\eyeeyes{\ensurestackMath{\stackinset{c}{}{c}{-.0pt}%
  {\smallereye}{\scriptstyle\bigcirc}}}
  \newcommand\eyeeye{\scalerel*{\eyeeyes}{.2pt}}

\newcommand{\controlt}{*!<0em,.0em>-=-<.2em>{\eye}}
\newcommand{\controltt}{*!<0em,.0em>-=-<.5em>{\eyeeye}}

\newcommand{\target}{*!<0em,.0em>-=-<.2em>{\bigoplus}}

\newcommand{\ctrlt}[1]{\controlt \qwx[#1] \qw}
\newcommand{\ctrltt}[1]{\controltt \qwx[#1] \qw}

\newcommand*\circled[1]{\tikz[baseline=(char.base)]{
            \node[shape=circle,draw,inner sep=0.8pt] (char) {#1};}}

\usepackage{tikz}

\xyoption{color}

\newcommand{\qwl}{\ar @{-/} [0,-1]}
\newcommand{\ctrlp}[2]{\control \qwx[#1] \ar @{<-} [0,-1] _{#2}}
\newcommand{\ctrlpd}[2]{\control \qwx[#1] \ar @{<-} [0,-1] ^{#2}}

\newtheorem{lemma}{Lemma}

\newcommand{\eq}[1]{\ref{eq:#1}}
\renewcommand{\sec}[1]{Section~\hyperref[sec:#1]{\ref*{sec:#1}}}
\newcommand{\ssec}[1]{Section~\hyperref[ssec:#1]{\ref*{ssec:#1}}}
\newcommand{\fig}[1]{Figure~\hyperref[fig:#1]{\ref*{fig:#1}}}
\newcommand{\tab}[1]{Table~\hyperref[tab:#1]{\ref*{tab:#1}}}
\newcommand{\lem}[1]{Lemma~\hyperref[lem:#1]{\ref*{lem:#1}}}
\newcommand{\prop}[1]{\hyperref[prop:#1]{Proposition~\ref*{prop:#1}}}
\newcommand{\thm}[1]{Theorem~\hyperref[thm:#1]{\ref*{thm:#1}}}
\newcommand{\apx}[1]{Appendix~\hyperref[apx:#1]{\ref*{apx:#1}}}
\newcommand{\qc}[1]{Circuit~(\hyperref[qc:#1]{\ref*{qc:#1}})}

\newcounter{qcnum}
\newcommand{\qcref}[1]{\refstepcounter{qcnum}\tag{\theqcnum}\label{qc:#1}}

\usepackage{graphicx}
\font\myfont=cmr12 at 8pt

\title{Multi-Controlled Quantum Gates in Linear Nearest Neighbor} 
\author{Ben Zindorf and Sougato Bose}
\affil{\myfont Department of Physics and Astronomy, University College London,\\
\myfont Gower Street, WC1E 6BT London, United Kingdom}
\date{\today}

\begin{document}

\maketitle

\begin{abstract}
Multi-controlled single-target (MC) gates are some of the most crucial building blocks for varied quantum algorithms.
How to implement them optimally is thus a pivotal question. To answer this question in an architecture-independent manner, and to get a worst-case estimate, we should look at a linear nearest-neighbor (LNN) architecture, as this can be
embedded in almost any qubit connectivity.
Motivated by the above, here we describe a method which implements MC gates using no more than $\sim\! 4k+8n$ CNOT gates -- up-to $60\%$ reduction over state-of-the-art -- while allowing for complete flexibility to choose the locations of $n$ controls, the target, and a dirty ancilla out of $k$ qubits.
More strikingly, in case $k\approx n$, our upper bound is $\sim\!12n$ -- the best known for unrestricted connectivity -- and if $n=1$, our upper bound is $\sim\! 4k$ -- the best known for a single long-range CNOT gate over $k$ qubits --
therefore, if our upper bound can be reduced, then the cost of one or both of these simpler versions of MC gates will be immediately reduced accordingly.
In practice, our method provides circuits that tend to require fewer CNOT gates than our upper bound for almost any given instance of MC gates. 

\end{abstract}

\section{Introduction}

Multi-controlled single-target quantum gates are used as a fundamental building block for many quantum algorithms \cite{arrazola_universal_2022,de_carvalho_parametrized_2024,ni_progressive_2024,iten_quantum_2016, malvetti_quantum_2021,grinko_efficient_2023,tanasescu_distribution_2022,park_circuit-based_2019,ali_function_2018,grover_quantum_1997,lubasch_quantum_2025,suzuki_double-bracket_2025,sambasivam_tepid-adapt_2025,balaji_quantum_2025,gluza_double-bracket_2024,ortega_implementing_2025,bosch_quantum_2025,wang_improved_2025,anderson_solving_2025,nibbi_block_2024}, as they provide a convenient logical framework, being the quantum counterpart of the classical if-then-else statements \cite{lewis_matrix_2022}.
However, in most existing technologies, the MC gates must be implemented using a set of available operations which are constrained to be one or two qubit gates. The Clifford+T gate set \cite{bravyi_universal_2005,nielsen_quantum_2002,jones_logic_2013}, constituted of the Hadamard, CNOT and T gates, is a good example for such a set, as it is a prime candidate for future error-corrected fault-tolerant quantum computation, {\em and} it can be applied to near-term NISQ (Noisy Intermediate-Scale Quantum) devices as well.
Although not all NISQ devices can implement the CNOT directly, there is generally an available two-qubit gate \cite{abughanem_ibm_2024,abughanem_toffoli_2025,xue_quantum_2022} which can be easily converted to a CNOT using few single-qubit rotations which are considered less ''expensive'' due to their lower error rate and shorter execution time, compared to the two-qubit gates. On the contrary, in the fault-tolerant regime, the most expensive resource is the T gate, and consequently, single qubit rotations with arbitrary angles such as the $R_z$ gates, as these are approximated using a very large number of T gates, which depends on the chosen angle and the level of accuracy \cite{bravyi_universal_2005}.
In the quest of finding an efficient implementation of MC gates, one must consider its relevance both for the future and for near-term applications --  use the minimal number of CNOT, T, and arbitrary $R_z$ gates. 

As many quantum computers are constrained by restricted qubit interactions, an implementation of MC gates must be mapped to a given qubit connectivity. A linear nearest-neighbor (LNN) connectivity only allows to apply two-qubit gates on two nearest-neighboring qubits in a 1D array and is therefore considered to be one of the most restrictive options. Moreover, efficient LNN decompositions are highly useful in practice as these can be efficiently mapped to unrestricted all-to-all (ATA) connectivity, as well as to 2D qubit arrangements, which are widely used in existing devices. As there are many possible arrangements of qubits in a 2D structure, such as the square lattice used by Google, and the heavy-hex used by IBM, or ones that may arise in the future, it is useful to focus on LNN which can be then mapped to any of them.
If, on the other hand, one focused on another specific architecture, its applicability would be limited as it may not be efficiently mapped to other ones.
In the extreme, mapping from ATA to LNN could quadratically increase the cost due to many added SWAP gates \cite{pedram_layout_2016,saeedi_synthesis_2011,he_mapping_2019,khan_cost_2008,lukac_optimization_2021,ding_fast_2019,zhang_method_2022}. 
In fact, it was only recently discovered that MC gates in LNN can be implemented using a linear CNOT count \cite{zindorf_efficient_2024}, similarly to the ATA case \cite{barenco_elementary_1995,vale_decomposition_2023,iten_quantum_2016,he_decompositions_2017,balauca_efficient_2022,maslov_reversible_2011,kole_improved_2017,sasanian_ncv_2011,biswal_improving_2016,ali_quantum_2015,niemann_t-depth_2019,leng_decomposing_2024,abdessaied_technology_2016,khattar_rise_2024}, instead of a quadratic gate count \cite{arsoski_implementing_2024,arsoski_multi-controlled_2025,cheng_mapping_2018,chakrabarti_nearest_2007,miller_elementary_2011,li_quantum_2023,tan_multi-strategy_2018}. 
The approach taken in the methods which produce quadratic cost was to first assume a specific choice of the qubits of interest ($n$ controls, the target, and a dirty ancilla out of $k$ qubits), out of $\sim\!\binom{k}{n}$ possible combinations, and provide a linear cost implementation for this choice, then simply relocate the qubits to a chosen location. However, this relocation results in an overhead of CNOT gates that scales quadratically with the number of qubits. The linear-cost method only assumes the location of one qubit (target/ancilla), and allows to implement the MC gate in linear gate count for {\em any} choice of the controls, then a linear overhead is added in case the qubit with assumed location is located elsewhere. A first step in the attempt to bring the LNN cost as close as possible to the ATA one is to provide a structure which can be applied for any choice of all qubits of interest, with no assumptions, and thus requires no qubit relocation. 

In this paper, we focus on minimizing the gate count of CNOT,T and H (Clifford+T) required to implement the multi-controlled Toffoli gate (MCX) in LNN connectivity, using one dirty ancilla qubit, and the multi-controlled $SU(2)$ (MCSU2) gate, without any ancilla requirement. As the latter is a parameterized gate, it will require a number of arbitrary $R_z$ gates, which we also try to minimize. Our method allows to upper-bound the CNOT gate count as $4k+8n-16$ for MCX, and as $4k+8n-14$ for MCSU2.
When the qubit arrangement allows to further reduce the cost, using non-control qubits as dirty ancilla, we apply those reductions. This results in lower gate counts than the upper bound in many cases. In addition, we maximize the parallelization of our structure when it is possible to do so without increasing the gate count.
In the process of developing our structure, we provide an implementation of the MCX gate up to a relative phase (MCX-$\Delta$), which is also very useful for many quantum algorithms \cite{maslov_advantages_2016,oonishi_efficient_2022,kuroda_optimization_2022,selinger_quantum_2013}. As an aside, we also present a new gate, that generalizes the MCX-$\Delta$, which we find to be useful for circuit optimization.



\subsection{Notation}\label{sec:notation}
We introduce a set of notations which will be used throughout the paper, starting with Hermitian $\pi$-rotations \cite{zindorf_all_2025}, followed by multi-controlled gates \cite{zindorf_efficient_2024}.
These notations will help simplify the description of our methods. 

Any single-qubit $W\in SU(2)$ operator can be represented as a rotation about the Bloch sphere by an angle $\lambda$ about an axis $\hat{v}$ and written as $W=\rv{\lambda}{\hat{v}}$ \cite{nielsen_quantum_2010}.
We define a Hermitian $\pi$-rotation about an axis $\hat{v}=(\sin{\theta}\cos{\phi},\sin{\theta}\sin{\phi},\cos{\theta})$, using standard spherical coordinates, as follows.
\[
\rpi{\hat{v}}:=i\rv{\pi}{\hat{v}}=
\left(
\begin{matrix}
\cos{\theta} & e^{-i\phi}\sin{\theta}\\
e^{i\phi}\sin{\theta} & -\cos{\theta} 
\end{matrix}
\right)
\]
Any such $\Pi$ gate satisfies $\rpi{\hat{v}}\rpi{\hat{v}}=I$ and $\Pi(\hat{v})\Pi(-\hat{v})=-I$ for any axis $\hat{v}$.
It has been shown in [\cite{zindorf_all_2025}, Lemma 1] that any $SU(2)$ rotation can be decomposed using two $\Pi$ gates as follows.
\begin{lemma}\label{lem:2_rpi}
Any $\rv{\lambda}{\hat{v}}\in SU(2)$ operator can be implemented as $\rpi{\hat{v}_2}\rpi{\hat{v}_1}$ 
with $\hat{v}_1$ as any unit vector perpendicular to $\hat{v}$ (i.e., $\hat{v}_1 \perp \hat{v}$), and $\hat{v}_2 = \hat{R}_{\hat{v}}(\frac{\lambda}{2})\hat{v}_1$ with $\frac{\lambda}{2}\in (-\pi,\pi]$ ($\hat{v}_2 \perp \hat{v}$ as well).
\end{lemma}

Here, $\hat{R}_{\hat{v}}(\frac{\lambda}{2})$ is the three-dimensional rotation matrix by angle $\frac{\lambda}{2}$ about the axis $\hat{v}$. We now define a notation for $\Pi$ gates about an axis on the planes $\bar{x},\bar{y},\bar{z}$, perpendicular to the $\hat{x},\hat{y}, \hat{z}$ axes. These are visualized in \fig{pi_sphere_angles}.

\[
\Pi^\theta_{\bar{x}}:=\Pi(\hat{v}^\theta_{\bar{x}}), \Pi^\theta_{\bar{y}}:=\Pi(\hat{v}^\theta_{\bar{y}}), \Pi^\phi_{\bar{z}}:=\Pi(\hat{v}^\phi_{\bar{z}}), \text{with }  \hat{v}^\theta_{\bar{x}}:=\hat{R}_{\hat{x}}(\theta)\hat{z}, \hat{v}^\theta_{\bar{y}}:=\hat{R}_{\hat{y}}(\theta)\hat{z}, \hat{v}^\phi_{\bar{z}}:=\hat{R}_{\hat{z}}(\phi)\hat{x}
\]

The Pauli operators $X,Y,Z$ are simply $\Pi$ gates about the axes $\hat{x},\hat{y},\hat{z}$, respectively. Similarly, the Hadamard gate $H=\Pi^{\pi/4}_{\bar{y}}$ is a $\Pi$ gate about the vector $\hat{v}_H:=\hat{v}^{\pi/4}_{\bar{y}} = (\hat{x}+\hat{z})/\sqrt{2}$ located in the middle between $\hat{x}$ and $\hat{z}$, as visualized in \fig{pi_sphere}. Similarly, we define the gates $\Pi_S:=\Pi^{\pi/4}_{\bar{z}}$ and $\Pi_V:=\Pi^{\pi/4}_{\bar{x}}$ about the axes $\hat{v}_S:=\hat{v}^{\pi/4}_{\bar{z}} = (\hat{x}+\hat{y})/\sqrt{2}$ and $\hat{v}_V:=\hat{v}^{\pi/4}_{\bar{x}} = (\hat{z}+(-\hat{y}))/\sqrt{2}$. We are naming these gates this way because of their relation to the $S=\sqrt{Z}=e^{i\tfrac{\pi}{4}} R_{\hat{z}}(\tfrac{\pi}{2})$ and $V=\sqrt{X} = e^{i\tfrac{\pi}{4}} R_{\hat{x}}(\tfrac{\pi}{2})$ gates, such that $R_{\hat{z}}(\tfrac{\pi}{2}) = \Pi_S X$ and $R_{\hat{x}}(\tfrac{\pi}{2}) = \Pi_V Z$ according to \lem{2_rpi}. 

The Pauli, H, $\Pi_S$ and $\Pi_V$ gates are all part of the Clifford group and therefore do not suffice for universal quantum computation. In order to add the missing {\em magic} \cite{bravyi_universal_2005}, we introduce the gate $\Pi_T:=\Pi^{\pi/8}_{\bar{z}}$ about the axis $\hat{v}_T:=\hat{v}^{\pi/8}_{\bar{z}} = \tfrac{\hat{x}+\hat{v}_S}{\abs{\hat{x}+\hat{v}_S}}$  which satisfies $R_{\hat{z}}(\tfrac{\pi}{4})=\Pi_T X$, and corresponds to the $T=\sqrt{S}$ gate. Similarly, we define the gate $\Pi_{Tx}:=\Pi^{\pi/8}_{\bar{x}}$ with $\hat{v}_{Tx}:=\hat{v}^{\pi/8}_{\bar{x}} = \tfrac{\hat{z}+\hat{v}_V}{\abs{\hat{z}+\hat{v}_V}}$ satisfying $R_{\hat{x}}(\tfrac{\pi}{4})=\Pi_{Tx} Z$, and corresponds to $T_x:=\sqrt{V}$. All of these gates are presented in \fig{pi_sphere}, which also demonstrates that $v_H=\tfrac{\hat{v}_S+\hat{v}_V}{\abs{\hat{v}_S+\hat{v}_V}}=\tfrac{\hat{v}_T+\hat{v}_{Tx}}{\abs{\hat{v}_T+\hat{v}_{Tx}}}$, which can be easily verified using vector arithmetic. 
As shown in \cite{zindorf_all_2025}, a Hermitian gate set composed of the $H$ and $\Pi_T$ gates -- implementable as $\pi$-rotations/pulses about $\hat{v}_H$ and $\hat{v}_T$ -- along with the CNOT gate form a Hermitian set which suffices for universal quantum computation.


\begin{figure}[H]
    \centering
    \begin{subfigure}{.49\linewidth}
    \centering
    \includegraphics[width=0.70\linewidth]{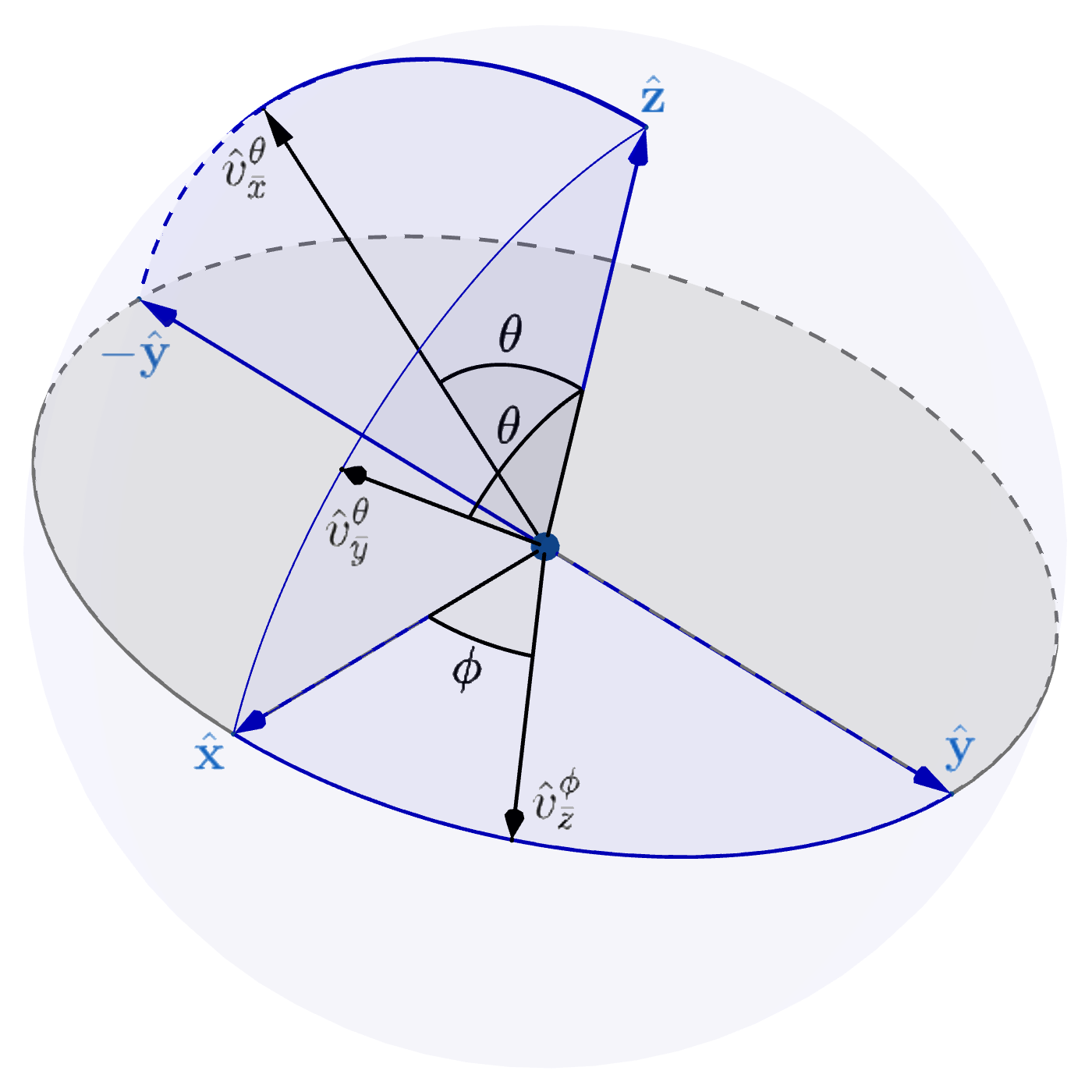}
    \caption{}
    \label{fig:pi_sphere_angles}
    \end{subfigure}
    \begin{subfigure}{.49\linewidth}
    \centering
    \includegraphics[width=0.70\linewidth]{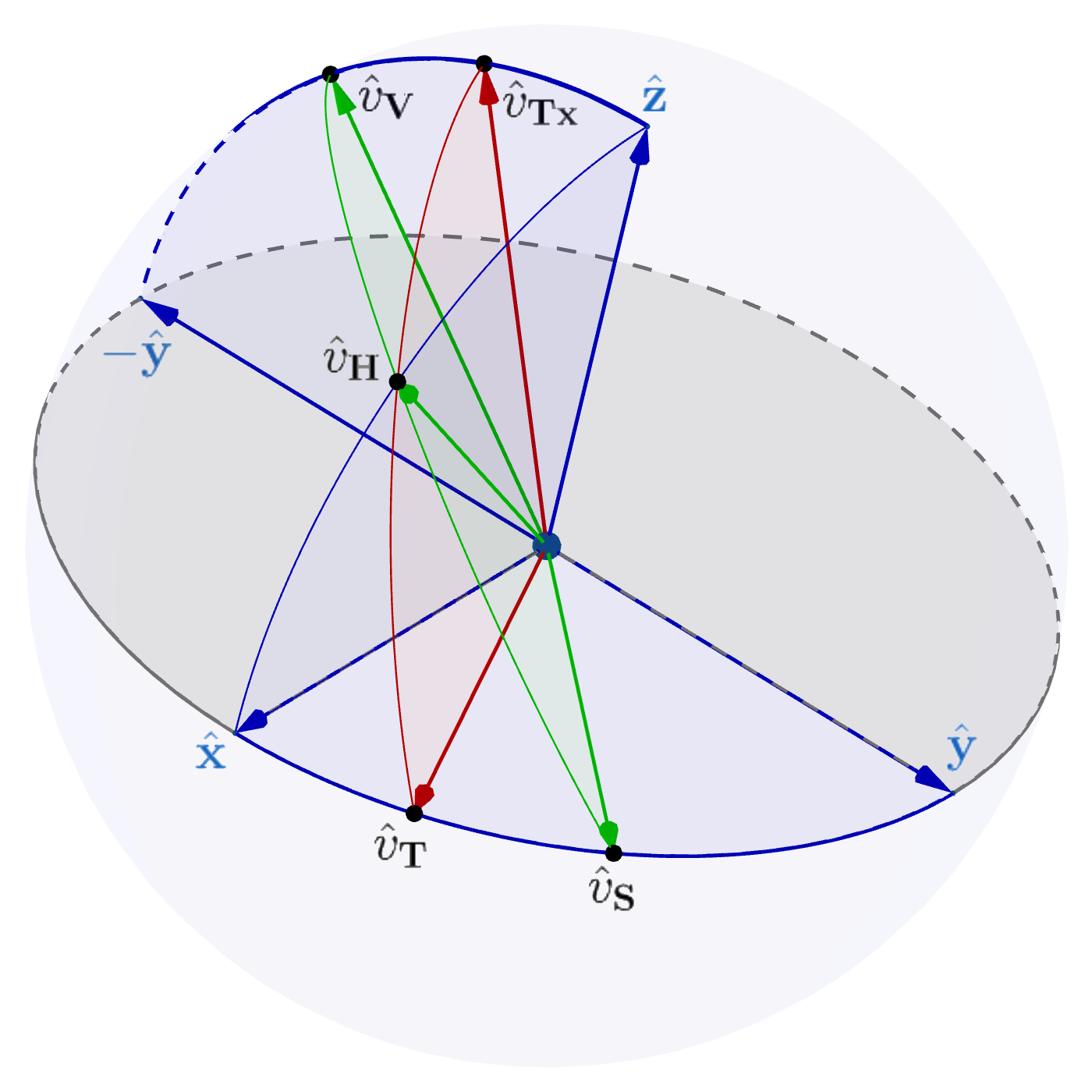}
    \caption{}
    \label{fig:pi_sphere}    
    \end{subfigure}
    \caption{A geometrical description of the axes mentioned in this section. (a) The parameterized axes $\hat{v}^\theta_{\bar{x}},\hat{v}^\theta_{\bar{y}}$ and $\hat{v}^\phi_{\bar{z}}$. (b) The fixed axes $\hat{v}_H,\hat{v}_S,\hat{v}_V,\hat{v}_T$ and $\hat{v}_{Tx}$.} 
    \label{fig:sphere_vecs}    
\end{figure}

The CNOT gate, a.k.a controlled $X$, is a two-qubit gate which applies the Pauli $X$ operator on a target qubit $q_t$ if a control qubit $q_c$ is in the state $\ket{1}$, and applies $I$ for state $\ket{0}$. The control qubit can be replaced by a control set $C=\{q_{c_1},q_{c_2},..,q_{c_n}\}$, holding $n$ qubits, to define a multi-controlled version of this gate which we mark as $\MCO{X}{C}{q_t}$ \cite{zindorf_efficient_2024}, such that if $n=1$, this gate is a CNOT, and if $n=2$, the gate is referred to as the Toffoli. In general, this gate is referred to as the multi-controlled Toffoli, MCT, or MCX and applies the Pauli $X$ on $q_t$ if $C$ is in the state $\ket{11..1}$, and the identity operator $I$ otherwise. 
Similarly, a multi-controlled version of any single qubit operator $O\in U(2)$ can be defined as $\MCO{O}{C}{q_t}$. In case $n=0$, we get a single-qubit gate operating on $q_t$ which we mark as $\sqO{O}{q_t}:=\MCO{O}{\emptyset}{q_t}$. 

If $O$ is a diagonal matrix of the form $P(\phi):=\left(
\begin{matrix}
1 & 0\\
0 & e^{i\phi}
\end{matrix}
\right)$, the phase $e^{i\phi}$ is only applied if the entire controls-target qubit set $\{C,q_t\}:=\{q_{c_1},..,q_{c_n},q_t\}$ is in state $\ket{11..1}$ and therefore it is unnecessary to specify the target qubit out of the controls-target set. In this case we can write $\MCO{P(\phi)}{\{C,q_t\}}{}:=\MCO{P(\phi)}{C}{q_t}$ as a multi-controlled phase gate (MCP). The multi-controlled Pauli Z gate (MCZ) is a special case of an MCP with $\phi=\pi$, and therefore has a unique visual representation. The MCZ, MCX and the general $\MCO{O}{C}{q_t}$ gates are represented as follows.
\[
\scalebox{0.7}{
  \Qcircuit @C=1.0em @R=0.38em @!R {
              \nghost{{C} :  } \\
	 	\nghost{{C} :  } & \lstick{{C} :  } & \qw {/^n}& \ctrl{3} & \qw \\
        \nghost{{C} :  } \\
        \nghost{{C} :  } \\
	 	\nghost{{q_t} :  } & \lstick{{q_t} :  } & \qw  & \control\qw & \qw \\
        \nghost{{C} :  }\\
 }
\hspace{5mm}\raisebox{-12mm}{:=}\hspace{0mm}
  \Qcircuit @C=1.0em @R=0em @!R {
              \nghost{{C} :  } \\
	 	\nghost{{C} :  } & \lstick{{C} :  } & \qw {/^n}& \ctrl{3} & \qw \\
        \nghost{{C} :  } \\
        \nghost{{C} :  } \\
	 	\nghost{{q_t} :  } & \lstick{{q_t} :  } & \qw  & \gate{\mathrm{Z}} & \qw \\
        \nghost{{C} :  } \\
 }
\hspace{5mm}\raisebox{-20mm}{;}\hspace{0mm}
  \Qcircuit @C=1.0em @R=0.38em @!R {
              \nghost{{C} :  } \\
	 	\nghost{{C} :  } & \lstick{{C} :  } & \qw {/^n}& \ctrl{3} & \qw \\
        \nghost{{C} :  } \\
        \nghost{{C} :  } \\
	 	\nghost{{q_t} :  } & \lstick{{q_t} :  } & \qw  & \targ & \qw \\
        \nghost{{C} :  }\\
 }
\hspace{5mm}\raisebox{-12mm}{:=}\hspace{0mm}
  \Qcircuit @C=1.0em @R=0em @!R {
              \nghost{{C} :  } \\
	 	\nghost{{C} :  } & \lstick{{C} :  } & \qw {/^n}& \ctrl{3} & \qw \\
        \nghost{{C} :  } \\
        \nghost{{C} :  } \\
	 	\nghost{{q_t} :  } & \lstick{{q_t} :  } & \qw  & \gate{\mathrm{X}} & \qw \\
        \nghost{{C} :  } \\
 }
\hspace{5mm}\raisebox{-20mm}{;}\hspace{0mm}
  \Qcircuit @C=1.0em @R=0em @!R {
              \nghost{{C} :  } \\
	 	\nghost{{C} :  } & \lstick{{C} :  } & \qw {/^n}& \ctrl{3} & \qw \\
        \nghost{{C} :  } \\
        \nghost{{C} :  } \\
	 	\nghost{{q_t} :  } & \lstick{{q_t} :  } & \qw  & \gate{\mathrm{O}} & \qw \\
        \nghost{{C} :  } \\
 }
\hspace{5mm}\raisebox{-12mm}{:=}\hspace{0mm}
\Qcircuit @C=1.0em @R=0.0em @!R { 
	 	\nghost{q_{{c}_{1}} :  } & \lstick{q_{{c}_{1}} :  } & \ctrl{1} & \qw \\
	 	\nghost{q_{{c}_{2}} :  } & \lstick{q_{{c}_{2}} :  } & \ctrl{1} & \qw \\
        \nghost{{q}_{1} :  } &  & \ar @{.} [1,0]   \\
        \nghost{{q}_{1} :  } &  &  \\
	 	\nghost{q_{{c}_{n}} :  } & \lstick{q_{{c}_{n}} :  } & \ctrl{1}\ar @{-} [-1,0] & \qw \\
	 	\nghost{{q_t} :  } & \lstick{{q_t} :  } & \gate{\mathrm{O}} & \qw\\
   }
   }
\]

These notations are useful when there is no restriction to the qubit connectivity, such as in all-to-all connected architectures. However, when the qubit connectivity is restricted, we must account for qubit ordering. In the case of LNN connectivity, a circuit is applied on an ordered set of qubits $Q=\{q_1,q_2,..,q_k\}$ ('LNN set'), such that two-qubit gates, e.g. CNOT, can only be applied on two nearest-neighboring qubits in a 1D array, i.e., on $\{q_j,q_{j+1}\}$ for $j\in [1,k-1]$. The $n+1$ indices $t,c_1,c_2,..,c_n$ can be freely chosen from the $k$ qubits in the circuit, with $k\binom{k-1}{n}$ possible combinations. As we wish to avoid drawing many circuits to account for every possible combination, we provide the notation in \qc{arrow_notation} to describe a multi-controlled gate, controlled by a set $C$ which is a subset of the LNN set $Q'\in Q\setminus q_t$ of size $k'\leq k-1$. 

\[
    \scalebox{0.7}{
  \Qcircuit @C=1.0em @R=0.em @!R {
  \\
	 	\nghost{{Q'} :  } & \lstick{{Q'} :  } & \qwl& \ctrlp{2}{C} & \qw \\
        \\
	 	\nghost{{q_t} :  } & \lstick{{q_t} :  } & \qw  & \gate{\mathrm{O}} & \qw \\
 }
 \hspace{5mm}\raisebox{-9mm}{:=}\hspace{0mm}
 \Qcircuit @C=0.5cm @R=0.cm @!R {
         \nghost{} & \lstick{} &  &  & & \qw & \nghost{} & \lstick{Q'\setminus C:} & \qw {/^{k'-n}} &  \qw \\
         \nghost{{Q'} :  } & \lstick{{Q'} :  } & \qw {/^{k'}} & \qw \link{-1}{1}\link{1}{1} &  & & & & & & \link{1}{-1}\link{-1}{-1} & \qw {/^{k'}} & \qw & \lstick{Q'} \\
         \nghost{} & \lstick{} & & & &  \qw  &  \lstick{C:} & \qw {/^n} & \ctrl{1}  &  \qw \\
         \nghost{{q_t} :  } & \lstick{{q_t} :  }  & \qw & \qw & \qw & \qw & \qw & \qw  & \gate{\mathrm{O}} & \qw & \qw  & \qw & \qw & \lstick{q_t}
     }
     }
     \qcref{arrow_notation}
\]

For a more general description and clarification, we can define $Q'_1=\{q_1,..,q_{t-1}\}$ and $Q'_2=\{q_{t+1},..,q_{k}\}$ such that $Q'_1 \cup Q'_2 =  Q\setminus q_t$. Similarly, the control set can be split into $C'_1\in Q'_1$ and $C'_2\in Q'_2$ such that $C'_1\cup C'_2 = C$. This allows us to account for any location of the target, as well as the controls. For example, if $Q'_2=\emptyset$, the target is located at the bottom of the circuit, and if $Q'_1,Q'_2\not=\emptyset$, then the target is somewhere in the middle. For a compact description, we use $Q_1=Q'_1\setminus q_{t-1}$, $Q_2=Q'_2\setminus q_{t+1}$, $c_-\in\{q_{t-1}\}$, $c_+\in\{q_{t+1}\}$, $C_1=C'_1\setminus c_-$, and $C_2=C'_2\setminus c_+$ in the following circuits.

\[
    \begin{gathered}
            \scalebox{0.7}{
  \Qcircuit @C=0.5em @R=0.6em @!R {
	 \nghost{{Q'_1} :  } & \lstick{{Q'_1} :  } & \qwl& \ctrlp{2}{C'_1} & \qw \\
     \\
	 	\nghost{{q_t} :  } & \lstick{{q_t} :  } & \qw  & \gate{\mathrm{O}} & \qw \\
        \\
	 \nghost{{Q'_2} :  } & \lstick{{Q'_2} :  } & \qwl& \ctrlpd{-2}{C'_2} & \qw \\
 }
     }
    \end{gathered}
\hspace{4mm}\raisebox{0.mm}{=}\hspace{0mm}
\begin{gathered}
\scalebox{0.7}{
 \Qcircuit @C=0.5em @R=1.3em @!R { 
  \nghost{{Q_1} :  } & \lstick{{Q_1} :  } & \qwl& \ctrlp{1}{C_1} & \qw \\
	 	\nghost{q_{t-1} :} & \lstick{q_{t-1} :} & \qw & \ctrlp{1}{c_-} & \qw \\
	 	\nghost{q_{t} :} & \lstick{q_{t} :} & \qw & \gate{\mathrm{O}}  & \qw \\
            \nghost{q_{t+1} :} & \lstick{q_{t+1} :} & \qw & \ctrlpd{-1}{c_+}  & \qw \\
  \nghost{{Q_2} :  } & \lstick{{Q_2} :  } & \qwl& \ctrlpd{-1}{C_2} & \qw \\
   }
 }
 \end{gathered}
\hspace{2mm}\raisebox{0.mm}{=}\hspace{2mm}
    \begin{tabular}{|c|c|c|}
    \hline
        \text{case} & $c_+\!=\!\{q_{t+1}\}$ & $c_+=\emptyset$ \\
         \hline
        $c_-\!=\!\{q_{t-1}\}$ & $\begin{gathered}
    \scalebox{0.6}{
 \Qcircuit @C=0.5em @R=0.em @!R { 
  \lstick{} & \qwl& \ctrlp{1}{C_1} & \qw \\
	 	\lstick{} & \qw & \ctrl{1} & \qw \\
	 	\lstick{} & \qw & \gate{\mathrm{O}}  & \qw \\
            \lstick{} & \qw & \ctrl{-1}  & \qw \\
  \lstick{} & \qwl& \ctrlpd{-1}{C_2} & \qw \\
   }
 }
    \end{gathered}$   &  $\begin{gathered}
    \scalebox{0.6}{
 \Qcircuit @C=0.5em @R=0.em @!R { 
  \lstick{} & \qwl& \ctrlp{1}{C_1} & \qw \\
	 	\lstick{} & \qw & \ctrl{1} & \qw \\
	 	\lstick{} & \qw & \gate{\mathrm{O}}  & \qw \\
            \lstick{} & \qw & \qw  & \qw \\
 \lstick{} & \qwl& \ctrlpd{-2}{C_2} & \qw \\
   }
 }
    \end{gathered}$\\
    \hline
        $c_-\!=\!\emptyset$ & $\begin{gathered}
    \scalebox{0.6}{
 \Qcircuit @C=0.5em @R=0.em @!R { 
  \lstick{} & \qwl& \ctrlp{2}{C_1} & \qw \\
	 	\lstick{} & \qw & \qw & \qw \\
	 	\lstick{} & \qw & \gate{\mathrm{O}}  & \qw \\
            \lstick{} & \qw & \ctrl{-1}  & \qw \\
  \lstick{} & \qwl& \ctrlpd{-1}{C_2} & \qw \\
   }
 }
    \end{gathered}$  & 
    $\begin{gathered}
    \scalebox{0.6}{
 \Qcircuit @C=0.5em @R=0.em @!R { 
  \lstick{} & \qwl& \ctrlp{2}{C_1} & \qw \\
	 	\lstick{} & \qw & \qw & \qw \\
	 	\lstick{} & \qw & \gate{\mathrm{O}}  & \qw \\
            \lstick{} & \qw & \qw  & \qw \\
 \lstick{} & \qwl& \ctrlpd{-2}{C_2} & \qw \\
   }
 }
    \end{gathered}$\\
    \hline
    \end{tabular}
    \qcref{small_arrow_not}
\]

While we develop our method using the $\pi$-rotation formalism, in many cases the $\Pi$ gates must be decomposed in terms of other available gates. We will provide the final decompositions in terms of Clifford+T and $R_{\hat{z}}$ gates with arbitrary angles. In order to keep track of all gate counts we define the following vectors:
\[
    L_{CX} = (1,0,0,0),
    L_{T} = (0,1,0,0),
    L_{H} = (0,0,1,0),
    L_{R_z} = (0,0,0,1)
    .
\]
Such that, for example, the cost of a CZ gate can be written as $L_{CZ} = L_{CX}+2L_{H} = (1,0,2,0)$, as it requires one CNOT and two Hadamard gates.
This allows to conveniently describe all gate counts using one vector.

{

\section{Multi-Controlled SU(2)} \label{sec:MCSU2_LNN}

In this section, we present a structure which can be used to efficiently implement any multi-controlled $SU(2)$ (MCSU2) gate, without any ancilla requirements, over a set $Q=\{q_1,..,q_k\}$ of $k\geq 2$ LNN connected qubits. 
The MCSU2 gate, $\MCO{R_{\hat{v}}(\lambda)}{C}{q_t}$ with $R_{\hat{v}}(\lambda)\in SU(2)$, is controlled by a qubit set $C\in\ Q\setminus q_t$ of size $n\in [1,k-1]$, such that $q_t$ is the target qubit.
The qubit set $Q$ is defined as the smallest LNN circuit that holds all the qubits on which the MCSU2 gate is applied, and the ones in between, ignoring additional qubits located above and below, i.e. $q_1,q_k\in\{C,q_t\}$.
Our structure is described using the $\Pi$ gate formalism, and will later be decomposed in terms of the standard Clifford+T and $R_{\hat{z}}$ rotations with arbitrary angles.

In order to simplify the problem, we use \lem{MC_transform} (Lemma 4 in\cite{zindorf_all_2025}). At the cost of two $\Pi$ gates, it allows us to focus on a subgroup of all $SU(2)$, holding rotations by arbitrary angles, about a specific axis.

\begin{lemma}\label{lem:MC_transform} 
$\MCO{\omega\rv{\lambda}{\hat{v}}}{C}{q_t}=\sqO{\rpi{\hat{v}_{M}}}{q_t}\MCO{\omega\rv{\lambda}{\hat{v}'}}{C}{q_t}\sqO{\rpi{\hat{v}_{M}}}{q_t}$ with $\omega = e^{i\psi}$
for any angles $\lambda,\psi$, and unit vectors $\hat{v},\hat{v}'$. The vector $\hat{v}_M$ is located in the middle between $\hat{v},\hat{v}'$ such that
$\hat{v}' = \hat{R}_{\hat{v}_M}(\pi)\hat{v}$.
\end{lemma}

Choosing $\hat{v}'=\hat{x}$ and $\psi=0$, we get that any MCSU2 gate $\MCO{\rv{\lambda}{\hat{v}}}{C}{q_t}$ can be implemented as a multi-controlled $R_x$ gate $\MCO{\rx(\lambda)}{C}{q_t}$, in addition to two single-qubit $\pi$-rotations, as \qc{mcsu2_to_rx_new}.
\[
\scalebox{0.7}{
\Qcircuit @C=0.5em @R=0.7em @!R { 
    \nghost{{C} :  } & \lstick{{C} :  } & \qw {/^n} & \ctrl{1} & \qw & \qw\\
    \nghost{{q_t} :  } & \lstick{{q_t} :  } & \qw & \gate{\mathrm{R_{\hat{v}}(\lambda)}} & \qw & \qw\\
}
\hspace{5mm}\raisebox{-5mm}{=}\hspace{0mm}
\Qcircuit @C=0.2em @R=0.7em @!R { 
    \nghost{{C} :  } & \lstick{{C} :  } & \qw {/^n} & \qw & \ctrl{1} & \qw & \qw & \qw\\
    \nghost{{q_t} :  } & \lstick{{q_t} :  } & \gate{\mathrm{\Pi(\hat{v}_M)}} & \qw & \gate{\mathrm{R_{\hat{x}}(\lambda)}} & \qw & \gate{\mathrm{\Pi(\hat{v}_M)}} & \qw\\
}
}
\qcref{mcsu2_to_rx_new}
\]

Now we can focus on implementing multi-controlled $R_x$ gates. 
We wish to avoid making any assumptions regarding the choice of the control qubits, as this may result in a quadratic overhead of SWAP gates \cite{cheng_mapping_2018}. Even assuming a specific location of the target qubit $q_t$ alone may require a communication overhead that scales linearly with $k$ \cite{zindorf_efficient_2024}.
 Using \lem{ccrv_reg} (Lemma 4 in \cite{zindorf_efficient_2024}), we provide the structure in \qc{mcsu2_new} which can be used for {\em any choice} of the target and control qubits. 

\begin{lemma}\label{lem:ccrv_reg}
Any $\MCO{\rv{\lambda}{\hat{v}'}}{C}{q_t}$ gate can be implemented as $\MCO{\rpi{\hat{v}_2}}{C_2}{q_t}
\MCO{\rpi{\hat{v}_1}}{C_1}{q_t}
\MCO{\rpi{\hat{v}_2}}{C_2}{q_t}
\MCO{\rpi{\hat{v}_1}}{C_1}{q_t}$, such that $\rv{\lambda}{\hat{v}'}\in SU(2)$ and $C_1\cup C_2 = C$. $\hat{v}_1$ can be chosen as any unit vector perpendicular to $\hat{v}'$, and $\hat{v}_2 = \hat{R}_{\hat{v}'}(\frac{\lambda}{4})\hat{v}_1$.
\end{lemma}

Since we focus on $\hat{v}'=\hat{x}$, we can choose $\hat{v}_1=\hat{z}$ as the perpendicular vector. For this choice, we get $\Pi(\hat{v}_2)=\Pi^{\lambda/4}_{\bar{x}}$ by definition.
We can add a few relative-phase diagonal gates marked as $\Delta$. These gates commute with any other relative-phase gates, such as the $\MCO{Z}{C_1}{q_t}$ gates in \qc{mcsu2_new}, and with the control lines of any multi-controlled gate. Therefore, and since these gates come in pairs of $\Delta,\Delta^\dagger$, these gates cancel out such that without them \qc{mcsu2_new} simply applies \lem{ccrv_reg}. As this is true for any choice of the diagonal $\Delta$ gates, we are free to choose the relative phase which allows to minimize the gate count.

\[
\scalebox{0.7}{
\Qcircuit @C=0.5em @R=0.7em @!R { 
    \nghost{{Q}_{1} :  } & \lstick{{Q}_{1} :  } & \qwl & \ctrlpd{1}{C_{1}} & \qw \\
    \nghost{{q}_{t-1} :  } & \lstick{{q}_{t-1} :  } & \qw & \ctrlpd{1}{c_-} & \qw \\
    \nghost{{q_t} :  } & \lstick{{q_t} :  } & \qw & \gate{\mathrm{R_{\hat{x}}(\lambda)}} & \qw\\
    \nghost{{q}_{t+1} :  } & \lstick{{q}_{t+1} :  } & \qw & \ctrlp{-1}{c_+} & \qw \\
    \nghost{{Q}_{2} :  } & \lstick{{Q}_{2}:  } & \qwl & \ctrlp{-1}{C_{2}} & \qw \\
}
\hspace{5mm}\raisebox{-17mm}{=}\hspace{0mm}
\Qcircuit @C=0.5em @R=0.45em @!R { 
    \nghost{{Q}_{1} :  } & \lstick{{Q}_{1} :  } & \qwl & \qw & \ctrlpd{2}{C_1\;\;\;\;\;\;} & \multigate{1} {\mathrm{\Delta_1}} & \qw  & \qw & \qw & \multigate{1} {\mathrm{\Delta_1^{\dagger}}} & \qw & \ctrlpd{2}{C_1}  & \qw  & \qw & \qw & \qw  & \qw \\ 
    \nghost{{q}_{t-1} :  } & \lstick{{q}_{t-1} :  } & \qw & \qw & \qw & \ghost{\mathrm{\Delta_1}} & \qw & \ctrlpd{1}{c_-} &\multigate{3}{\mathrm{\Delta_2}}   & \ghost{\mathrm{\Delta_1^{\dagger}}} & \qw & \qw  & \multigate{3}{\mathrm{\Delta_2^\dagger}} & \qw & \ctrlpd{1}{c_-}   & \qw  & \qw \\
    \nghost{{q_t} :  } & \lstick{{q_t} :  } & \qw & \qw  & \gate{\mathrm{Z}} & \qw & \qw & \gate{\mathrm{\Pi^{\lambda\!/\!4}_{\bar{x}}}} & \ghost{\mathrm{\Delta_2}}   & \qw & \qw & \gate{\mathrm{Z}} &  \ghost{\mathrm{\Delta_2}} & \qw & \gate{\mathrm{\Pi^{\lambda\!/\!4}_{\bar{x}}}} & \qw & \qw  \\
    \nghost{{q}_{t+1} :  } & \lstick{{q}_{t+1} :  } & \qw & \qw & \qw & \qw & \qw & \ctrlp{-1}{c_+} & \ghost{\mathrm{\Delta_2}}  & \qw & \qw & \qw &  \ghost{\mathrm{\Delta_2^\dagger}} & \qw  & \ctrlp{-1}{c_+}   & \qw & \qw  \\
        \nghost{{Q}_{2} :  } & \lstick{{Q}_{2} :  } & \qwl & \qw & \qw & \qw & \qw & \ctrlp{-1}{C_2} & \ghost{\mathrm{\Delta_2}}   & \qw & \qw & \qw & \ghost{\mathrm{\Delta_2^\dagger}} & \qw & \ctrlp{-1}{C_2} & \qw  & \qw 
    \gategroup{2}{7}{5}{9}{0.3em}{-} 
    \gategroup{2}{13}{5}{16}{0.3em}{-} 
    \gategroup{1}{10}{3}{12}{0.3em}{--} 
    \gategroup{1}{5}{3}{6}{0.3em}{--} 
}
}
\qcref{mcsu2_new}
\]
Here, $Q_1$ and $Q_2$ hold the first $k_1$ and last $k_2$ qubits in $Q$, respectively, such that $Q_1 \cup Q_2 = Q\setminus \{q_{t-1},q_t,q_{t+1}\}$ and $k_1+k_2=k-1-\abs{\{q_{t-1},q_{t+1}\}\cap Q}$. Similarly, $C_1,C_2\in C$ are of size $n_1,n_2$, respectively, such that $C_1\cup C_2 = C\setminus \{c_+,c_-\}$, and $n_1+n_2 =n-n_+-n_-$, with $n_{\pm}=\abs{c_\pm}\in \{0,1\}$. The value of $\abs{\{q_{t-1},q_{t+1}\}\cap Q}\in \{0,1,2\}$ is determined according to the definition of $Q$, satisfying $q_1,q_k\in\{C_,q_t\}$. For example, if $n_2=n_+=0$ then the target is located below all control qubits, which implies that $k_2=0$ and $q_{t+1}\not\in Q$.

As can be seen, our structure in \qc{mcsu2_new} is composed of four boxed gates. Two of which are MC$\Pi_{\bar{x}}$-$\Delta$ gates, comprised of a multi-controlled $\Pi^{\theta}_{\bar{x}}$ gate (MC$\Pi_{\bar{x}}$) and a $\Delta$ gate. We note that for the choice $\theta\in\{0,\pi,2\pi..\}$, we get $\Pi^{\theta}_{\bar{x}}=\pm Z$ and the entire MC$\Pi_{\bar{x}}$-$\Delta$ gate only applies a relative phase. The distinction between the $\Delta$ and the MC$\Pi_{\bar{x}}$ gate in this case can be maintained if the $\Delta$ gate does not apply on the target qubit. 
We refer to this special case as the MCZ-$\Delta$ gate, a multi-controlled Pauli Z gate (MCZ) and a relative phase gate $\Delta$ which does not apply on the target qubit.
The dashed boxes in \qc{mcsu2_new} hold MCZ-$\Delta$ gates for which the target is located below all control qubits, and the qubit neighboring above the target acts as a dirty quasi-ancilla, such that it is only affected by the $\Delta$ gate.

The MCSU2 implementation therefore relies on the ability to efficiently decompose these boxed gates without any assumption on the number of controls or their location, while for each of these gates, it suffices to find an implementation which can be applied here, i.e. for MCZ-$\Delta$ the target is at the bottom, and for MC$\Pi_{\bar{x}}$-$\Delta$, the target is either at the top or one below it.
As our goal is to decompose the MCSU2 in terms of Clifford+T and $R_z$ gates, we can define a function that holds all gate counts for every possible case.  So far we have the following.
\begin{equation}\label{eq:su2_count}  
L_{SU}(k_1,k_2,n_1,n_2,n_+,n_-)=
\begin{cases}
2L_{Z\Delta}(k_1,n_1) + 2L_{\Pi\Delta}(k_2,n_2,n_+,n_-) + 2L_{\Pi_M} & , n_1+n_2>0\\
L_{SU}(n_+,n_-) & , n_1+n_2=0
\end{cases}
\end{equation}
With $L_{Z\Delta}(\cdot)$ and $L_{\Pi\Delta}(\cdot)$ holding the gate count required for the MCZ-$\Delta$ and MC$\Pi_{\bar{x}}$-$\Delta$ gates, respectively, and $L_{\Pi_M}$ holds the gate count required for the axis-transforming $\Pi$ gates used in \qc{mcsu2_to_rx_new}. We will discuss the implementation of these gates and provide the required gate count in \sec{mcx_g_d}.
We note that in case $n_1=0$, the qubits in $Q$ can be labeled in reverse order. This means that for the case $n_1+n_2>0$, we can always guarantee that the MCZ-$\Delta$ gate has at least one control qubit. 

The $L_{SU}(n_+,n_-)$ holds the gate count of the MCSU2 gates for the case $n_1=n_2=0$. Since $n\geq 1$, we get that $n_++n_-\in\{1,2\}$.
For these cases, we can use the following circuits, such that \qc{CCsu2_mid_new}.a can be realized directly from \lem{ccrv_reg} with added $\Delta$ gates, and \qc{CCsu2_mid_new}.b from \lem{2_rpi}. The gate counts are given in \eq{su2_small_count}.
\[
\scalebox{0.7}{
\Qcircuit @C=0.5em @R=0.7em @!R { 
    \nghost{} & \lstick{} & \qw & \ctrl{1} & \qw \\
    \nghost{} & \lstick{} & \qw & \gate{\mathrm{R_{\hat{x}}(\lambda)}} & \qw\\
    \nghost{} & \lstick{} & \qw & \ctrl{-1} & \qw \\
}
\hspace{5mm}\raisebox{-9mm}{=}\hspace{0mm}
\Qcircuit @C=0.5em @R=0.45em @!R { 
    \nghost{} & \lstick{}   & \ctrl{1} &  \qw  &\qw & \ctrl{1}   &  \qw & \qw   & \qw   \\
    \nghost{} & \lstick{} &  \gate{\mathrm{Z}} &  \gate{\mathrm{\Pi^{\lambda\!/\!4}_{\bar{x}}}} & \multigate{1}{\mathrm{\Delta_2}}  & \gate{\mathrm{Z}} &  \multigate{1}{\mathrm{\Delta_2^\dagger}} & \gate{\mathrm{\Pi^{\lambda\!/\!4}_{\bar{x}}}} & \qw  \\
    \nghost{} & \lstick{} & \qw  & \ctrl{-1} & \ghost{\mathrm{\Delta_2}}  & \qw  &  \ghost{\mathrm{\Delta_2^\dagger}}  & \ctrl{-1}   & \qw
    \gategroup{2}{4}{3}{5}{0.3em}{-} 
    \gategroup{2}{7}{3}{8}{0.3em}{-} 
}
\hspace{5mm}\raisebox{0mm}{(a)}\hspace{0mm}
\hspace{5mm}\raisebox{-17mm}{;}\hspace{0mm}
\Qcircuit @C=0.5em @R=0.7em @!R { 
    \nghost{} & \lstick{} & \qw & \ctrl{1} & \qw \\
    \nghost{} & \lstick{} & \qw & \gate{\mathrm{R_{\hat{x}}(\lambda)}} & \qw\\
}
\hspace{5mm}\raisebox{-5mm}{=}\hspace{0mm}
\Qcircuit @C=0.5em @R=0.45em @!R { 
    \nghost{} & \lstick{}  & \ctrl{1} & \ctrl{1} & \qw \\
    \nghost{} & \lstick{} &  \gate{\mathrm{Z}} &  \gate{\mathrm{\Pi^{\lambda\!/\!2}_{\bar{x}}}} & \qw\\
    }
\hspace{5mm}\raisebox{0mm}{(b)}\hspace{0mm}
}
\qcref{CCsu2_mid_new}
\]
\begin{equation}\label{eq:su2_small_count}  
L_{SU}(n_+,n_-)=
\begin{cases}
L_{CZ}+L_{C\Pi} + 2L_{\Pi_M}& , n_-+n_+=1\\
2L_{CZ}+2L_{\Pi\Delta}(n_-=0,n_+=1) + 2L_{\Pi_M} & , n_-+n_+=2
\end{cases}
\end{equation}
With $L_{CZ}$ and $L_{C\Pi}$ hold the costs to implement the singly-controlled CZ and C$\Pi_{\bar{x}}$ gates exactly (not up-to a relative phase), and $L_{\Pi\Delta}(n_-,n_+) := L_{\Pi\Delta}(k_2=0,n_2=0,n_-,n_+)$ holds the the cost of a small MC$\Pi_{\bar{x}}$-$\Delta$. 

\section{Multi-controlled Toffoli}\label{sec:mcx_sec}

In this section, we implement the MCX gate in LNN connectivity, using one dirty ancilla. Since the MCX gate $\MCO{X}{C}{q_t}$ can be implemented using an MCZ gate and two Hadamard gates as $\sqO{H}{q_t}\MCO{Z}{C}{q_t}\sqO{H}{q_t}$, we will focus on implementing the MCZ gate instead. The target of the MCZ can be arbitrarily chosen from the controls-target qubit set $C'=\{C,q_t\}$ of size $n'=n+1$.

The MCZ gate $\MCO{Z}{C'}{}$ can be implemented as $\MCO{-I}{C'}{q_a}$ \cite{he_decompositions_2017,zindorf_efficient_2024}, a special case of the MCSU2 gate with $n'$ controls. The dirty ancilla qubit $q_a$ can be arbitrarily chosen from the qubit set $Q\setminus C'$, such that the size of $Q$ in this case is $k\geq n+2$. 
As a $SU(2)$ rotation by $2\pi$ about any chosen axis is equivalent to $-I$, we can choose the $\hat{x}$ axis and implement $\MCO{X}{C}{q_t} = \sqO{H}{q_t}\MCO{R_{\hat{x}}(2\pi)}{C'}{q_a}\sqO{H}{q_t}$ as in \qc{mcx_to_mcz_to_2pi}.a.
We can use the structure described in \qc{mcsu2_new} to implement the MC$R_{\hat{x}}$ gate with $\lambda=2\pi$. The MC$\Pi_{\bar{x}}$-$\Delta$ gates used in this structure will apply $\Pi_{\bar{x}}^{\pi/2}$ on the target qubit. This is a Hermitian $\pi$-rotation about the axis $-\hat{y}$, which is equivalent to $-Y=(iZ)X$. Therefore, this MC$\Pi_{\bar{x}}$-$\Delta$ gate is equivalent to MCX-$\Delta$, a relative-phase MCX gate, as shown in \qc{mcx_to_mcz_to_2pi}.b.

\[
\scalebox{0.7}{
\Qcircuit @C=0.5em @R=0.7em @!R { 
    \nghost{{C} :  } & \lstick{{C} :  } & \qw {/^n} & \ctrl{2} & \qw & \qw\\
    \nghost{{q_a} :  } & \lstick{{q_a} :  } & \qw & \qw & \qw & \qw\\
    \nghost{{q_t} :  } & \lstick{{q_t} :  } & \qw & \targ & \qw & \qw\\
}
\hspace{5mm}\raisebox{-6.5mm}{=}\hspace{0mm}
\Qcircuit @C=0.5em @R=0.32em @!R { 
    \nghost{{C} :  } & \lstick{{C} :  } & \qw {/^n} & \ctrl{2} & \qw & \qw\\
    \nghost{{q_a} :  } & \lstick{{q_a} :  } & \qw & \qw & \qw & \qw\\
    \nghost{{q_t} :  } & \lstick{{q_t} :  } & \gate{\mathrm{H}} & \control\qw & \gate{\mathrm{H}} & \qw\\
}
\hspace{5mm}\raisebox{-6.5mm}{=}\hspace{0mm}
\Qcircuit @C=0.5em @R=0.em @!R { 
    \nghost{{C} :  } & \lstick{{C} :  } & \qw {/^n} & \ctrl{1} & \qw & \qw\\
    \nghost{{q_a} :  } & \lstick{{q_a} :  } & \qw & \gate{\mathrm{R_{\hat{x}}(2\pi)}} & \qw & \qw\\
    \nghost{{q_t} :  } & \lstick{{q_t} :  } & \gate{\mathrm{H}} & \ctrl{-1} & \gate{\mathrm{H}} & \qw\\
}
\hspace{5mm}\raisebox{0mm}{(a)}\hspace{0mm}
\hspace{0mm}\raisebox{-12mm}{;}\hspace{0mm}
\Qcircuit @C=0.5em @R=2.3em @!R { 
    \nghost{} & \lstick{} & \qw {/} & \ctrl{1} & \multigate{1} {\mathrm{\Delta'}} & \qw\\
    \nghost{} & \lstick{} & \qw & \targ & \ghost {\mathrm{\Delta'}} & \qw\\
}
\hspace{5mm}\raisebox{-6.5mm}{=}\hspace{3mm}
\Qcircuit @C=0.5em @R=1.35em @!R { 
     \lstick{} & \qw {/} & \ctrl{1} & \multigate{1} {\mathrm{\Delta}} & \qw\\
     \lstick{} & \qw & \gate{\mathrm{\Pi^{\pi\!/\!2}_{\bar{x}}}} & \ghost {\mathrm{\Delta}} & \qw\\
}
\hspace{5mm}\raisebox{0mm}{(b)}\hspace{0mm}
}
\qcref{mcx_to_mcz_to_2pi}
\]

The MCZ gate can be implemented using
\qc{mcx_new}, which is achieved using the identities in \qc{mcsu2_new}, \qc{mcx_to_mcz_to_2pi}.a and \qc{mcx_to_mcz_to_2pi}.b.

\[
\scalebox{0.7}{
\Qcircuit @C=1.0em @R=1.4em @!R { 
    \nghost{{Q}_{1} :  } & \lstick{{Q}_{1} :  } & \qwl & \ctrlpd{1}{C_{1}} & \qw & \qw\\
    \nghost{{q}_{a-1} :  } & \lstick{{q}_{a-1} :  } & \qw & \ctrlpd{1}{c_-} & \qw & \qw\\
    \nghost{{q_a} :  } & \lstick{{q_a} :  } & \qw & \qw & \qw & \qw\\
    \nghost{{q}_{a+1} :  } & \lstick{{q}_{a+1} :  } & \qw & \ctrlp{-1}{c_+} & \qw & \qw\\
    \nghost{{Q}_{2} :  } & \lstick{{Q}_{2}:  } & \qwl & \ctrlp{-1}{C_{2}} & \qw & \qw\\
}
\hspace{5mm}\raisebox{-17mm}{=}\hspace{0mm}
\Qcircuit @C=0.5em @R=1.0em @!R { 
    \nghost{{Q}_{1} :  } & \lstick{{Q}_{1} :  } & \qwl & \qw & \ctrlpd{2}{C_1\;\;\;\;\;\;} & \multigate{1} {\mathrm{\Delta_1}} & \qw  & \qw & \qw & \multigate{1} {\mathrm{\Delta_1^{\dagger}}} & \qw & \ctrlpd{2}{C_1}  & \qw  & \qw & \qw & \qw  & \qw \\ 
    \nghost{{q}_{a-1} :  } & \lstick{{q}_{a-1} :  } & \qw & \qw & \qw & \ghost{\mathrm{\Delta_1}} & \qw & \ctrlpd{1}{c_-} &\multigate{3}{\mathrm{\Delta'_2}} & \ghost{\mathrm{\Delta_1^{\dagger}}}  & \qw & \qw   & \multigate{3}{\mathrm{{\Delta'_2}^\dagger}} & \qw & \ctrlpd{1}{c_-}   & \qw  & \qw \\
    \nghost{{q_a} :  } & \lstick{{q_a} :  } & \qw & \qw  & \gate{\mathrm{Z}} & \qw & \qw & \targ & \ghost{\mathrm{\Delta_2}}   & \qw & \qw & \gate{\mathrm{Z}} &  \ghost{\mathrm{{\Delta'_2}^\dagger}} & \qw & \targ & \qw & \qw  \\
    \nghost{{q}_{a+1} :  } & \lstick{{q}_{a+1} :  } & \qw & \qw & \qw & \qw & \qw & \ctrlp{-1}{c_+} & \ghost{\mathrm{\Delta_2}}  & \qw & \qw & \qw &  \ghost{\mathrm{{\Delta'_2}^\dagger}} & \qw  & \ctrlp{-1}{c_+}   & \qw & \qw  \\
        \nghost{{Q}_{2} :  } & \lstick{{Q}_{2} :  } & \qwl & \qw & \qw & \qw & \qw & \ctrlp{-1}{C_2} & \ghost{\mathrm{\Delta_2}}   & \qw & \qw & \qw & \ghost{\mathrm{{\Delta'_2}^\dagger}} & \qw & \ctrlp{-1}{C_2} & \qw  & \qw 
    \gategroup{2}{7}{5}{9}{0.3em}{-} 
    \gategroup{2}{13}{5}{16}{0.3em}{-} 
    \gategroup{1}{10}{3}{12}{0.3em}{--} 
    \gategroup{1}{5}{3}{6}{0.3em}{--} 
}
}
\qcref{mcx_new}
\]
with $Q_1,Q_2\in Q$ of size 
$k_1,k_2$, respectively, such that $Q_1 \cup Q_2 = Q\setminus \{q_{a-1},q_a,q_{a+1}\}$ and $k_1+k_2=k-1-\abs{\{q_{a-1},q_{a+1}\}\cap Q}$. Similarly, $C_1,C_2\in C'$ are of size $n_1,n_2$, respectively, such that $C_1\cup C_2 = C'\setminus \{c_+,c_-\}$,  and $n_1+n_2 =n'-n_+-n_-$. The value of $\abs{\{q_{a-1},q_{a+1}\}\cap Q}\in \{0,1,2\}$ is determined according to the definition of $Q$, satisfying $q_1,q_k\in\{C',q_a\}$. Here, $k\geq n+2$, and $n\geq1$. The gate count of the MCX gate can be expressed as \eq{MCX_small_count}, with $L_{X\Delta}(\cdot)$ holding the gate count required for the MCX-$\Delta$ gate, which will be discussed in \sec{mcx_g_d}.
\begin{equation}\label{eq:MCX_small_count}
L_{X}(k_1,k_2,n_1,n_2,n_+,n_-)=
\begin{cases}
    2L_{Z\Delta}(k_1,n_1) + 2L_{X\Delta}(k_2,n_2,n_+,n_-) + 2L_H &,n_1+n_2>0\\
    L_{X}(n_+=1,n_-=1) &,n_1+n_2=0
\end{cases}
.
\end{equation}
In the MCX case, as any qubit from the set $Q\setminus C'$ can be chosen as a dirty ancilla, we have a freedom to choose which qubit would be labeled as $q_a$. We could use this fact to check the gate count for each choice of $q_a$, and use the cheapest option; however, this might require a large classical processing time. 
In the task of providing an upper bound for any possible case, we simply choose $q_a$ such that at least one of its nearest neighbors is a qubit in $C'$. We can always guarantee that $n_-+n_+\in\{1,2\}$, and similarly to the MCSU2 case, we also have the freedom to reverse the order of qubits, and therefore guarantee that if $n_1+n_2>0$ then $n_1>0$, which means that we do not need to implement the MCZ-$\Delta$ gate without controls in this case as well. 

In case $n_1=n_2=0$, we get $n_-=n_+=1$, since $n'\geq2$. This case is covered by the known identity \cite{kumar_efficient_2013} in \qc{CNOT_mid_new}, which can also be achieved directly from \qc{CCsu2_mid_new}.a, using the identity in \qc{mcx_to_mcz_to_2pi}.a. The gate count is expressed in terms of $L_{CX}$, the cost of a single LNN CNOT gate, in \eq{MCX_small_small_count}.
\[
\scalebox{0.7}{
\Qcircuit @C=0.5em @R=0.9em @!R { 
    \nghost{} & \lstick{}  & \targ & \qw \\
    \nghost{} & \lstick{}  & \qw & \qw\\
    \nghost{} & \lstick{}  & \ctrl{-2} & \qw \\
}
\hspace{5mm}\raisebox{-7mm}{=}\hspace{0mm}
\Qcircuit @C=0.2em @R=0.2em @!R { 
    \nghost{} & \lstick{} & \gate{\mathrm{H}} & \ctrl{1} & \gate{\mathrm{H}} & \qw \\
    \nghost{} & \lstick{} & \qw & \gate{\mathrm{R_{\hat{x}}(2\pi)}} & \qw & \qw\\
    \nghost{} & \lstick{} & \qw & \ctrl{-1} & \qw & \qw \\
}
\hspace{5mm}\raisebox{-7mm}{=}\hspace{0mm}
\Qcircuit @C=0.5em @R=0.9em @!R { 
    \nghost{} & \lstick{}   & \targ &  \qw   & \targ    & \qw   & \qw   \\
    \nghost{} & \lstick{} &  \ctrl{-1} &  \targ  & \ctrl{-1}  & \targ & \qw  \\
    \nghost{} & \lstick{} & \qw  & \ctrl{-1}   & \qw  &  \ctrl{-1}   & \qw
}
}
\qcref{CNOT_mid_new}
\]

\begin{equation}\label{eq:MCX_small_small_count}
L_X(n_+=1,n_-=1) = 4L_{CX}
\end{equation}

We note that since $X=\Pi(\hat{x})$, the MCX structure can be used directly, with two additional single-qubit $\Pi_M$ gate, to implement {\em any} multi-controlled $\Pi$ (MC$\Pi$) gate using the axis transformation from \lem{MC_transform}. Clearly, for MCZ and MCY, the $\Pi_M$ gates are Clifford operators. For a multi-controlled Hadamard (MCH), which applies $H=\Pi_H$ on the target, the $\Pi_M$ gates are in Clifford+T and require only one T gate each, similarly to MC$\Pi_S$ and MC$\Pi_V$. 

\section{Multi-controlled relative-phase Toffoli (Three Different Ones)}\label{sec:mcx_g_d}
\subsection{MC$\mathbf{\Pi_{\bar{x}}-\Delta}$}\label{sec:mcxpidelta}
In this section, we implement the MC$\Pi_{\bar{x}}$-$\Delta$ gate in LNN connectivity without any ancilla requirements. We focus on the case which is used in \qc{mcsu2_new} for the MCSU2 implementation, i.e. $\sqO{\Delta}{Q\setminus Q_1}\MCO{\Pi_{\bar{x}}^\lambda}{\{C_2,c_-,c_+\}}{q_t}$, with all qubit sets defined as in \sec{MCSU2_LNN}. The definition of these qubit sets implies that we can assume that the location of the target qubit $q_t$ is either at the top of the circuit or neighboring below the top qubit. We cannot, however, make any assumptions regarding the choice of the control qubits and must account for the case without controls as well. 
As shown in \qc{mcx_to_mcz_to_2pi}.b, 
 the MC$\Pi_{\bar{x}}$-$\Delta$ gate is a generalized version of the MCX-$\Delta$ gate. We find this generalization to be useful for circuit optimization, and, in fact,
a version of this gate with three control qubits has been shown to be useful for the construction of the LNN four-controlled Toffoli gate without ancilla in \cite{zindorf_all_2025}.

We wish to use the MC$R_{\hat{x}}$ structure from \qc{mcsu2_new} to implement the MC$\Pi_{\bar{x}}$-$\Delta$ gate. First, according to \lem{2_rpi}, we know that
$\Pi^\theta_{\bar{x}}=ZR^\dagger_{\hat{x}}(2\theta)$. By adding a control set $C''=\{C_2,c_-,c_+\}$ to each of these gates and introducing $\Delta$ gates, we get
$\sqO{\Delta}{\{C'',q_t\}}\MCO{\Pi^\theta_{\bar{x}}}{C''}{q_t}=\sqO{\Delta}{\{C'',q_t\}}\MCO{Z}{C''}{q_t}\MCO{R^\dagger_{\hat{x}}(2\theta)}{C''}{q_t}$.
Setting $\sqO{\Delta'}{\{C'',q_t\}}=\sqO{\Delta}{\{C'',q_t\}}\MCO{Z}{C''}{q_t}$ provides the identity in \qc{mcpid_to_mcrx}. 

\[
\scalebox{0.7}{
\Qcircuit @C=0.5em @R=0.6em @!R { 
    \nghost{{C''} :  } & \lstick{{C''} :  } & \qw {/^{n''}} & \ctrl{1} & \multigate{1} {\mathrm{\Delta}} & \qw\\
    \nghost{{q_t} :  } & \lstick{{q_t} :  } & \qw & \gate{\mathrm{\Pi^{\theta}_{\bar{x}}}} & \ghost {\mathrm{\Delta}} & \qw\\
}
\hspace{5mm}\raisebox{-5mm}{=}\hspace{0mm}
\Qcircuit @C=0.5em @R=0.6em @!R { 
    \nghost{{C''} :  } & \lstick{{C''} :  } & \qw {/^{n''}} & \ctrl{1} & \multigate{1} {\mathrm{\Delta'}} & \qw\\
    \nghost{{q_t} :  } & \lstick{{q_t} :  } & \qw & \gate{\mathrm{R^\dagger_{\hat{x}}(2\theta)}} & \ghost {\mathrm{\Delta'}} & \qw\\
}
}
\qcref{mcpid_to_mcrx}
\]
We can therefore directly use the inverted version of \qc{mcsu2_new}, with a reversed qubit ordering, for this implementation. One MCZ-$\Delta$ gate which only applies a relative phase is removed to achieve \qc{mcpix_delta_new}. The cost of the MC$\Pi_{\bar{x}}$-$\Delta$ gate can therefore be written as \eq{MCpi_d_count}.
\[
\scalebox{0.7}{
\Qcircuit @C=0.5em @R=0.45em @!R { 
	 	\nghost{q_{t-1} :} & \lstick{q_{t-1} :} & \qw & \ctrlpd{1}{c_-}  & \ghost{\mathrm{\Delta}} & \qw \\
	 	\nghost{q_{t} :} & \lstick{q_{t} :} & \qw & \gate{\mathrm{\Pi^{\theta}_{\bar{x}}}}  & \ghost{\mathrm{\Delta}} & \qw \\
        \nghost{q_{t+1} :} & \lstick{q_{t+1} :} & \qw & \ctrlp{-1}{c_+}  & \ghost{\mathrm{\Delta}} & \qw \\
        \nghost{Q_2 :} & \lstick{Q_2 :} & \qwl & \ctrlp{-1}{C_2}  & \multigate{-3}{\mathrm{\Delta}} & \qw \\
 } 
\hspace{5mm}\raisebox{-12mm}{=}\hspace{0mm}
\Qcircuit @C=0.5em @R=0.3em @!R {  
    \nghost{{q}_{t-1} :  } & \lstick{{q}_{t-1} :  } & \qw  & \qw & \ctrlpd{1}{c_-} &\multigate{2}{\mathrm{\Delta_2}}  & \qw & \qw   & \multigate{2}{\mathrm{\Delta_2^\dagger}} & \qw & \ctrlpd{1}{c_-}   & \qw  & \qw \\
    \nghost{{q_t} :  } & \lstick{{q_t} :  } & \qw  & \qw & \gate{\mathrm{\Pi^{\theta\!/\!2}_{\bar{x}}}} & \ghost{\mathrm{\Delta_2}}  & \gate{\mathrm{Z}} & \qw &  \ghost{\mathrm{\Delta_2}} & \qw & \gate{\mathrm{\Pi^{\theta\!/\!2}_{\bar{x}}}} & \qw & \qw  \\
    \nghost{{q}_{t+1} :  } & \lstick{{q}_{t+1} :  } & \qw  & \qw & \ctrlp{-1}{c_+} & \ghost{\mathrm{\Delta_2}}  & \qw & \ghost{\mathrm{\Delta_1^{\dagger}}} &  \ghost{\mathrm{\Delta_2^\dagger}} & \qw  & \ctrlp{-1}{c_+}   & \qw & \qw  \\
    \nghost{{Q}_{2} :  } & \lstick{{Q}_{2} :  } & \qwl  & \qw  & \qw & \qw & \ctrlp{-2}{C_2\;\;} & \multigate{-1} {\mathrm{\Delta_1}} & \qw  & \qw & \qw & \qw  & \qw
    \gategroup{1}{4}{3}{6}{0.3em}{-} 
    \gategroup{1}{9}{3}{12}{0.3em}{-} 
    \gategroup{2}{7}{4}{8}{0.3em}{--} 
}
}
\qcref{mcpix_delta_new}
\]
\begin{equation}\label{eq:MCpi_d_count}
L_{\Pi\Delta}(k_2,n_2,n_+,n_-) = 
\begin{cases}
    L_{Z\Delta}(k_2,n_2) + 2L_{\Pi\Delta,0}(n_+,n_-) & n_2>0 \\
    L_{\Pi\Delta,1}(n_+,n_-) & n_2=0
\end{cases}
\end{equation}
with $L_{\Pi\Delta,1}(n_+,n_-)$ holding the cost of the MC$\Pi_{\bar{x}}$-$\Delta$ gate for every choice of $n_+,n_-$, in case $n_2=0$. We use $L_{\Pi\Delta,0}(n_+,n_-)$ to hold the cost of the boxed MC$\Pi_{\bar{x}}$-$\Delta$ gates  on the right hand side of \qc{mcpix_delta_new} which can always be implemented at a cost of $L_{\Pi\Delta,1}(n_+,n_-)$, however, a lower cost can be achieved in some cases.

We wish to decompose these small gates using $R_{\hat{z}}$ and Clifford+T gates. For an arbitrary choice of the $R_{\hat{z}}$ angle and level of accuracy, the added number of T gates might be very large and therefore it might be substantially the most expensive resource in the fault-tolerant regime. However, in near-term implementations, the CNOT gate is the most expensive. We will therefore provide a few ways to implement the MC$\Pi_{\bar{x}}$-$\Delta$ gates with up-to two controls, allowing for trade-offs between the CNOT and the $R_{\hat{z}}$ gates. 

\tab{ccpix_table} covers the implementations for any option of $n_-,n_+$, providing a few implementations for each case. 
It can be noted that all of these are constructed from a pair of axis-transforming gates located on both sides of a central gate, such that controls are added to either the pair, the central gate, or both. Moreover, the controlled gates are always chosen to be controlled $\Pi$ gates which are generally less expensive than controlled gates with an arbitrary angle of rotation (\cite{zindorf_all_2025}, Theorem 1).
In case \circled{6} of \tab{ccpix_table}, we use the notation in  \qc{mch_circ_circ}.a as a special case of \qc{mcpid_to_mcrx}.
In case \circled{3} we use the notation in \qc{mch_circ_circ}.b for a Hermitian gate (contingent on $\Delta''$ choice) which is equivalent to the two-controlled Hadamard (CCH) up to a relative phase and an additional CNOT gate, allowing for a cheaper implementation (a similar notation will be used shortly, which we give as \qc{mch_circ_circ}.c).
\[
\scalebox{0.7}{
\Qcircuit @C=1.0em @R=0.2em @!R { 
	 	 \lstick{}  & \ctrlt{1}  & \qw \\
	 	 \lstick{}  & \gate{\mathrm{\Pi_{\bar{x}}^{\theta}}}   & \qw \\
            \lstick{} & \ctrl{-1}  & \qw \\
   }
\hspace{5mm}\raisebox{-7mm}{:=}\hspace{0mm}
\Qcircuit @C=0.5em @R=0.05em @!R { 
    \nghost{} & \lstick{} & \qw  & \ctrl{1} & \qw & \qw\\
    \nghost{} & \lstick{} & \qw & \gate{\mathrm{R^\dagger_{\hat{x}}(2\theta)}} & \gate{\mathrm{Z}} & \qw\\
    \nghost{} & \lstick{} & \qw & \ctrl{-1} & \ctrl{-1} & \qw \\
}
 \hspace{5mm}\raisebox{-0mm}{(a)}\hspace{0mm}
 \hspace{0mm}\raisebox{-15mm}{;}\hspace{0mm}
\Qcircuit @C=0.5em @R=0.6em @!R { 
    \nghost{} & \lstick{}  & \ctrltt{1} & \qw\\
    \nghost{} & \lstick{}  & \gate{\mathrm{H}} & \qw\\
    \nghost{} & \lstick{}  & \ctrl{-1} & \qw\\
}
\hspace{5mm}\raisebox{-7mm}{:=}\hspace{0mm}
\Qcircuit @C=0.5em @R=0.6em @!R { 
    \nghost{} & \lstick{} & \qw & \qw  & \ctrl{1} & \multigate{2} {\mathrm{\Delta''}}  & \qw\\
    \nghost{} & \lstick{} & \qw & \targ  & \gate{\mathrm{H}} & \ghost {\mathrm{\Delta''}}   & \qw\\
    \nghost{} & \lstick{} & \qw & \ctrl{-1} & \ctrl{-1} & \ghost {\mathrm{\Delta''}}  & \qw
}
 \hspace{5mm}\raisebox{-0mm}{(b)}\hspace{0mm}
\hspace{0mm}\raisebox{-15mm}{;}\hspace{5mm}
\Qcircuit @C=0.5em @R=0.2em @!R { 
	 	 \lstick{}  & \ctrltt{1}  & \qw \\
	 	 \lstick{}  & \gate{\mathrm{\Pi_{\bar{x}}^{\theta}}}  & \qw \\
            \lstick{}  & \ctrl{-1} & \qw \\
   }
 \hspace{5mm}\raisebox{-7mm}{:=}\hspace{5mm}
  \Qcircuit @C=0.5em @R=0.2em @!R { 
	 	 \lstick{}  & \ctrl{1}  & \qw  & \qw  \\
	 	 \lstick{}  & \gate{\mathrm{i\Pi_{\bar{x}}^{\theta}}} & \targ  & \qw \\
            \lstick{}  & \ctrl{-1}  & \ctrl{-1} & \qw \\
   }    
\hspace{5mm}\raisebox{-0mm}{(c)}\hspace{0mm}
}
\qcref{mch_circ_circ}
\]
\begin{table}[H]
    \centering
\begin{tabular}{|c|c|c|c|}
    \hline
        \text{case} & $1R_{\hat{z}}(2\theta)$ & $2R_{\hat{z}}(\theta)$ & $4R_{\hat{z}}(\theta/2)$\\
         \hline
        $\begin{gathered}
    \scalebox{0.7}{
 \Qcircuit @C=0.2em @R=0.0em @!R { 
    \lstick{} & \gate{\mathrm{\Pi^{\theta}_{\bar{x}}}} & \gate{\mathrm{\Delta}} & \qw\\
}
\hspace{1mm}\raisebox{-0.8mm}{=}\hspace{0mm}
 }
  \vspace{0.6mm}
    \end{gathered}$
    & 
    $\begin{gathered}
    \scalebox{0.7}{
 \Qcircuit @C=0.2em @R=0.0em @!R { 
    \lstick{} & \gate{\mathrm{H}} & \gate{\mathrm{R^\dagger_{\hat{z}}(2\theta)}} & \gate{\mathrm{H}} & \qw\\
}
\hspace{1mm}\raisebox{0mm}{\circled{1}}
 }
 \vspace{0.5mm}
    \end{gathered}$  
    &  
    &
    \\
    \hline
        $\begin{gathered}
    \scalebox{0.7}{
 \Qcircuit @C=0.2em @R=0.0em @!R { 
    \lstick{} & \ctrl{1} & \multigate{1} {\mathrm{\Delta}} & \qw\\
    \lstick{} & \gate{\mathrm{\Pi^{\theta}_{\bar{x}}}} & \ghost {\mathrm{\Delta}} & \qw\\
}
\hspace{1mm}\raisebox{-4mm}{=}\hspace{0mm}
 } \vspace{0.6mm}
    \end{gathered}$ 
        & 
        $\begin{gathered}
    \scalebox{0.7}{
        \Qcircuit @C=0.2em @R=0em @!R { 
    \lstick{} & \ctrl{1} & \qw  & \ctrl{1} & \qw\\
    \lstick{} & \gate{\mathrm{H}} & \gate{\mathrm{R^\dagger_{\hat{z}}(2\theta)}} & \gate{\mathrm{H}} & \qw\\
}
\hspace{1mm}\raisebox{0mm}{\circled{2}}
} \vspace{0.5mm}
    \end{gathered}$  & 
$\begin{gathered}
    \scalebox{0.7}{
        \Qcircuit @C=0.2em @R=0em @!R { 
    \lstick{}  & \qw & \ctrl{1} & \qw & \qw\\
    \lstick{}  & \gate{\mathrm{H}} & \gate{\mathrm{\Pi^{\theta}_{\bar{z}}}} & \gate{\mathrm{H}} & \qw\\
}
\hspace{1mm}\raisebox{0mm}{\circled{4}}
}\vspace{0.6mm}
    \end{gathered}$ 
    &
    \\
    \hline
    $\begin{gathered}
    \scalebox{0.7}{
 \Qcircuit @C=0.2em @R=0.0em @!R { 
    \lstick{}  & \ctrl{1} & \multigate{2} {\mathrm{\Delta}} & \qw\\
    \lstick{} & \gate{\mathrm{\Pi^{\theta}_{\bar{x}}}} & \ghost {\mathrm{\Delta}} & \qw\\
    \lstick{}  & \ctrl{-1} & \ghost {\mathrm{\Delta}} & \qw\\
}
\hspace{1mm}\raisebox{-6.8mm}{=}\hspace{0mm}
 } \vspace{0.7mm}
    \end{gathered}$
    &
    $\begin{gathered}
    \scalebox{0.7}{
        \Qcircuit @C=0.2em @R=0em @!R { 
    \lstick{} & \ctrltt{1} & \qw  & \ctrltt{1} & \qw\\
    \lstick{} & \gate{\mathrm{H}} & \gate{\mathrm{R^\dagger_{\hat{z}}(2\theta)}} & \gate{\mathrm{H}} & \qw\\
    \lstick{} & \ctrl{-1} & \qw  & \ctrl{-1} & \qw\\
}
\hspace{1mm}\raisebox{0mm}{\circled{3}}
} \vspace{0.5mm}
    \end{gathered}$ 
    &
    $\begin{gathered}
    \scalebox{0.7}{
        \Qcircuit @C=0.2em @R=0.15em @!R { 
    \lstick{} & \qw & \ctrl{1} & \qw & \qw\\
    \lstick{} & \gate{\mathrm{H}} & \gate{\mathrm{i\Pi^{\theta}_{\bar{z}}}} & \gate{\mathrm{H}} & \qw\\
    \lstick{} & \qw & \ctrl{-1} & \qw & \qw\\
}
\hspace{1mm}\raisebox{0mm}{\circled{5}}
} \vspace{0.6mm}
    \end{gathered}$ 
    &
    $\begin{gathered}
    \scalebox{0.7}{
        \Qcircuit @C=0.2em @R=0.1em @!R { 
	 	 \lstick{}  & \ctrlt{1}  & \qw \\
	 	 \lstick{}  & \gate{\mathrm{\Pi_{\bar{x}}^{\theta}}}   & \qw \\
            \lstick{} & \ctrl{-1}  & \qw \\
   }
\hspace{1mm}\raisebox{0mm}{\circled{6}}
} \vspace{0.5mm}
    \end{gathered}$ 
    \\
    \hline
    \end{tabular}
    \caption{Implementations of MC$\Pi_{\bar{x}}$-$\Delta$ with up to two controls.}
    \label{tab:ccpix_table}
\end{table}
It can be appreciated that cases \circled{4},\circled{5} are correct by considering that the Hadamard gates transform $\hat{z}\rightarrow\hat{x}$, and therefore $\bar{z}\rightarrow\bar{x}$ as well. Cases \circled{1},\circled{2},\circled{3} utilize the same transformation to apply $R_{\hat{x}}$ rotations, applicable in this case due to the identity in \qc{mcpid_to_mcrx}. In the controlled cases \circled{2},\circled{3}, it can be noted that if the H gates are not applied, then a $R_{\hat{z}}$ rotation is applied on the target, adding a relative phase which can be assigned to the $\Delta$ gate. We show these in more detail in \apx{pi_smalls} --  \lem{MCpi_transform} and \lem{CHH_delta}. In case \circled{3}, given that the CCH variation in \qc{mch_circ_circ}.b is Hermitian (and Unitary so self-inverse), it is clear that the $\Delta''$ gates cancel out. The added CNOT gates are therefore applied on both sides of a CC$R_{\hat{x}}$-$\Delta$ gate, not affecting the CC$R_{\hat{x}}$ part, and changing the $\Delta$ part to apply another relative phase. 

We now provide the Clifford+T and $R_{\hat{z}}$ decompositions of these gates. We can use \lem{h_vsv_svs} to decompose the CH and the CCH variation in cases \circled{2},\circled{3}. The decomposition of the CH gate in \qc{mch_circ_circ_decom}.b is then immediately achieved from \lem{MC_transform} by adding a control. We prove the CCH decomposition in \apx{pi_smalls} -- \lem{dhx_vsv}. Here we simply note that if both controls are ''on'', then the decomposition in \qc{mch_circ_circ_decom}.a applies the Hadamard (with a $-1$ phase) from \lem{h_vsv_svs}, and in case the control of C$\Pi_S$ is ''off'' then $I$ is applied since the C$\Pi_V$ gates cancel out. Finally, if only the control of C$\Pi_S$ is ''on'', considering the effect of $\Pi_{\bar{z}}$ rotations, then effectively a CNOT gate and a relative phase are applied. The C$\Pi_V$ and C$\Pi_S$ are special cases of C$\Pi_{\bar{x}}$ and C$\Pi_{\bar{z}}$, respectively, and can be realized from case \circled{4}, which we address below.
\[
\scalebox{0.7}{
\Qcircuit @C=0.5em @R=0.6em @!R { 
    \nghost{} & \lstick{}  & \ctrltt{1} & \qw\\
    \nghost{} & \lstick{}  & \gate{\mathrm{H}} & \qw\\
    \nghost{} & \lstick{}  & \ctrl{-1} & \qw\\
}
\hspace{5mm}\raisebox{-7mm}{=}\hspace{0mm}
\Qcircuit @C=0.5em @R=0.45em @!R { 
    \nghost{} & \lstick{}   & \ctrl{1} & \qw & \ctrl{1} & \qw\\
    \nghost{} & \lstick{}  & \gate{\mathrm{\Pi_V}} & \gate{\mathrm{\Pi_S}} & \gate{\mathrm{\Pi_V}} & \qw\\
    \nghost{} & \lstick{}  & \qw & \ctrl{-1} & \qw & \qw\\
}
\hspace{5mm}\raisebox{-0mm}{(a)}\hspace{0mm}
\hspace{5mm}\raisebox{-15mm}{;}\hspace{0mm}
\Qcircuit @C=0.5em @R=0.6em @!R { 
    \nghost{} & \lstick{}  & \ctrl{1} & \qw\\
    \nghost{} & \lstick{}  & \gate{\mathrm{H}} & \qw\\
}
\hspace{5mm}\raisebox{-5mm}{=}\hspace{0mm}
\Qcircuit @C=0.5em @R=0.45em @!R { 
    \nghost{} & \lstick{}  & \qw & \ctrl{1} & \qw & \qw\\
    \nghost{} & \lstick{}  & \gate{\mathrm{\Pi_S}} & \gate{\mathrm{-\Pi_V}} & \gate{\mathrm{\Pi_S}} & \qw\\
}
\hspace{5mm}\raisebox{-0mm}{(b)}\hspace{0mm}
}
\qcref{mch_circ_circ_decom}
\]
\begin{lemma}\label{lem:h_vsv_svs}
    $H=\Pi_S(-\Pi_V)\Pi_S=\Pi_V(-\Pi_S)\Pi_V$
\end{lemma}
\begin{proof}
    From \lem{MC_transform}, it suffices to show that $\hat{v}_H = \hat{R}_{\hat{v}_S}(\pi)(-\hat{v}_V)=\hat{R}_{\hat{v}_V}(\pi)(-\hat{v}_S)$, which means that $\hat{v}_S$ is located in the middle between $-\hat{v}_V$ and $\hat{v}_H$, and that $\hat{v}_V$ is located in the middle between $-\hat{v}_S$ and $\hat{v}_H$. This can be shown directly from the definition of these vectors as follows.
    $\hat{v}_S = (\hat{x}+\hat{y})/\sqrt{2}=(\hat{x}+\hat{z})/\sqrt{2}+(\hat{y}-\hat{z})/\sqrt{2} = \hat{v}_H+(-\hat{v}_V)$, and similarly, $\hat{v}_V = (\hat{z}-\hat{y})/\sqrt{2}=(\hat{x}+\hat{z})/\sqrt{2}+(-\hat{y}-\hat{x})/\sqrt{2} = \hat{v}_H+(-\hat{v}_S)$.
\end{proof}

In \qc{ccpi_two_OG}.a we provide two steps for the decomposition of case \circled{6}. The first step is achieved from its definition in \qc{mch_circ_circ}.a, using \qc{CCsu2_mid_new}.a to decompose the inverted CC$R_{\hat{x}}$ gate, removing the $\Delta$ gates and noting that two CZ gates cancel out. The second step is simply achieved by applying Hadamard gates to transform $\bar{z}\rightarrow\bar{x}$ as mentioned above. Finally, the C$\Pi_{\bar{z}}$ gates are decomposed using \lem{Cpi_zz_yy} (Lemma 12 in \cite{zindorf_all_2025}), with $\hat{\tau}=\hat{x}$ and $\hat{\sigma}=\hat{z}$, as \qc{ccpi_two_OG}.b, which provides the decomposition of case \circled{4} as well.
\[
\scalebox{0.7}{
 \Qcircuit @C=1.0em @R=0.37em @!R { 
	 	 \lstick{}  & \ctrlt{1}  & \qw \\
	 	 \lstick{}  & \gate{\mathrm{\Pi_{\bar{x}}^{\theta}}}   & \qw \\
            \lstick{} & \ctrl{-1}  & \qw \\
   }
\hspace{5mm}\raisebox{-8mm}{=}\hspace{0mm}
\Qcircuit @C=0.5em @R=0.22em @!R { 
    \nghost{} & \lstick{} & \qw  & \ctrl{1} & \qw & \ctrl{1} & \qw\\
    \nghost{} & \lstick{} & \qw & \gate{\mathrm{\Pi^{\theta\!/\!2}_{\bar{x}}}} & \gate{\mathrm{Z}} & \gate{\mathrm{\Pi^{\theta\!/\!2}_{\bar{x}}}} & \qw\\
    \nghost{} & \lstick{} & \qw & \qw & \ctrl{-1} & \qw & \qw\\
}
\hspace{5mm}\raisebox{-8mm}{=}\hspace{0mm}
\Qcircuit @C=0.5em @R=0.22em @!R { 
    \nghost{} & \lstick{} & \qw  & \ctrl{1} & \qw & \ctrl{1} & \qw & \qw\\
    \nghost{} & \lstick{} & \gate{\mathrm{H}} & \gate{\mathrm{\Pi^{\theta\!/\!2}_{\bar{z}}}} & \targ & \gate{\mathrm{\Pi^{\theta\!/\!2}_{\bar{z}}}} & \gate{\mathrm{H}} & \qw\\
    \nghost{} & \lstick{} & \qw & \qw & \ctrl{-1} & \qw & \qw & \qw\\
}
 \hspace{5mm}\raisebox{-0mm}{(a)}\hspace{0mm}
 \hspace{5mm}\raisebox{-15mm}{;}\hspace{0mm}
 \Qcircuit @C=0.5em @R=0.4em @!R { 
    \nghost{} & \lstick{} & \ctrl{1} & \qw\\
    \nghost{} & \lstick{}  & \gate{\mathrm{\Pi^{\theta}_{\bar{z}}}} & \qw\\
}
\hspace{5mm}\raisebox{-5mm}{=}\hspace{0mm}
\Qcircuit @C=0.5em @R=0.25em @!R { 
    \nghost{} & \lstick{}  & \qw & \ctrl{1} & \qw & \qw\\
    \nghost{} & \lstick{}  & \gate{\mathrm{R^\dagger_{\hat{z}}(\theta)}} & \targ & \gate{\mathrm{R_{\hat{z}}(\theta)}} & \qw\\
}
 \hspace{5mm}\raisebox{-0mm}{(b)}\hspace{0mm}
 }
 \qcref{ccpi_two_OG}
 \]

\begin{lemma}\label{lem:Cpi_zz_yy}
If $\hat{v}=\hat{R}_{\hat{\sigma}}(\theta)\hat{\tau}$ for any unit vectors $\hat{\tau},\hat{\sigma}$ s.t. $\hat{\tau}\perp \hat{\sigma}$, then $\MCO{\rpi{\hat{v}}}{C}{q_t}=\sqO{\rv{\theta}{\hat{\sigma}}}{q_t}\MCO{\rpi{\hat{\tau}}}{C}{q_t}\sqO{R^\dagger_{\hat{\sigma}}({\theta})}{q_t}$.
\end{lemma}

The CCi$\Pi_{\bar{z}}$ gate in case \circled{5} can be implemented using \lem{Cpi_zz_yy} as \qc{ccipi_neww}.a, noting that the $R_{\hat{z}}$ gates commute with the relative phase $i$.
The resulting CCiX gate, which is equivalent to CC$R^\dagger_{\hat{x}}(\pi)$, can be implemented as \qc{CCsu2_mid_new}.a and can be expressed as case \circled{6} with an additional CZ gate. 
When used to implement the boxed CC$\Pi_{\bar{x}}$ gates in \qc{mcpix_delta_new}, this additional CZ gate can be removed -- reducing the CNOT count by $1$. This is achieved by commuting the CZ gate with other gates and canceling it with its counterpart on the other side. The CZ is first commuted with an H gate, transforming it to a CNOT, then the CNOT is commuted with the MCZ-$\Delta$ gate while adding an MCZ gate with a different target \cite{zindorf_efficient_2024}, thus only changing the relative phase. \qc{ccipi_neww}.b presents this reduced version of case \circled{5} which applies a CCi$\Pi_{\bar{x}}$ gate up-to an additional CNOT gate, which we can use to decompose the boxed gates in \qc{mcpix_delta_new}.
\[
\scalebox{0.7}{
{
 \Qcircuit @C=1.0em @R=0.37em @!R { 
	 	 \lstick{}  & \ctrl{1}  & \qw \\
	 	 \lstick{}  & \gate{\mathrm{i\Pi_{\bar{z}}^{\theta}}}   & \qw \\
            \lstick{} & \ctrl{-1}  & \qw \\
   }
 \hspace{2mm}\raisebox{-8mm}{=}\hspace{2mm}
 \Qcircuit @C=0.2em @R=0.2em @!R { 
	 	 \lstick{} & \qw & \ctrl{1}  & \qw  & \qw  \\
	 	 \lstick{} & \gate{\mathrm{R^\dagger_{\hat{z}}(\theta)}}  & \gate{\mathrm{iX}} & \gate{\mathrm{R_{\hat{z}}(\theta)}}  & \qw \\
            \lstick{} & \qw & \ctrl{-1}  & \qw & \qw \\
   }
   \hspace{2mm}\raisebox{0mm}{(a)}\hspace{0mm}
   \hspace{0mm}\raisebox{-15mm}{;}\hspace{5mm}
 \Qcircuit @C=0.2em @R=0.37em @!R { 
	 	 \lstick{}  & \ctrltt{1}  & \qw \\
	 	 \lstick{}  & \gate{\mathrm{\Pi_{\bar{x}}^{\theta}}}  & \qw \\
            \lstick{}  & \ctrl{-1} & \qw \\
   }
 \hspace{2mm}\raisebox{-8mm}{=}\hspace{2mm}
   \Qcircuit @C=0.2em @R=0.2em @!R{ 
	 	\lstick{} & \qw & \qw & \ctrlt{1} & \qw & \qw  & \qw \\
	 	\lstick{} & \gate{\mathrm{H}} & \gate{\mathrm{R_z^\dagger(\theta)}} & \gate{\mathrm{\Pi_{\bar{x}}^{\pi\!/\!2}}} & \gate{\mathrm{R_z(\theta)}} & \gate{\mathrm{H}}   & \qw \\
        \lstick{} & \qw & \qw & \ctrl{-1} & \qw & \qw & \qw \\
 }
 \hspace{2mm}\raisebox{0mm}{(b)}\hspace{0mm}
}
}
\qcref{ccipi_neww}
\]

 The following summarizes the gate counts of the above implementations, for all options of $L_{\Pi\Delta,\alpha}(n_+,n_-)$ with $\alpha\in\{0,1\}$, as required for \eq{MCpi_d_count}. 
\begin{equation}\label{eq:MCpi_d_small_count}
L_{\Pi\Delta,\alpha}(n_+,n_-)
=
\begin{cases}
    \circled{1} (0,0,2,1) & n_-+ n_+=0\\
    \circled{2} (2,4,4,1)\;\;\text{, or }\circled{4}(1,0,2,2) & n_-+ n_+=1\\
    \circled{3} (6,12,8,1)\text{, or }\circled{5}(3+\alpha,4,4,2)\text{, or } \circled{6}(3,0,2,4) & n_-+ n_+=2 \\
\end{cases}
\end{equation}
{
\subsection{MCX$\mathbf{-\Delta}$}

Here we present our LNN implementation, with no ancilla requirements, for the MCX-$\Delta$ -- a multi-controlled Toffoli gate up-to a relative phase which can be written as $\sqO{\Delta}{Q\setminus Q_1}\MCO{X}{\{C_2,c_-,c_+\}}{q_t}$, with all qubit sets defined as in \sec{MCSU2_LNN}. As shown in \qc{mcx_to_mcz_to_2pi}.b, this is a special case of the MC$\Pi_{\bar{x}}$-$\Delta$ gate, and thus can be implemented using \qc{mcpix_delta_new} with $\theta=\tfrac{\pi}{2}$. 
We address this special case as it is known to be useful on its own as a replacement of the MCX gate, in some cases \cite{maslov_advantages_2016,oonishi_efficient_2022,kuroda_optimization_2022,selinger_quantum_2013}, with no ancilla requirements and a lower Clifford+T gate count. 
 Recalling that $\Pi_V=\Pi^{\pi/4}_{\bar{x}}$, we get the implementation in \qc{mcx_d_new}. This can be used for any choice of the control qubits, given that the target is located at the top or one below the top.
The cost of the MCX-$\Delta$ gate can therefore be written as \eq{MCX_d_count}.
\[
\scalebox{0.7}{
\Qcircuit @C=0.5em @R=0.95em @!R { 
	 	\nghost{q_{t-1} :} & \lstick{q_{t-1} :} & \qw & \ctrlpd{1}{c_-}  & \ghost{\mathrm{\Delta}} & \qw \\
	 	\nghost{q_{t} :} & \lstick{q_{t} :} & \qw & \targ  & \ghost{\mathrm{\Delta}} & \qw \\
        \nghost{q_{t+1} :} & \lstick{q_{t+1} :} & \qw & \ctrlp{-1}{c_+}  & \ghost{\mathrm{\Delta}} & \qw \\
        \nghost{Q_2 :} & \lstick{Q_2 :} & \qwl & \ctrlp{-1}{C_2}  & \multigate{-3}{\mathrm{\Delta}} & \qw \\
 } 
\hspace{5mm}\raisebox{-10mm}{=}\hspace{0mm}
\Qcircuit @C=0.5em @R=0.45em @!R {  
    \nghost{{q}_{t-1} :  } & \lstick{{q}_{t-1} :  } & \qw  & \qw & \ctrlpd{1}{c_-} &\multigate{2}{\mathrm{\Delta_2}}  & \qw & \qw   & \multigate{2}{\mathrm{\Delta_2^\dagger}} & \qw & \ctrlpd{1}{c_-}   & \qw  & \qw \\
    \nghost{{q_t} :  } & \lstick{{q_t} :  } & \qw  & \qw & \gate{\mathrm{\Pi_V}} & \ghost{\mathrm{\Delta_2}}  & \gate{\mathrm{Z}} & \qw &  \ghost{\mathrm{\Delta_2}} & \qw & \gate{\mathrm{\Pi_V}} & \qw & \qw  \\
    \nghost{{q}_{t+1} :  } & \lstick{{q}_{t+1} :  } & \qw  & \qw & \ctrlp{-1}{c_+} & \ghost{\mathrm{\Delta_2}}  & \qw & \ghost{\mathrm{\Delta_1^{\dagger}}} &  \ghost{\mathrm{\Delta_2^\dagger}} & \qw  & \ctrlp{-1}{c_+}   & \qw & \qw  \\
    \nghost{{Q}_{2} :  } & \lstick{{Q}_{2} :  } & \qwl  & \qw  & \qw & \qw & \ctrlp{-2}{C_2\;\;} & \multigate{-1} {\mathrm{\Delta_1}} & \qw  & \qw & \qw & \qw  & \qw
    \gategroup{1}{4}{3}{6}{0.3em}{-} 
    \gategroup{1}{9}{3}{12}{0.3em}{-} 
    \gategroup{2}{7}{4}{8}{0.3em}{--} 
}
}
\qcref{mcx_d_new}
\]
\begin{equation}\label{eq:MCX_d_count}
L_{X\Delta}(k_2,n_2,n_+,n_-) = 
\begin{cases}
    L_{Z\Delta}(k_2,n_2) + 2L_{\Pi_{V}\Delta,0}(n_+,n_-) & n_2>0 \\
    L_{X\Delta}(n_+,n_-) & n_2=0
\end{cases}
\end{equation}
with $L_{X\Delta}(n_+,n_-)$ holding the cost of the MCX-$\Delta$ gate for every choice of $n_+,n_-$, in case $n_2=0$. $L_{\Pi_{V}\Delta,0}(n_+,n_-)$ is analogous to $L_{\Pi\Delta,0}(n_+,n_-)$, such that it holds the cost of the boxed $\Pi_V$ gates in \qc{mcx_d_new}. 
As mentioned, these can be implemented with an reduced CNOT count as \qc{ccipi_neww}. Since the MCX-$\Delta$ implements a specific rotation, arbitrary $R_{\hat{z}}$ gates used in \eq{MCpi_d_small_count} have a fixed angle in this case. We can therefore choose the option which provides the lowest Clifford+T gate count.
 
\tab{ccx_d_table} provides the implementations for both gates for each choice of $n_-,n_+$. When $n_-=n_+=0$, the $\Pi_V$ gate can be replaced with a Hadamard to transform the MCZ-$\Delta$ gate to a MCX-$\Delta$ gate.
In this unique case, the relative phase does not apply on the target qubit which means that the MCX-$\Delta$ gate commutes with $\rx$ gates applied on the target, however, it is only achieved when one dirty quasi-ancilla is available.

\begin{table}[H]
    \centering
\begin{tabular}{|c|c|c|}
    \hline
        \text{case} & 
        $\begin{gathered}
    \scalebox{0.7}{
 \Qcircuit @C=0.0em @R=0.0em @!R { 
    \lstick{} & \gate{\mathrm{O}}
}
\hspace{1mm}\raisebox{-0.8mm}{=}\hspace{0mm}
 \Qcircuit @C=0.0em @R=0.0em @!R { 
    \lstick{} & \target
}
 }
  \vspace{0.6mm}
    \end{gathered}$ 
        & 
        $\begin{gathered}
    \scalebox{0.7}{
 \Qcircuit @C=0.0em @R=0.0em @!R { 
    \lstick{} & \gate{\mathrm{O}}
}
\hspace{1mm}\raisebox{-0.8mm}{=}\hspace{0mm}
 \Qcircuit @C=0.0em @R=0.0em @!R { 
    \lstick{} & \gate{\mathrm{\Pi_V}}
}
 }
  \vspace{0.6mm}
    \end{gathered}$ \\
         \hline
        $\begin{gathered}
    \scalebox{0.7}{
 \Qcircuit @C=0.2em @R=0.0em @!R { 
    \lstick{} & \gate{\mathrm{O}} & \gate{\mathrm{\Delta}} & \qw\\
}
\hspace{1mm}\raisebox{-0.8mm}{$\Rightarrow$}\hspace{0mm}
 }
  \vspace{0.6mm}
    \end{gathered}$
    & 
    $\begin{gathered}
    \scalebox{0.7}{
 \Qcircuit @C=0.2em @R=0.0em @!R { 
    \lstick{}  & \gate{\mathrm{X}}  & \qw\\
}
 }
 \vspace{0.5mm}
    \end{gathered}$  
    & 
    $\begin{gathered}
    \scalebox{0.7}{
 \Qcircuit @C=0.2em @R=0.0em @!R { 
    \lstick{}  & \gate{\mathrm{H}}  & \qw\\
}
 }
 \vspace{0.5mm}
    \end{gathered}$ 
    \\
    \hline
        $\begin{gathered}
    \scalebox{0.7}{
 \Qcircuit @C=0.2em @R=0.0em @!R { 
    \lstick{} & \ctrl{1} & \multigate{1} {\mathrm{\Delta}} & \qw\\
    \lstick{} & \gate{\mathrm{O}} & \ghost {\mathrm{\Delta}} & \qw\\
}
\hspace{1mm}\raisebox{-3mm}{=}\hspace{0mm}
 } \vspace{0.6mm}
    \end{gathered}$ 
        & 
        $\begin{gathered}
    \scalebox{0.7}{
        \Qcircuit @C=0.2em @R=0.65em @!R { 
    \lstick{} & \ctrl{1} & \qw\\
    \lstick{} &  \targ  & \qw\\
}} \vspace{1.4mm}
    \end{gathered}$  & 
$\begin{gathered}
    \scalebox{0.7}{
        \Qcircuit @C=0.2em @R=0em @!R { 
	 	 \lstick{}  & \ctrl{1}  & \qw \\
	 	 \lstick{}  & \gate{\mathrm{\Pi_{V}}}  & \qw \\
   }}\vspace{0.6mm}
    \end{gathered}$ 
    \\
    \hline
    $\begin{gathered}
    \scalebox{0.7}{
 \Qcircuit @C=0.2em @R=0.0em @!R { 
    \lstick{}  & \ctrl{1} & \multigate{2} {\mathrm{\Delta}} & \qw\\
    \lstick{} & \gate{\mathrm{O}} & \ghost {\mathrm{\Delta}} & \qw\\
    \lstick{}  & \ctrl{-1} & \ghost {\mathrm{\Delta}} & \qw\\
}
\hspace{1mm}\raisebox{-5.5mm}{$\Rightarrow$}\hspace{0mm}
 } \vspace{0.7mm}
    \end{gathered}$
    &
    $\begin{gathered}
    \scalebox{0.7}{
        \Qcircuit @C=0.2em @R=0.75em @!R { 
	 	 \lstick{} & \ctrlt{1} & \qw \\
	 	\lstick{} & \targ & \qw \\
	 	 \lstick{} & \ctrl{-1} & \qw \\
 }} \vspace{0.5mm}
    \end{gathered}$ 
    &
    $\begin{gathered}
    \scalebox{0.7}{
        \Qcircuit @C=0.2em @R=0.15em @!R { 
	 	 \lstick{}  & \ctrltt{1}  & \qw \\
	 	 \lstick{}  & \gate{\mathrm{\Pi_{V}}}  & \qw \\
            \lstick{}  & \ctrl{-1} & \qw \\
   }} \vspace{0.6mm}
    \end{gathered}$ 
    \\
    \hline
    \end{tabular}
    \caption{Implementations of MCX-$\Delta$ and MC$\Pi_V$-$\Delta$ with up to two controls. The MC$\Pi_V$-$\Delta$ implementations are not exact, but can be used for the construction in \qc{mcx_d_new}.}
    \label{tab:ccx_d_table}
\end{table}
As these are special cases of the MC$\Pi_{\bar{x}}$-$\Delta$ gates, we already have all of the implementations from \sec{mcxpidelta}. We simply replace the $R_{\hat{z}}(\tfrac{\pi}{4})$ gates with T, as these are equivalent up to a relative phase. This provides the decompositions in \qc{rftoff_ccv_cv} which only require Clifford+T gates. For example, \qc{rftoff_ccv_cv}.a presents the known implementation of the relative phase Toffoli gate from \cite{maslov_advantages_2016,zindorf_efficient_2024}, as a special case of the CC$\Pi_{\bar{x}}$-$\Delta$ gate. We summarize the cost of these gates in \eq{MCX_d_small_count}. 
\[
\scalebox{0.62}{
{
\Qcircuit @C=0.2em @R=1.35em @!R { 
	 	\nghost{} & \lstick{} & \ctrlt{1} & \qw \\
	 	\nghost{} & \lstick{} & \targ & \qw \\
	 	\nghost{} & \lstick{} & \ctrl{-1} & \qw \\
 }
 \hspace{2mm}\raisebox{-8.5mm}{=}\hspace{2mm}
\Qcircuit @C=0.5em @R=0.4em @!R { 
	 	 \lstick{}  & \ctrlt{1}  & \qw \\
	 	 \lstick{}  & \gate{\mathrm{\Pi_{\bar{x}}^{\pi\!/\!2}}}   & \qw \\
            \lstick{} & \ctrl{-1}  & \qw \\
   }
 \hspace{2mm}\raisebox{-8.5mm}{=}\hspace{2mm}
 \Qcircuit @C=0.2em @R=0.82em @!R { 
	 	 \lstick{} & \qw & \qw & \ctrl{1} & \qw & \qw & \qw & \ctrl{1} & \qw & \qw & \qw \\
	 	 \lstick{} & \gate{\mathrm{H}} & \gate{\mathrm{T^\dagger}} & \targ  & \gate{\mathrm{T}} & \targ & \gate{\mathrm{T^\dagger}} & \targ & \gate{\mathrm{T}} & \gate{\mathrm{H}} & \qw \\
	 	 \lstick{} & \qw & \qw & \qw & \qw & \ctrl{-1}  & \qw & \qw & \qw & \qw & \qw \\
 }
 \hspace{5mm}\raisebox{-0mm}{(a)}\hspace{0mm}
 \hspace{0mm}\raisebox{-15mm}{;}\hspace{5mm}
 \Qcircuit @C=0.2em @R=0.8em @!R { 
	 	 \lstick{}  & \ctrltt{1}  & \qw \\
	 	 \lstick{}  & \gate{\mathrm{\Pi_{V}}}  & \qw \\
            \lstick{}  & \ctrl{-1} & \qw \\
   }
 \hspace{2mm}\raisebox{-8.5mm}{=}\hspace{2mm}
  \Qcircuit @C=0.2em @R=0.37em @!R { 
	 	 \lstick{}  & \ctrltt{1}  & \qw \\
	 	 \lstick{}  & \gate{\mathrm{\Pi_{\bar{x}}^{\pi\!/\!4}}}  & \qw \\
            \lstick{}  & \ctrl{-1} & \qw \\
   }   
 \hspace{2mm}\raisebox{-8.5mm}{=}\hspace{2mm}
   \Qcircuit @C=0.2em @R=0.78em @!R{ 
	 	\lstick{} & \qw & \qw & \ctrlt{1} & \qw & \qw  & \qw \\
	 	\lstick{} & \gate{\mathrm{H}} & \gate{\mathrm{T^\dagger}} & \targ & \gate{\mathrm{T}} & \gate{\mathrm{H}}   & \qw \\
        \lstick{} & \qw & \qw & \ctrl{-1} & \qw & \qw & \qw \\
 }
}
\hspace{5mm}\raisebox{-0mm}{(b)}\hspace{0mm}
\hspace{0mm}\raisebox{-15mm}{;}\hspace{5mm}
{
 \Qcircuit @C=0.2em @R=0.8em @!R { 
	 	 \lstick{}  & \ctrl{1}  & \qw \\
	 	 \lstick{}  & \gate{\mathrm{\Pi_{V}}}  & \qw \\
   }
 \hspace{2mm}\raisebox{-5mm}{=}\hspace{2mm}
   \Qcircuit @C=0.2em @R=0.78em @!R{ 
	 	\lstick{} & \qw & \qw & \ctrl{1} & \qw & \qw  & \qw \\
	 	\lstick{} & \gate{\mathrm{H}} & \gate{\mathrm{T^\dagger}} & \targ & \gate{\mathrm{T}} & \gate{\mathrm{H}}   & \qw \\
 }
 \hspace{5mm}\raisebox{-0mm}{(c)}\hspace{0mm}
}
\qcref{rftoff_ccv_cv}
}
\]
\begin{equation}\label{eq:MCX_d_small_count}
L_{X\Delta}(n_+,n_-)
=
\begin{cases}
    (0,0,0,0) & n_-+ n_+=0\\
    (1,0,0,0)& n_-+ n_+=1\\
    (3,4,2,0) & n_-+ n_+=2 \\
\end{cases}
\text{\;\;\;\;\;, and\;\;\;\;}
L_{\Pi_V\Delta,0}(n_+,n_-)
=
\begin{cases}
    (0,0,1,0) & n_-+ n_+=0\\
    (1,2,2,0)& n_-+ n_+=1\\
    (3,6,4,0) & n_-+ n_+=2 \\
\end{cases}
\end{equation}

}

\subsection{MCZ$\mathbf{-\Delta}$}
Finally we present our Clifford+T decomposition for the LNN MCZ-$\Delta$ gate which is used to construct the gates discussed in the previous sections. We can focus on the case in which the MCZ-$\Delta$ is controlled by a qubit set $C_1$ which is located above the target qubit, and simply apply our resulting structure upside-down in case the control set is $C_2$ which is located below the target.

The MCZ-$\Delta$ gate which we use can be written as $\sqO{\Delta}{\{Q_{1},q_{k_1+1}\}} \MCO{Z}{C_1}{q_{k_1+2}}$, such that $Q_1=\{q_1,..,q_{k_1}\}$, the control qubit set $C_1\in Q_1$ is of size $n_1\in[1,k_1]$, the target qubit $q_{k_1+2}$ is unaffected by the relative phase $\Delta$ gate, and the dirty quasi-ancilla qubit $q_{k_1+1}$ is unaffected by the MCZ gate. 

The control qubits can be chosen freely out of $Q_1$ while satisfying $q_1\in C$, due to the definition of $Q$ being the smallest LNN qubit set.
We wish to develop a decomposition of this gate which provides low values for $L_{Z\Delta}(k_1,n_1)$ which holds the Clifford+T gate count of the MCZ-$\Delta$ gate. 

As we would like to build our structure using a recursive formula, we use an index $j\in [3,k_1+2]$ and define 
$ \{Z\}^j_{C_1,Q_1} := \sqO{\Delta^{j}}{Q_1^{j+1}} \MCO{Z}{C_1^j}{q_j}$ with a cost function $L_{\{Z\Delta\}}(Q_1,C_1,j)$, such that $Q_1^{j}:= \{q_1,..,q_{j-2}\}$ and $C^{j}_1:=C_1 \cap Q^{j}_1$ holding all qubits and all {\em control} qubits above $q_{j-1}$ respectively. 
We note that for the choice $j=k_1+2$ we get $\{Z\}^{k_1+2}_{C_1,Q_1} = \sqO{\Delta}{\{Q_{1},q_{k_1+1}\}} \MCO{Z}{C_1}{q_{k_1+2}}$, and therefore $L_{Z\Delta}(k_1,n_1)=L_{\{Z\Delta\}}(Q_1,C_1,k_1+2)$.

Using the inverted version of \qc{mcx_new}, with $n_2=0$ and $n_+=1$, removing one MCZ-$\Delta$ gate as it only applies a relative phase to achieve \qc{mcz_d_new}.

\[
\scalebox{0.7}{
\Qcircuit @C=0.5em @R=1.0em @!R { \\ 
    \nghost{{Q}^{j}_1 :  } & \lstick{{Q}^{j}_1 :  } & \qwl & \ctrlp{2}{C^{j}_1} & \multigate{1} {\mathrm{\Delta^{j}}} & \qw & \qw\\
    \nghost{{q_{j-1}} :  } & \lstick{{q_{j-1}} :  } & \qw & \qw & \ghost {\mathrm{\Delta^{j}}} & \qw\\
    \nghost{{q}_{j} :  } & \lstick{{q}_{j} :  } & \qw & \gate{\mathrm{Z}} & \qw & \qw\\
}
\hspace{5mm}\raisebox{-17mm}{=}\hspace{0mm}
\Qcircuit @C=0.5em @R=1.0em @!R { 
    \nghost{{Q}^{j-1}_1 :  } & \lstick{{Q}^{j-1}_1 :  } & \qwl & \ctrlp{1}{C^{j-1}_1} & \multigate{2} {\mathrm{\Delta^{j}}} & \qw & \qw\\
    \nghost{{q}_{j-2} :  } & \lstick{{q}_{j-2} :  } & \qw & \ctrlpd{2}{c^j_-} & \ghost {\mathrm{\Delta^{j}}} & \qw\\
    \nghost{{q_{j-1}} :  } & \lstick{{q_{j-1}} :  } & \qw & \qw & \ghost {\mathrm{\Delta^{j}}} & \qw\\
    \nghost{{q}_{j} :  } & \lstick{{q}_{j} :  } & \qw & \gate{\mathrm{Z}} & \qw & \qw\\
}
\hspace{5mm}\raisebox{-17mm}{=}\hspace{0mm}
\Qcircuit @C=0.5em @R=1.0em @!R { 
    \nghost{{Q}^{j-1}_1 :  } & \lstick{{Q}^{j-1}_1 :  } & \qwl & \qw  & \qw  & \qw & \qw & \ctrlp{2}{C^{j-1}_1\;\;\;\;\;} & \multigate{1} {\mathrm{\Delta^{j-1}}}   & \qw  & \qw & \qw & \qw  & \qw \\ 
    \nghost{{q}_{j-2} :  } & \lstick{{q}_{j-2} :  } & \qw & \qw  & \qw & \ctrlpd{1}{c^j_-} &\multigate{2}{\mathrm{\Delta'}} & \qw & \ghost{\mathrm{\Delta^{j-1}}}      & \multigate{2}{\mathrm{{\Delta'}^\dagger}} & \qw & \ctrlpd{1}{c^j_-}   & \qw  & \qw \\
    \nghost{{q_{j-1}} :  } & \lstick{{q_{j-1}} :  } & \qw & \qw  & \qw & \targ & \ghost{\mathrm{\Delta'}}  & \gate{\mathrm{Z}} & \qw &  \ghost{\mathrm{{\Delta'}^\dagger}} & \qw & \targ & \qw & \qw  \\
    \nghost{{q}_{j} :  } & \lstick{{q}_{j} :  } & \qw & \qw  & \qw & \ctrl{-1} & \ghost{\mathrm{\Delta'}}  & \qw & \qw &  \ghost{\mathrm{{\Delta'}^\dagger}} & \qw  & \ctrl{-1}   & \qw & \qw 
    \gategroup{2}{5}{4}{7}{0.3em}{-} 
    \gategroup{2}{10}{4}{13}{0.3em}{-} 
    \gategroup{1}{8}{3}{9}{0.3em}{--} 
}
}
\qcref{mcz_d_new}
\]
Here, $c^j_-=C_1^{j}\setminus C_1^{j-1}$ is either empty or holds $q_{j-2}$. \lem{recursive_formula} simply follows, and the recursive formula in \eq{MCZ_d_recursive_count} therefore provides the cost of the MCZ-$\Delta$ gate.
 \begin{lemma} \label{lem:recursive_formula}
$ \{Z\}^{j}_{C_1,Q_1} = \MCO{X}{\{q_{j},c^j_-\}}{q_{j-1}} \{Z\}^{j-1}_{C_1,Q_1} \MCO{X}{\{q_{j},c^j_-\}}{q_{j-1}} $ for $j\in[4,k_1+2]$, and
$c^j_- := 
    \begin{cases}
    \emptyset & q_{j-2} \not\in C_1\\
    q_{j-2}  & q_{j-2} \in C_1
    \end{cases}$.
\end{lemma}
\begin{equation}\label{eq:MCZ_d_recursive_count}
L_{\{Z\Delta\}}(Q_1,C_1,j) = 
\begin{cases}
    L_{\{Z\Delta\}}(Q_1,C_1,j-1) + 2L_{X\Delta}(n_+=1,n_-=1) & j\geq4 \text{, \;\;\;\;\;} q_{j-2}\in C_1 \\
    L_{\{Z\Delta\}}(Q_1,C_1,j-1) + 2L_{X\Delta}(n_+=1,n_-=0) & j\geq4 \text{, \;\;\;\;\;} q_{j-2}\not\in C_1 \\
    L_{\{Z\Delta\}}(Q_1,C_1,3) & j=3 
\end{cases}
\end{equation}
By repeatedly applying \lem{recursive_formula} we achieve \lem{VZ_chain}. The resulting structure can be described by \qc{V_chain_Z}, in which we omit the $\Delta$ gates associated with the MCX-$\Delta$ gates for a compact visualization. We remove the control qubit labels as well, noting that all arrowed controls on a line $j$ would be marked as $c_-^{j+2}$. We also exemplify this structure for a specific choice of control qubits in \qc{V_chain_Z_examples}.  
\begin{lemma}\label{lem:VZ_chain}
$\{ Z \}^{k_1+2}_{C_1,Q_1}=
\left(\prod_{j=4}^{k_1+2}\MCO{X}{\{q_{j},c^j_-\}}{q_{j-1}}\right)^{\dagger}
\{ Z \}^{3}_{C,Q}
\left(\prod_{j=4}^{k_1+2}\MCO{X}{\{q_{j},c^j_-\}}{q_{j-1}}\right)$ with $c^j_- := 
    \begin{cases}
    \emptyset & q_{j-2} \not\in C_1\\
    q_{j-2}  & q_{j-2} \in C_1
    \end{cases}$.
\end{lemma}
\[
\scalebox{0.7}{
\Qcircuit @C=1.0em @R=0em @!R {
	 	\nghost{} & \lstick{} & \ctrl{1} & \multigate{10}{\mathrm{\Delta}} & \qw\\
	 	\nghost{} & \lstick{} & \ctrlp{1}{} & \ghost{\mathrm{\Delta}} & \qw\\
	 	\nghost{} & \lstick{} & \ctrlp{1}{} & \ghost{\mathrm{\Delta}} & \qw\\
	 	\nghost{} & \lstick{} & \ctrlp{1}{} & \ghost{\mathrm{\Delta}} & \qw\\
	 	\nghost{} & \lstick{} & \ctrlp{1}{} & \ghost{\mathrm{\Delta}} & \qw\\
	 	\nghost{} & \lstick{} & \ar @{.} [2,0] &  & \\
	 	\nghost{} & \lstick{} &  &  & \\
            \nghost{} & \lstick{} &  &  & \\
	 	\nghost{} & \lstick{} & \ctrlp{1}{} \ar @{-} [-1,0] & \ghost{\mathrm{\Delta}} & \qw\\
	 	\nghost{} & \lstick{} & \ctrlp{2}{} & \ghost{\mathrm{\Delta}} & \qw\\
	   \nghost{} & \lstick{} & \qw & \ghost{\mathrm{\Delta}} & \qw\\
        \nghost{} & \lstick{} & \gate{\mathrm{Z}} & \qw & \qw\\
 }
 \hspace{5mm}\raisebox{-22mm}{=}\hspace{0mm}
 \Qcircuit @C=0.3em @R=0.38em @!R {
	 	\nghost{} & \lstick{} & \qw & \qw & \qw & \qw & \qw & \qw & \qw & \qw & \qw & \ctrl{2} & \multigate{1}{\mathrm{\Delta^3}} & \qw & \qw & \qw & \qw & \qw & \qw & \qw & \qw & \qw & \qw \\
	 	\nghost{} & \lstick{} & \qw & \qw & \qw & \qw & \qw & \qw & \qw & \ctrlp{1}{{}} & \qw & \qw & \ghost{\mathrm{\Delta^3}} & \qw & \ctrlp{1}{{}} & \qw & \qw & \qw & \qw & \qw & \qw & \qw & \qw\\
	 	\nghost{} & \lstick{} & \qw & \qw & \qw & \qw & \qw & \qw & \ctrlp{1}{{}} & \targ & \qw & \control\qw & \qw & \qw & \targ & \ctrlp{1}{{}} & \qw & \qw & \qw & \qw & \qw & \qw & \qw \\
	 	\nghost{} & \lstick{} & \qw & \qw & \qw & \qw & \qw & \ctrlp{1}{{}} & \targ & \ctrl{-1} & \qw & \qw & \qw & \qw & \ctrl{-1} & \targ & \ctrlp{1}{{}} & \qw & \qw & \qw & \qw & \qw & \qw\\
	 	\nghost{} & \lstick{} & \qw & \qw & \qw & \qw & \ctrlp{1}{{}} & \targ \ar @{-} [1,0]& \ctrl{-1} & \qw & \qw & \qw & \qw & \qw & \qw & \ctrl{-1} & \targ \ar @{-} [1,0]& \ctrlp{1}{{}} & \qw & \qw & \qw & \qw & \qw \\
	 	\nghost{} & \lstick{} &  &  & &  & \nghost{} \ar @{.} [2,-2] &  \nghost{} \ar @{.} [2,-2]&  &  &  &  &  &  &  &  & \nghost{} \ar @{.} [2,2]& \nghost{} \ar @{.} [2,2]&  &  &  &  & \\
	 	\nghost{} & \lstick{} &  &  & &  &  &  &  &  &  &  &  &  &  &  &  &  &  &  &  &  & \\
        \nghost{} & \lstick{} &  &  & \nghost{}& \nghost{} &  &  &  &  &  &  &  &  &  &  &  &  & \nghost{} & \nghost{} &  &  & \\
	 	\nghost{} & \lstick{} & \qw & \ctrlp{1}{{}} & \targ \ar @{-} [-1,0]& \ctrl{-1} & \qw & \qw & \qw & \qw & \qw & \qw & \qw & \qw & \qw & \qw & \qw & \qw & \ctrl{-1} & \targ \ar @{-} [-1,0] & \ctrlp{1}{{}} & \qw & \qw \\
	 	\nghost{} & \lstick{} & \ctrlp{1}{{}} & \targ & \ctrl{-1} & \qw & \qw & \qw & \qw & \qw & \qw & \qw & \qw & \qw & \qw & \qw & \qw & \qw & \qw & \ctrl{-1} & \targ & \ctrlp{1}{{}} & \qw \\
	 	\nghost{} & \lstick{} & \targ & \ctrl{-1} & \qw & \qw & \qw & \qw & \qw & \qw & \qw & \qw & \qw & \qw & \qw & \qw & \qw & \qw & \qw & \qw & \ctrl{-1} & \targ & \qw \\ 
        \nghost{} & \lstick{} & \ctrl{-1} & \qw  & \qw & \qw & \qw & \qw & \qw & \qw & \qw & \qw & \qw & \qw & \qw & \qw & \qw & \qw & \qw & \qw & \qw & \ctrl{-1}  & \qw
        \gategroup{1}{12}{3}{13}{0.3em}{--} 
 }
 }
 \qcref{V_chain_Z}
\]
One notable feature of our structure is that the target of each LNN relative-phase Toffoli gate is located between its controls, thus allowing each Toffoli to be implemented using \qc{rftoff_ccv_cv}.a with 3 CNOT, 4 T and 2 H gates, i.e. with no overhead to the state-of-the-art implementation in unrestricted connectivity \cite{zindorf_efficient_2024,maslov_advantages_2016}. 

Repeatedly applying the recursive rules in \eq{MCZ_d_recursive_count} provides \eq{MCZ_d_count}, noting that the total number of \\$L_{X\Delta}(n_+=1,n_-)$ used is $2(k_1-1)$. Out of these, $2(n_1-1)$ are $L_{X\Delta}(n_+=1,n_-=1)$, since there are $n_1-1$ qubits that satisfy $q_{j-2}\in C_1$ for $j\in [4,k_1+2]$.
\begin{equation}\label{eq:MCZ_d_count}
L_{Z\Delta}(k_1,n_1) = L_{Z\Delta}(1,1)+ 2((n_1-1)L_{X\Delta}(n_+\!=\!1,n_-\!=\!1) + (k_1-n_1)L_{X\Delta}(n_+\!=\!1,n_-\!=\!0))
.
\end{equation}
We already know the cost of $L_{X\Delta}(n_+,n_-)$ from \eq{MCX_d_small_count}, and the cost of $L_{Z\Delta}(1,1)=L_{\{Z\Delta\}}(Q_1,C_1,3)$ is $(3,0,2,0)$ using the following implementation \cite{maslov_depth_2022}.
\[
\scalebox{0.75}{
 \Qcircuit @C=1.0em @R=0.4em @!R { 
	 	\nghost{} & \lstick{}  & \ctrl{2} & \multigate{1}{\mathrm{\Delta^3}} & \qw \\
	 	\nghost{} & \lstick{} & \qw & \ghost{\mathrm{\Delta^3}}  & \qw \\
	 	\nghost{} & \lstick{} &  \control\qw & \qw & \qw \\
 }
   \hspace{5mm}\raisebox{-4mm}{=}\hspace{0mm}
\Qcircuit @C=0.5em @R=0em @!R { 
	 	\nghost{} & \lstick{} & \gate{\mathrm{H}} & \targ & \gate{\mathrm{H}}& \qw \\
	 	\nghost{} & \lstick{} & \targ & \ctrl{-1} & \targ & \qw \\
	 	\nghost{} & \lstick{}  & \ctrl{-1} & \qw &  \ctrl{-1} & \qw \\
 }
  \qcref{cz_cz_d}
}
\]
Finally we reach the gate count required to implement the MCZ-$\Delta$ gate.
\begin{equation}\label{eq:MCZ_d_final_count}
L_{Z\Delta}(k_1,n_1) =
(2k_1+4n_1-3,8n_1-8,4n_1-2,0)
\end{equation}

\section{Results}

Now that we have found the cost of the MCZ-$\Delta$ gate, it can be used to find all other costs. 
The cost of the MCX-$\Delta$ gate is stated in \eq{MCX_d_final_costs}, achieved directly from \eq{MCX_d_count} and \eq{MCZ_d_final_count}.
\begin{equation}\label{eq:MCX_d_final_costs}
L_{X\Delta}(k_2,n_2,n_+,n_-) = 
\begin{cases}
    (2k_2+4n_2-3+2(n_\pm)_3,8n_2-8+4(n_\pm)_3,4n_2+2(n_\pm)_3,0) & n_2>0 \\
    ((n_\pm)_3,4(n_\pm)_1,2(n_\pm)_1,0) & n_2=0
\end{cases}
\end{equation}
using the new convenient notations $(n_\pm)_1:=n_+n_-$ and $(n_\pm)_3:=(n_-+n_++(n_\pm)_1)$ satisfying $0\leq(n_\pm)_j\leq j$, and
noting that the costs listed in \eq{MCX_d_small_count} can be written as
\[
L_{X\Delta}(n_+,n_-)=
((n_\pm)_3,4(n_\pm)_1,2(n_\pm)_1,0) \text{, and }
L_{\Pi_V\Delta,0}(n_+,n_-)=
((n_\pm)_3,2(n_\pm)_3,1+(n_\pm)_3,0)
.
\]

Next, the cost of MCX can be achieved using \eq{MCX_small_count}, \eq{MCZ_d_final_count}, \eq{MCX_small_small_count} and \eq{MCX_d_final_costs}, recalling that in the MCX case, $n_1+n_2 =(n+1-n_+-n_-)$, and $k_1+k_2$
equals $(k-3)$ if $n_1,n_2>0$, and $(k-2-n_+)$ if $n_1>0,n_2=0$ as stated in \sec{mcx_sec}. The cost is stated in \eq{MCX_final_count_n1n2k1k2} using the notations $(\overline{n_\pm})_1 = (n_++n_--(n_\pm)_1)$
$(\overline{n_\pm})_3=(2(n_++n_-)-(n_\pm)_1)$,
and $(\overline{n_\pm})_5=(3(n_++n_-)-(n_\pm)_1)$.
\begin{equation}\label{eq:MCX_final_count_n1n2k1k2}
\begin{aligned}
L_{X}&(k_1,k_2,n_1,n_2,n_+,n_-)=\\
&\begin{cases}
    (
    {\color{purple}4k\!+\!8n}\!-\!16\!-\!4(\overline{n_\pm})_1,   
    {\color{purple}16n}\!-\!16\!-\!8(\overline{n_\pm})_1,
    {\color{purple}8n}\!+\!8\!-\!4(\overline{n_\pm})_1,
    0
    )
    &n_1>0,n_2>0\\      
    ({\color{purple}4k\!+\!8n}\!-\!6\!-\!2(\overline{n_\pm})_5\!-\!4n_+,
    {\color{purple}16n}\!-\!8(\overline{n_\pm})_3,
    {\color{purple}8n}\!+\!8\!-\!4(\overline{n_\pm})_3,
    0)    
    &n_1>0,n_2=0\\    
    (4,0,0,0) &n_1=0,n_2=0
\end{cases}
\end{aligned}
\end{equation}

As can be seen, the CNOT gate count of the MCX gate scales as $4k+8n+O(1)$. For the simple case $n_1=n_2=0$, $n=1$ and $k=3$, an upper bound of less than $4$ CNOT gates is not known, and therefore we cannot upper bound all cases in less than $4k+8n-16$. Unfortunately, if $n_1>0,n_2=0$, we get a higher cost, unless we can guarantee that $n_+=1$.  
As mentioned in \sec{mcx_sec}, the dirty ancilla qubit $q_a$ can be freely chosen from the set $Q\setminus \{C,q_t\}$. Clearly, it is always possible to choose $q_a$ so that it neighbors a qubit from $\{C,q_t\}$, which provides $n_-+n_+\in\{1,2\}$. Therefore, we can always guarantee that $(\overline{n_\pm})_1=1$, $(\overline{n_\pm})_3\in\{2,3\}$, and $(\overline{n_\pm})_5\in\{3,5\}$. This brings us closer to the ideal value, yet one more step is required.

To reach the lowest possible constant term, we can set another rule for the choice of $q_a$ - only choose the bottom qubit as the dirty ancilla if it is the only option. 
This simple rule allows us to make a useful guarantee for the case in which the ancilla qubit is at the bottom, that is, in which $n_1>0,n_2=0$ and $n_+=0$. We know that if this case occurs, it means that $n_-=1$, and $n_1=n=k_1=k-2$, as all qubits above $q_a$ must be in the controls-target set. Now, if $n=1$, we get the trivial case of a single CNOT gate connecting two neighboring qubits. Since we deal with circuits over $k\geq3$, the CNOT gate count of $1$ can be upper bounded by $4k+8n-19$.
For $n\geq 2$, if the bottom qubit is swapped with the one above it, the resulting cost is achieved by \eq{MCX_final_count_n1n2k1k2} with $n_2=0, n_1>0$ and $n_-=n_+=1$. In addition, two SWAP gates are required to swap the qubits; however, these only increase the CNOT count by $4$, as a pair of CNOT gates commute with the MCX gate and cancel out. We reach the total CNOT cost of $4k+8n-16$ for this case, which is the largest gate count, considering all possible options. We note that this can be further reduced using methods from \apx{const_reduc}, however, such reductions will not improve the upper bound, unless a CNOT connecting two next-nearest-neighbors ($k=3,n=1$) can be implemented using less than $4$ CNOT gates. 
Therefore, taking the worst case for each gate count separately, we get the following upper bound which we report in \tab{LNN_costs}.
\[
L_X(k,n)\leq (4k+8n-16,16n-16,8n+4)
\]
Now for the parameterized gates, we start with the MC$\Pi_{\bar{x}}$-$\Delta$. We prioratize achieveing the lowest count of arbitrary $R_{\hat{z}}$ gates, so we use the implementations from \eq{MCpi_d_small_count} that require only one such gate (cases \circled{1},\circled{2},\circled{3}), noting that trade-offs are available for other choices. 
For this choice, we can write
\[
L_{\Pi\Delta,\alpha}(n_+,n_-) = (2(n_\pm)_3,4(n_\pm)_3,2(1+(n_\pm)_3),1)
\]
The gate count can be computed for any case using \eq{MCpi_d_count} and \eq{MCZ_d_final_count} to achieve \eq{MCpi_d_final_count}. 
\begin{equation}\label{eq:MCpi_d_final_count}
L_{\Pi\Delta}(k_2,n_2,n_+,n_-) = 
\begin{cases}
    (2k_2+4n_2-3+4(n_\pm)_3,8n_2-8+8(n_\pm)_3,4n_2+2+4(n_\pm)_3,2)
    & n_2>0 \\
    (2(n_\pm)_3,4(n_\pm)_3,2(1+(n_\pm)_3),1) & n_2=0
\end{cases}
\end{equation}

Next, the MCSU2 cost is achieved from \eq{su2_count}, \eq{MCZ_d_final_count} and \eq{MCpi_d_final_count}. The cost of the axis-transforming $\Pi_{M}$ gates can be written as $L_{\Pi_M}=(0,0,2,2)$ when used for our purposes, as these require one $R_{\hat{z}}$ and one $R_{\hat{x}}$ gate,
as shown in [\cite{zindorf_efficient_2024} Lemma 5]. This is achieved by decomposing the $\Pi_M$ gates using the $XZX$ Euler angles (and inverse, since $\Pi_M=\Pi_M^\dagger$) and canceling two $R_{\hat{x}}$ gates that commute with the target of the MC$R_{\hat{x}}$.
 As mentioned in \sec{mcx_sec}, for the MCSU2 case we can write $n_1+n_2 =(n-n_+-n_-)$, and $k_1+k_2$
equals $(k-3)$ if $n_1,n_2>0$, and $(k-2-n_+)$ if $n_1>0,n_2=0$. \eq{su2_final_count} follows.

\begin{equation}\label{eq:su2_final_count} 
\begin{aligned}
L_{SU}&(k_1,k_2,n_1,n_2,n_+,n_-)=\\
&\begin{cases}
(
{\color{purple}4k+8n}-24+8(n_\pm)_1,
{\color{purple}16n}-32+16(n_\pm)_1,
{\color{purple}8n}+4+8(n_\pm)_1,
8
) 
& , n_1>0,n_2>0\\
(
{\color{purple}4k+8n}-14-4((\overline{n_\pm})_1+n_+),
{\color{purple}16n}-16-8(\overline{n_\pm})_1,
{\color{purple}8n}-4(\overline{n_\pm})_1,
6
)
& , n_1>0,n_2=0\\
L_{SU}(n_+,n_-) & , n_1=0,n_2=0
\end{cases}
\end{aligned}
\end{equation}

The case $n_1=0,n_2=0$ is given by \eq{su2_small_count} as follows, noting that since $L_{C\Pi}$ holds the gate count of a C$\Pi_{\bar{x}}$ (without a relative phase), it must be implemented exactly as case \circled{4} in \eq{MCpi_d_small_count}.
\[
L_{SU}(n_+,n_-)=
\begin{cases}
(2,0,8,6) & , n_-+n_+=1\\
(6,8,14,6) & , n_-+n_+=2
\end{cases}
.
\]
We get the following worst case which is reported in \tab{LNN_costs}.
\[
L_{SU}(k,n)\leq(
4k-8n-14,16n-16,8n+12,8
)
\]
In case $n_2=0$, we find that the number of $R_{\hat{z}}$ gates is reduced to 6, without increasing the other overheads. This means that in ATA connectivity, which is agnostic to qubit ordering, an MCSU2 can always be implemented using only 6 $R_{\hat{z}}$ gates without ancilla, with the other costs easily obtained by setting $k=n+1$. The resulting Clifford+T cost scales similarly to the best known decomposition that uses $8$ $R_{\hat{z}}$ gates \cite{zindorf_efficient_2024}, as reported in \tab{LNN_costs}. In LNN, since the reduction from $8$ to $6$ is only achieved when $n_2=0$, in order to guarantee this count, we must assume that the target qubit is located at the bottom of the circuit. This introduces a CNOT overhead of $\sim 2k$ in the worst case, if the target is located elsewhere, as covered in \cite{zindorf_efficient_2024}.

\begin{table}[H]
    \centering
    \resizebox{0.9\linewidth}{!}{
    \begin{tabular}{|c|c|c|c|c|c|c|}
     \hline
    \multicolumn{7}{|c|}{LNN} \\
 \hline
 \multicolumn{3}{|c|}{Gate} & \multicolumn{4}{c|}{Cost}  \\
  \hline
   Type & Ancilla & Source & CNOT & T & H & $R_{\hat{z}}$\\
 \hline
         MCX  & One dirty  & \cellcolor{lightgray!60} \cite{cheng_mapping_2018,chakrabarti_nearest_2007,miller_elementary_2011,li_quantum_2023,tan_multi-strategy_2018}       & \cellcolor{lightgray!60}$O(nk)$    &  \cellcolor{lightgray!60}$O(n)$  &  \cellcolor{lightgray!60}$O(n)$ & \cellcolor{lightgray!60}$0$    \\
        & ($k\geq n\!+\!2$)       & \cellcolor{lightgray!20} \cite{zindorf_efficient_2024} $(n\geq5)$ & \cellcolor{lightgray!20}$8k+14n-34$      &\cellcolor{lightgray!20}$16n-16$& \cellcolor{lightgray!20}$8n+4$ & \cellcolor{lightgray!20}$0$ \\ 
         & & Current $(n\geq1)$ & $4k+8n-16$ & $16n-16$ & $8n+4$ & $0$\\
 \hline
            MCSU2  & None  & \cellcolor{lightgray!60} \cite{cheng_mapping_2018,chakrabarti_nearest_2007,miller_elementary_2011,li_quantum_2023,tan_multi-strategy_2018}       & \cellcolor{lightgray!60}$O(nk)$      & \cellcolor{lightgray!60}$O(n)$ &  \cellcolor{lightgray!60}$O(n)$ & \cellcolor{lightgray!60}$O(1)$ \\
        &  ($k\geq n\!+\!1$)   &  \cellcolor{lightgray!20} \cite{zindorf_efficient_2024} $(n\geq6)$      & \cellcolor{lightgray!20}$10k+12n-50$    &  
 \cellcolor{lightgray!20}$16n-32$  & \cellcolor{lightgray!20}$8n+2$& \cellcolor{lightgray!20}$8$ \\
         & & Current $(n\geq1)$ & $4k+8n-14$ & $16n-16$ & $8n+12$ & 8\\
 \hline
  \hline
         \multicolumn{7}{|c|}{LNN --  $R_{\hat{z}}$/CNOT tradeoffs}\\
 \hline
             MCSU2  & None & Current& $6k+8n+O(1)$ & $16n+O(1)$ & $8n+O(1)$ & 6\\
    \hline
    MCSU2  & One dirty & Current& $6k+10n+O(1)$ & $16n+O(1)$ & $
 8n+O(1)$ & 5\\
    \hline
        \hline
         \multicolumn{7}{|c|}{ATA -- $R_{\hat{z}}$ reduction}\\
 \hline
            MCSU2 & None &\cellcolor{lightgray!20}\cite{zindorf_efficient_2024}      & \cellcolor{lightgray!20}$12n+O(1)$    &  
 \cellcolor{lightgray!20}$16n+O(1)$  & \cellcolor{lightgray!20}$8n+O(1)$& \cellcolor{lightgray!20}$8$ \\
              &  & Current & $12n+O(1)$ & $ 16n+O(1)$ & $  8n+O(1)$ & 6\\
    \hline
    MCSU2 & One dirty &\cellcolor{lightgray!20}\cite{zindorf_efficient_2024}       & \cellcolor{lightgray!20}$12n+O(1)$    &  
 \cellcolor{lightgray!20}$16n+O(1)$  & \cellcolor{lightgray!20}$8n+O(1)$& \cellcolor{lightgray!20}$8$ \\
      &  & Current  & $12n+O(1)$ & $16n+O(1)$ & $ 
 8n+O(1)$ & 5\\
    \hline
    \end{tabular}
    }
    \caption{A summary of our LNN upper bound results compared to state of the art. Reductions of $R_{\hat{z}}$ gates are provided in LNN with the maximal CNOT overhead, and ATA costs are provided as well.}
    \label{tab:LNN_costs}
\end{table}
In case $n_2=0$, and $n_+=n_-=0$, i.e., the target is at the bottom, and the qubit above it is a dirty ancilla, we can reduce the $R_{\hat{z}}$ count to $5$ without increasing any other cost, however, with a larger overhead in case the target and ancilla are placed elsewhere. As we show in \apx{const_reduc}, the CNOT overhead is no larger than $\sim 2k+2n$. When applied in ATA, no overhead is required.

In \apx{MCZd_reduce_depth} we present depth reductions which can be achieved in any case, allowing to upper bound the depth of our circuits as $(4k+6n,8n,6n,0)+O(1)$. The technique also allows to maximize the use of dirty ancilla qubits in order to reduce the gate count. As we show in \apx{useful_ancill}, in case the qubit ordering is favorable, our method provides a cost of $(4k+4n,8n,4n,0)+O(1)$ if $k\geq 2n+O(1)$, or $(12n,24n-8k,12n-4k,0)+O(1)$ for smaller values of $k$.

We have created a software which implements our methods and provides the decomposition for all mentioned gates. It applies the above reductions when applicable, along with ones which provide reductions by a constant term, as we discuss in \apx{const_reduc}. 
The graphs in \fig{avg_cnot_costs} present our CNOT cost upper bound and the average CNOT count of the circuits produced by our software for the implementation of LNN MCX gates. As our lower bound is missing a constant term, we compare our results to simpler gates which scales similarly - the gate count required to implement an MCX gate in ATA connectivity, when only one dirty ancilla qubit in available, and a single long-range CNOT gate implemented over $k+n$ qubits. Moreover, we note that our upper bound is smaller than the cost of a single long-range CNOT gate implemented over $k+2n$ qubits.

\begin{figure}[H]
    \centering
    \begin{subfigure}{.49\linewidth}
    \centering
    \includegraphics[width=0.99\linewidth]{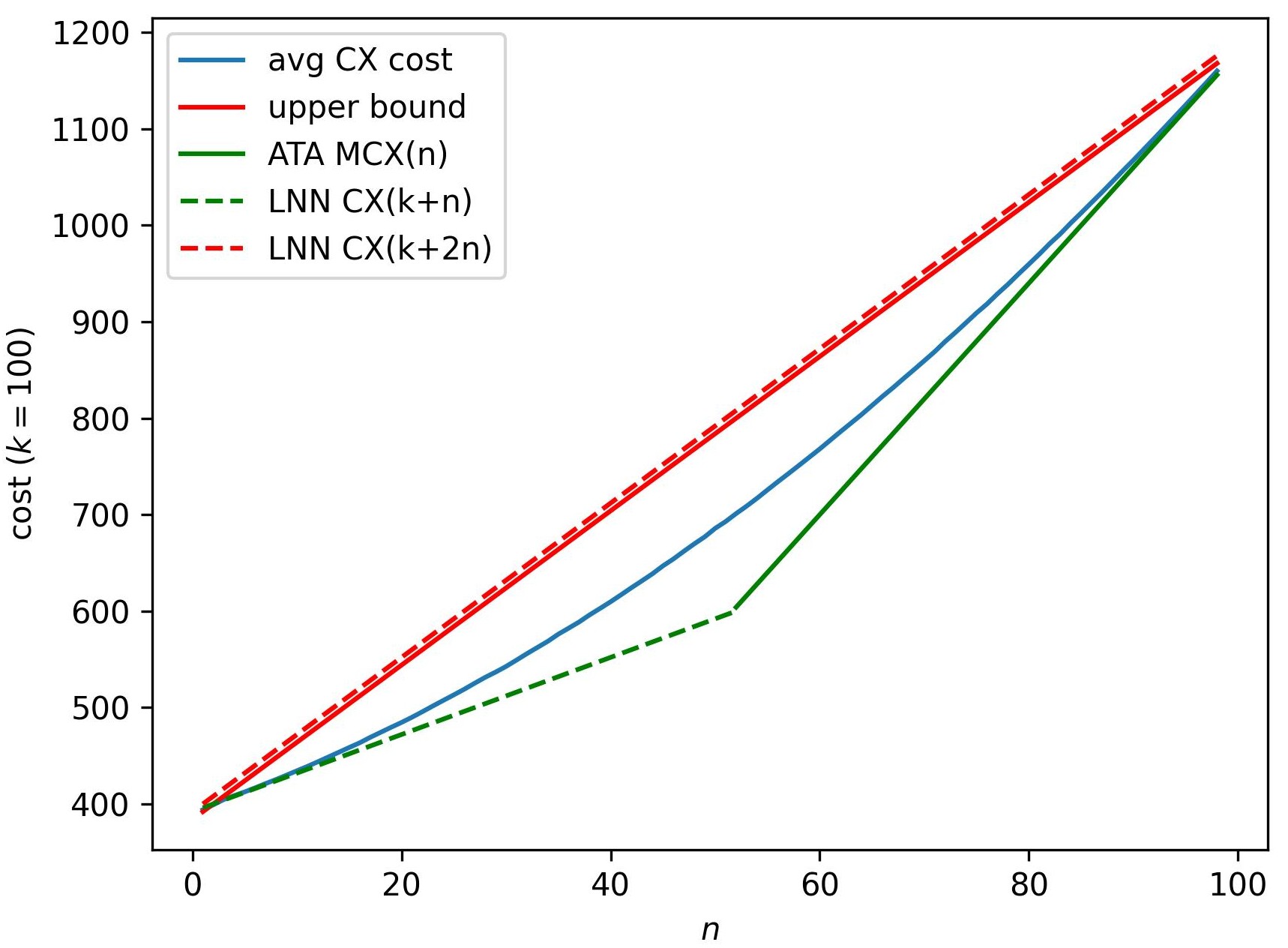}
    \caption{}
    \label{fig:avg_cnot_cost_n}
    \end{subfigure}
    \begin{subfigure}{.49\linewidth}
    \centering
    \includegraphics[width=0.99\linewidth]{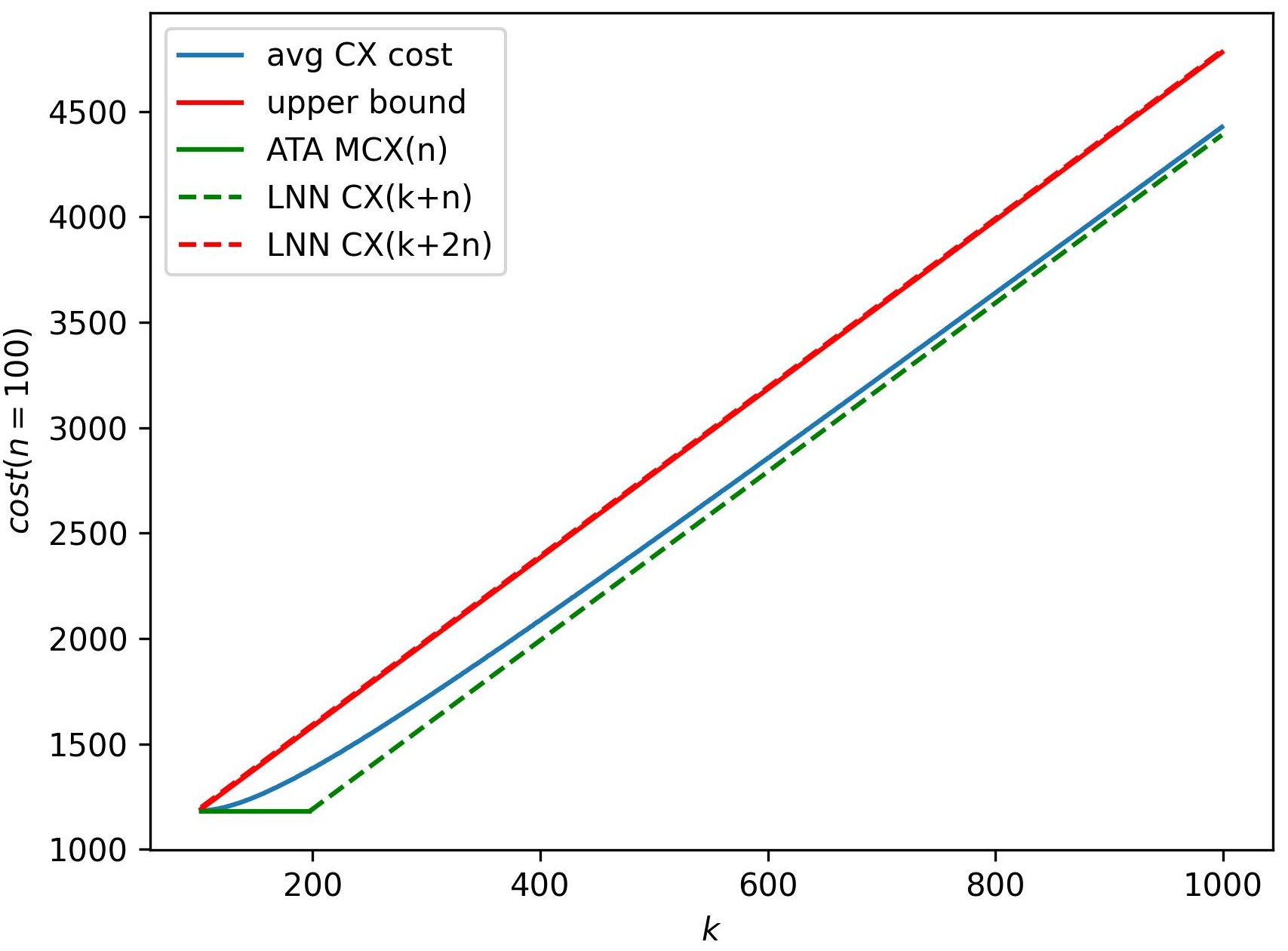}
    \caption{}
    \label{fig:avg_cnot_cost_k}    
    \end{subfigure}
    \caption{Our upper bound and the average CNOT count of the LNN MCX gate with $n$ controls over $k$ qubits produced by our software, averaged over $10^3$ random choices of the locations of the target, dirty ancilla, and control qubits. 
    We plot the ATA MCX cost, along with the cost of a single long-range CNOT over $k+n$ LNN qubits as these scale as our best case, and the cost of such a CNOT over $k+2n$ which scales as our worst case.
    (a) The costs as a function of $n\in[1,k-2]$ for a fixed number of qubits $k=100$ (b) The costs as a function of $k\in[n+2,10^3]$ for a fixed number of control qubits $n=100$.     
    It is noticeable that the average cost approaches the best case when the number of unused qubits $k-n$ approaches $\approx k$, and $\approx 0$.} 
    \label{fig:avg_cnot_costs}    
\end{figure}



\section{Conclusion}
In this paper, we presented new methods for the decomposition of multi-controlled {(MC)} gates in LNN connectivity. Specifically, we focused on implementing the MCX with one dirty ancilla, and the MCSU2 without ancilla. Our methods provide decompositions of these gates for any choice of the control, target and dirty ancilla qubits. The count of T and H gates scale as $\sim16n$ and $\sim8n$, respectively, as in the best known ATA methods \cite{zindorf_efficient_2024,maslov_advantages_2016}, and the CNOT count requires the minimal overhead of $\sim4(k-n)$ which is required even for the smallest MCX gate - a long-range CNOT gate over $k$ qubits \cite{kutin_computation_2007}.
Our CNOT count upper bounds: $4k+8n-14$ and $4k+8n-16$ for MCSU2 and MCX, respectively, provide the minimal results, not only for large $k,n$ instances, but also for small choices, for example -- the smallest version of an MCSU2 (k=2,n=1) requires two CNOT gates as precisely given by our result.

The MCSU2 requires $R_z$ gates with arbitrary angles as well. We guarantee that only $8$ of these are required, the same as the best known ATA implementation \cite{zindorf_efficient_2024}, and in favorable qubit arrangements, only $6$ are required. In less favorable arrangements, a CNOT overhead of up to $\approx2k$ is required to achieve this. As ATA implementations are agnostic to the qubit arrangement, no overhead is required and therefore, we show for the first time that without ancilla qubit, any MCSU2 gate can be implemented in ATA connectivity using $6$ arbitrary $R_z$ gates, while maintaining the state-of-the-art cost of Clifford+T gates ($\sim12n$ CNOT, $\sim16n$ T and $\sim8n$ H gates). 
Moreover, we show in the Appendix that our results can be further improved by using dirty ancilla qubits, if favorably placed,
demonstrating the practical value of our method for arbitrary cases beyond predicting  analytical upper bounds for the worst cases. 

Since many quantum algorithms rely on multi-controlled gates, our results lead to more efficient implementations by significantly reducing the number of basic gates required, which in turn naturally lowers the error rate. For perspective, our MC gate CNOT count is bounded between the cost of a {\em single} long-range CNOT gate over $k+n$ qubits, and such a CNOT over $k+2n$ qubits. 

\section*{Acknowledgements}
This work was supported by the Engineering and Physical Sciences Research Council [grant numbers EP/R513143/1, EP/T517793/1].
\section*{Competing Interests}
BZ and SB declare a relevant patent application:
United Kingdom Patent Application No. 2507156.4


\appendix

\section*{Appendix}

{
\section{Constant reductions}\label{apx:const_reduc}

While the main goal of this paper is to provide the lowest upper bound for the cost of  multi controlled gates in LNN connectivity, we find it important to to try and reduce the cost or depth when possible, instead of using the highest values for every case.

In this section we provide such reductions which can be achieved in many cases. These only reduce the constant terms of the costs reported in \tab{LNN_costs}, and therefore are more significant when the MC gates are small (in terms of $n,k$). Small MC gates have been extensively studied \cite{nemkov_efficient_2023,gwinner_benchmarking_2021,nakanishi_quantum-gate_2021,nakanishi_decompositions_2024,cruz_shallow_2023,zindorf_all_2025,barenco_elementary_1995} due to their vast use for quantum computation, and as these are likely to be used many times in a given circuit, such that a constant reduction of each of those will directly translate to a scaling reduction in the full circuit.

We start with a simple case which we do not believe have been stated previously. 
We provide a cheaper decomposition for the MCSU2 gate with two controls, such that one of the controls neighbors both the target qubit, and the second control qubit. Along with \qc{CCsu2_mid_new}.a, this will cover all cases for $n=2,k=3$.

We simply use \qc{CCsu2_mid_new}.a with two SWAP gates to relocate the target qubit. Two CNOT gates can commute with the CC$R_x$ gate and cancel out, providing the middle implementation in \qc{CCsu2_top_new}. The implementation on the right hand side is simply achieved by merging a CNOT gate and a CZ to one controlled $iY$ gate.

\[
\scalebox{0.7}{
\Qcircuit @C=0.5em @R=0.7em @!R { 
    \nghost{} & \lstick{} & \qw & \gate{\mathrm{R_{\hat{x}}(\lambda)}} & \qw\\
     \nghost{} & \lstick{} & \qw & \ctrl{-1} & \qw \\
     \nghost{} & \lstick{} & \qw & \ctrl{-1} & \qw \\
}
\hspace{5mm}\raisebox{-13mm}{=}\hspace{0mm}
\Qcircuit @C=0.5em @R=0.7em @!R { 
 \nghost{} & \lstick{} & \ctrl{1} & \targ  & \ctrl{1} & \targ & \ctrl{1} & \qw \\
    \nghost{} & \lstick{} & \targ & \ctrl{-1}  & \gate{\mathrm{R_{\hat{x}}(\lambda)}} & \ctrl{-1}  & \targ & \qw\\
     \nghost{} & \lstick{} & \qw & \qw & \ctrl{-1} & \qw & \qw & \qw \\
}
\hspace{5mm}\raisebox{-13mm}{=}\hspace{0mm}
\Qcircuit @C=0.5em @R=0.55em @!R { 
    \nghost{} & \lstick{} & \ctrl{1}  & \gate{\mathrm{iY}} &  \qw  &\qw & \ctrl{1}   &  \qw & \qw  & \targ & \ctrl{1}  & \qw   \\
    \nghost{} & \lstick{} & \targ &  \ctrl{-1} &  \gate{\mathrm{\Pi^{\lambda\!/\!4}_{\bar{x}}}} & \multigate{1}{\mathrm{\Delta_2}}  & \gate{\mathrm{Z}} &  \multigate{1}{\mathrm{\Delta_2^\dagger}} & \gate{\mathrm{\Pi^{\lambda\!/\!4}_{\bar{x}}}}  & \ctrl{-1} & \targ & \qw  \\
    \nghost{} & \lstick{} & \qw & \qw  & \ctrl{-1} & \ghost{\mathrm{\Delta_2}}  & \qw  &  \ghost{\mathrm{\Delta_2^\dagger}}  & \ctrl{-1}   & \qw & \qw & \qw
    \gategroup{2}{5}{3}{6}{0.3em}{-} 
    \gategroup{2}{8}{3}{9}{0.3em}{-} 
}
}
\qcref{CCsu2_top_new}
\]

Any MCSU2 gate with two controls can therefore be implemented over three LNN qubits using no more than $7$ CNOT gates and $8$ $R_{\hat{z}}$ rotations, or $9$ CNOT gates and $6$ $R_{\hat{z}}$ rotations, depending on the chosen implementation of the $C\Pi$-$\Delta$ gates as \circled{2} or \circled{4} in \tab{ccpix_table}.

For the chosen axes $\hat{z}$, and $\lambda=-\pi$, this provides an implementation of the CCiZ gate using $7$ CNOT gates, and $4$ T gates, choosing \circled{2}.  

This gate can be used to replace the middle part of the MCZ-$\Delta$ V-chain structure as follows, in case the top two control qubits are neighboring.

\[
\scalebox{0.75}{
 \Qcircuit @C=1.0em @R=0.38em @!R { 
	 	\nghost{} & \lstick{}  & \ctrl{1}  & \multigate{2}{\mathrm{{\Delta^4}'}} & \qw \\
	 	\nghost{} & \lstick{}  & \ctrl{2} & \ghost{\mathrm{{\Delta^4}'}} & \qw \\
	 	\nghost{} & \lstick{}  & \qw  & \ghost{\mathrm{{\Delta^4}'}} & \qw \\
	 	\nghost{} & \lstick{} & \control\qw  & \qw & \qw \\
 }
     \hspace{5mm}\raisebox{-7mm}{=}\hspace{0mm}
 \Qcircuit @C=0.5em @R=0.4em @!R { 
	 	\nghost{} & \lstick{} & \qw & \ctrl{2} & \multigate{1}{\mathrm{\Delta^3}} & \qw & \qw \\
	 	\nghost{} & \lstick{} &\ctrl{1} & \qw & \ghost{\mathrm{\Delta_3}}   &\ctrl{1} & \qw\\
	 	\nghost{} & \lstick{} & \targ &  \control\qw & \qw & \targ & \qw \\
        \nghost{} & \lstick{} &\ctrl{-1} & \qw & \qw &\ctrl{-1} & \qw 
        \gategroup{1}{4}{3}{5}{0.3em}{--}
 }
    \hspace{5mm}\raisebox{-7mm}{$\Rightarrow$}\hspace{0mm}
 \Qcircuit @C=1.0em @R=0.38em @!R { 
	 	\nghost{} & \lstick{}  & \ctrl{1}  & \multigate{2}{\mathrm{\Delta^4}} & \qw \\
	 	\nghost{} & \lstick{}  & \ctrl{2} & \ghost{\mathrm{\Delta_4}} & \qw \\
	 	\nghost{} & \lstick{}  & \qw  & \ghost{\mathrm{\Delta_4}} & \qw \\
	 	\nghost{} & \lstick{} & \control\qw  & \qw & \qw \\
 }
     \hspace{5mm}\raisebox{-7mm}{=}\hspace{0mm}
\Qcircuit @C=.5em @R=0.0em @!R { 
	 	\nghost{} & \lstick{} & \qw & \ctrl{1} & \qw & \qw \\
	 	\nghost{} & \lstick{} & \qw & \ctrl{1} & \qw & \qw \\
	 	\nghost{} & \lstick{} & \targ & \gate{\mathrm{iZ}} & \targ & \qw \\
	 	\nghost{} & \lstick{}  & \ctrl{-1} & \qw &  \ctrl{-1} & \qw \\
 }
  \qcref{cz_cciz_d}
}
\]

For example, this can be used in the following case.

\[
\scalebox{0.75}{
\Qcircuit @C=1.0em @R=0.2em @!R {
	 	\nghost{} & \lstick{} & \ctrl{1} & \multigate{10}{\mathrm{\Delta}} & \qw\\
	 	\nghost{} & \lstick{} & \ctrl{1} & \ghost{\mathrm{\Delta}} & \qw\\
	 	\nghost{} & \lstick{} & \ctrl{3} & \ghost{\mathrm{\Delta}} & \qw\\
	 	\nghost{} & \lstick{} & \qw & \ghost{\mathrm{\Delta}} & \qw\\
	 	\nghost{} & \lstick{} & \qw & \ghost{\mathrm{\Delta}} & \qw\\
	 	\nghost{} & \lstick{} & \ctrl{1} & \ghost{\mathrm{\Delta}} & \qw\\
	 	\nghost{} & \lstick{} & \ctrl{2} & \ghost{\mathrm{\Delta}} & \qw\\
	 	\nghost{} & \lstick{} & \qw & \ghost{\mathrm{\Delta}} & \qw\\
	 	\nghost{} & \lstick{} & \ctrl{1} & \ghost{\mathrm{\Delta}} & \qw\\
	 	\nghost{} & \lstick{} & \ctrl{2} & \ghost{\mathrm{\Delta}} & \qw\\
	   \nghost{} & \lstick{} & \qw & \ghost{\mathrm{\Delta}} & \qw\\
        \nghost{} & \lstick{} & \control\qw & \qw & \qw\\
 }
 \hspace{5mm}\raisebox{-22mm}{=}\hspace{0mm}
 \Qcircuit @C=0.1em @R=0.2em @!R {
	 	\nghost{} & \lstick{} & \qw  & \qw & \qw & \qw & \qw & \qw & \qw & \qw & \qw & \ctrl{2} & \multigate{1}{\mathrm{\Delta^3}} & \qw & \qw & \qw & \qw & \qw & \qw & \qw & \qw & \qw & \qw \\
	 	\nghost{} & \lstick{} & \qw  & \qw & \qw & \qw & \qw & \qw & \qw & \qw & \ctrl{1} & \qw & \ghost{\mathrm{\Delta_3}} & \ctrl{1} & \qw & \qw & \qw & \qw & \qw & \qw & \qw & \qw & \qw \\
	 	\nghost{} & \lstick{} & \qw  & \qw & \qw & \qw & \qw & \qw & \qw & \ctrl{1} & \targ & \control\qw & \qw & \targ & \ctrl{1} & \qw & \qw & \qw & \qw & \qw & \qw & \qw & \qw \\
	 	\nghost{} & \lstick{} & \qw  & \qw & \qw & \qw & \qw & \qw & \qw & \targ & \ctrl{-1} & \qw & \qw & \ctrl{-1} & \targ & \qw & \qw & \qw & \qw & \qw & \qw & \qw & \qw \\
	 	\nghost{} & \lstick{} & \qw  & \qw & \qw & \qw & \qw & \qw & \targ & \ctrl{-1} & \qw & \qw & \qw & \qw & \ctrl{-1} & \targ & \qw & \qw & \qw & \qw & \qw & \qw & \qw \\
	 	\nghost{} & \lstick{} & \qw  & \qw & \qw & \qw & \ctrl{1} & \targ & \ctrl{-1} & \qw & \qw & \qw & \qw & \qw & \qw & \ctrl{-1} & \targ & \ctrl{1} & \qw & \qw & \qw & \qw & \qw \\
	 	\nghost{} & \lstick{} & \qw  & \qw & \qw & \ctrl{1} & \targ & \ctrl{-1} & \qw & \qw & \qw & \qw & \qw & \qw & \qw & \qw & \ctrl{-1} & \targ & \ctrl{1} & \qw & \qw & \qw & \qw\\
	 	\nghost{} & \lstick{} & \qw  & \qw & \qw & \targ & \ctrl{-1} & \qw & \qw & \qw & \qw & \qw & \qw & \qw & \qw & \qw & \qw & \ctrl{-1} & \targ & \qw & \qw & \qw & \qw \\
	 	\nghost{} & \lstick{} & \qw  & \ctrl{1} & \targ & \ctrl{-1} & \qw & \qw & \qw & \qw & \qw & \qw & \qw & \qw & \qw & \qw & \qw & \qw & \ctrl{-1} & \targ & \ctrl{1} & \qw & \qw \\
	 	\nghost{} & \lstick{} & \ctrl{1} & \targ & \ctrl{-1} & \qw & \qw & \qw & \qw & \qw & \qw & \qw & \qw & \qw & \qw & \qw & \qw & \qw & \qw & \ctrl{-1} & \targ & \ctrl{1} & \qw \\
	 	\nghost{} & \lstick{} & \targ  & \ctrl{-1} & \qw & \qw & \qw & \qw & \qw & \qw & \qw & \qw & \qw & \qw & \qw & \qw & \qw & \qw & \qw & \qw & \ctrl{-1} & \targ & \qw \\
        \nghost{} & \lstick{} & \ctrl{-1} & \qw   & \qw & \qw & \qw & \qw & \qw & \qw & \qw & \qw & \qw & \qw & \qw & \qw & \qw & \qw & \qw & \qw  & \qw & \ctrl{-1} & \qw 
        \gategroup{1}{12}{3}{13}{0.3em}{--} 
 }
 \hspace{5mm}\raisebox{-22mm}{=}\hspace{0mm}
 \Qcircuit @C=0.1em @R=0.2em @!R {
	 	\nghost{} & \lstick{} & \qw & \qw & \qw & \qw & \qw & \qw & \qw & \qw & \qw & \ctrl{1} & \multigate{2}{\mathrm{{\Delta^4}'}} & \qw & \qw & \qw & \qw & \qw & \qw & \qw & \qw & \qw & \qw \\
	 	\nghost{} & \lstick{} & \qw & \qw & \qw & \qw & \qw & \qw & \qw & \qw & \qw & \ctrl{2} & \ghost{\mathrm{{\Delta^4}'}} & \qw & \qw & \qw & \qw & \qw & \qw & \qw & \qw & \qw & \qw \\
	 	\nghost{} & \lstick{} & \qw & \qw & \qw & \qw & \qw & \qw & \qw & \ctrl{1} & \qw & \qw & \ghost{\mathrm{{\Delta^4}'}} & \qw & \ctrl{1} & \qw & \qw & \qw & \qw & \qw & \qw & \qw & \qw\\
	 	\nghost{} & \lstick{} & \qw & \qw & \qw & \qw & \qw & \qw & \qw & \targ & \qw & \control\qw & \qw & \qw & \targ & \qw & \qw & \qw & \qw & \qw & \qw & \qw & \qw \\
	 	\nghost{} & \lstick{} & \qw & \qw & \qw & \qw & \qw & \qw & \targ & \ctrl{-1} & \qw & \qw & \qw & \qw & \ctrl{-1} & \targ & \qw & \qw & \qw & \qw & \qw & \qw & \qw\\
	 	\nghost{} & \lstick{} & \qw & \qw & \qw & \qw & \ctrl{1} & \targ & \ctrl{-1} & \qw & \qw & \qw & \qw & \qw & \qw & \ctrl{-1} & \targ & \ctrl{1} & \qw & \qw & \qw & \qw & \qw \\
	 	\nghost{} & \lstick{} & \qw & \qw & \qw & \ctrl{1} & \targ & \ctrl{-1} & \qw & \qw & \qw & \qw & \qw & \qw & \qw & \qw & \ctrl{-1} & \targ & \ctrl{1} & \qw & \qw & \qw & \qw\\
	 	\nghost{} & \lstick{} & \qw & \qw & \qw & \targ & \ctrl{-1} & \qw & \qw & \qw & \qw & \qw & \qw & \qw & \qw & \qw & \qw & \ctrl{-1} & \targ & \qw & \qw & \qw & \qw \\
	 	\nghost{} & \lstick{} & \qw & \ctrl{1} & \targ & \ctrl{-1} & \qw & \qw & \qw & \qw & \qw & \qw & \qw & \qw & \qw & \qw & \qw & \qw & \ctrl{-1} & \targ & \ctrl{1} & \qw & \qw \\
	 	\nghost{} & \lstick{} & \ctrl{1} & \targ & \ctrl{-1} & \qw & \qw & \qw & \qw & \qw & \qw & \qw & \qw & \qw & \qw & \qw & \qw & \qw & \qw & \ctrl{-1} & \targ & \ctrl{1} & \qw \\
	 	\nghost{} & \lstick{} & \targ & \ctrl{-1} & \qw & \qw & \qw & \qw & \qw & \qw & \qw & \qw & \qw & \qw & \qw & \qw & \qw & \qw & \qw & \qw & \ctrl{-1} & \targ & \qw \\ 
        \nghost{} & \lstick{} & \ctrl{-1} & \qw  & \qw & \qw & \qw & \qw & \qw & \qw & \qw & \qw & \qw & \qw & \qw & \qw & \qw & \qw & \qw & \qw & \qw & \ctrl{-1}  & \qw
        \gategroup{1}{12}{4}{13}{0.3em}{--} 
 }
 }
 \qcref{V_chain_Z_examples}
\]

In case the MCZ-$\Delta$ is used for the construction of MCX or MCSU2, an additional reduction is achieved by implementing the CCiZ gate up to an additional SWAP gate on the control qubits, as the SWAP gates can cancel out, as shown in \cite{zindorf_efficient_2024}.

\[
\scalebox{0.75}{
\Qcircuit @C=0.5em @R=0.3em @!R { 
	 	\nghost{} & \lstick{} & \ctrl{1} & \qswap & \qw \\
	 	\nghost{} & \lstick{} & \ctrl{1} & \qswap \qwx[-1] & \qw \\
	 	\nghost{} & \lstick{} & \gate{\mathrm{iZ}} & \qw & \qw \\
 }
 \hspace{5mm}\raisebox{-6mm}{=}\hspace{0mm}
\Qcircuit @C=0.5em @R=0.3em @!R { 
	 	\nghost{} & \lstick{} & \qswap & \ctrl{1}  & \qw \\
	 	\nghost{} & \lstick{}  & \qswap \qwx[-1] & \ctrl{1} & \qw \\
	 	\nghost{} & \lstick{} & \qw & \gate{\mathrm{iZ}}  & \qw \\
 }
\hspace{5mm}\raisebox{-6mm}{=}\hspace{0mm}
\Qcircuit @C=0.5em @R=0.2em @!R { 
	 	\nghost{} & \lstick{} & \qw & \targ & \qw & \qw & \ctrl{1} & \gate{\mathrm{T^\dagger}} & \targ & \qw & \qw  \\
	 	\nghost{} & \lstick{} & \targ & \ctrl{-1} & \gate{\mathrm{T}} & \targ & \targ & \gate{\mathrm{T}} & \ctrl{-1} & \targ & \qw \\
	 	\nghost{} & \lstick{} & \ctrl{-1} & \qw & \gate{\mathrm{T^\dagger}}  & \ctrl{-1} & \qw & \qw & \qw & \ctrl{-1} & \qw \\
 }
 }
 \qcref{rz_toffoli_LNN}
\]

In case the CCiZ is used to construct a stand-alone MCZ-$\Delta$ gate, or in any case in which the SWAPs do not cancel out, we can implement this gate using \qc{CCsu2_top_new}. However, the depth of the implementation can be reduced as follows. 
The CCiZ gate can be implemented using 4 T gates in depth 2, and 7 CNOT gates as \qc{cciz_deco}. The leftmost CNOT gate is executed in parallel with the CNOT to its left when this decomposition is applied to \qc{cz_cciz_d}.

\[
\scalebox{0.75}{
\Qcircuit @C=1.0em @R=0.2em @!R { 
	 	\nghost{} & \lstick{} & \ctrl{1} & \qw \\
	 	\nghost{} & \lstick{} & \ctrl{1} & \qw \\
	 	\nghost{} & \lstick{} & \gate{\mathrm{iZ}} & \qw \\
 }
 \hspace{5mm}\raisebox{-6mm}{=}\hspace{0mm}
 \Qcircuit @C=0.7em @R=0.2em @!R { 
	 	\nghost{} & \lstick{} & \ctrl{1} & \qswap & \ctrl{1} & \qswap & \qw \\
	 	\nghost{} & \lstick{} & \control\qw & \qswap \qwx[-1] & \ctrl{1} & \qswap \qwx[-1] & \qw \\
	 	\nghost{} & \lstick{} & \qw & \qw & \gate{\mathrm{-iZ}} & \qw & \qw 
   \gategroup{1}{5}{3}{6}{0.5em}{--} 
 }
 \hspace{5mm}\raisebox{-6mm}{=}\hspace{0mm}
\Qcircuit @C=0.5em @R=0.2em @!R { 
	 	\nghost{} & \lstick{} & \qw & \targ & \qw  & \gate{\mathrm{S}} & \ctrl{1} & \gate{\mathrm{H}} & \ctrl{1} & \qw  & \gate{\mathrm{T}} & \targ & \qw & \qw \\
	 	\nghost{} & \lstick{} & \gate{\mathrm{H}} & \ctrl{-1} & \targ  & \gate{\mathrm{S^\dagger}} & \targ & \gate{\mathrm{T}} & \targ & \targ  & \gate{\mathrm{T^\dagger}} & \ctrl{-1} & \targ & \qw\\
	 	\nghost{} & \lstick{} & \qw & \qw & \ctrl{-1} & \qw  & \qw & \gate{\mathrm{T}} & \qw & \ctrl{-1} & \qw & \qw  & \ctrl{-1} & \qw \\
 }
 }
 \qcref{cciz_deco}
\]
\qc{cciz_deco} was achieved by using the inverted version of \qc{rz_toffoli_LNN} from \cite{zindorf_efficient_2024} to decompose the boxed gates, and applying the following identities:

\[
\scalebox{0.75}{
\Qcircuit @C=1.0em @R=0.6em @!R { 
	 	\nghost{} & \lstick{} & \ctrl{1} & \qswap & \qw \\
	 	\nghost{} & \lstick{} & \control\qw & \qswap \qwx[-1] & \qw\\
 }
 \hspace{5mm}\raisebox{-3.5mm}{=}\hspace{0mm}
\Qcircuit @C=0,5em @R=0.2em @!R { 
	 	\nghost{} & \lstick{} & \qw & \targ & \ctrl{1} & \gate{\mathrm{H}} & \qw \\
	 	\nghost{} & \lstick{} & \gate{\mathrm{H}} & \ctrl{-1} & \targ & \qw & \qw \\
 }
  \hspace{7mm}\raisebox{-6mm}{\text{  and }}\hspace{2mm}
  \Qcircuit @C=0.5em @R=0.2em @!R { 
	 	\nghost{} & \lstick{} & \ctrl{1} & \gate{\mathrm{H}} & \targ & \qw & \qw \\
	 	\nghost{} & \lstick{} & \targ & \qw & \ctrl{-1} & \gate{\mathrm{T^\dagger}} & \qw \\
 }
 \hspace{5mm}\raisebox{-3.5mm}{=}\hspace{0mm}
\Qcircuit @C=0.5em @R=0.2em @!R { 
	 	\nghost{} & \lstick{} & \gate{\mathrm{S}} & \ctrl{1} & \gate{\mathrm{H}} & \qw \\
	 	\nghost{} & \lstick{} & \gate{\mathrm{S^\dagger}} & \targ & \gate{\mathrm{T}} & \qw \\
 }
 }
\]

}

As mentioned, the CCiZ gate with or without a SWAP gate can only be used if the two control qubits which are farthest away from the target of the MCZ-$\Delta$ gate are neighboring. 
We want to get the resulting cost reductions whenever is possible. In case these two qubits are not neighbouring, the top control qubit can be swapped with the qubit below it, so that it is closer to the nearest control, however, this comes at a cost of $4$ CNOT gates required to implement a pair of SWAP gates (after canceling two CNOTs). When this is applied for the implementation of MCX or MCSU2, we also get a reduction of $4$ CNOT gates, and therefore, this swap is practically for free. The reason is simple - the CNOT count of both these gates scale as $4k+O(n)$, and the mentioned swap reduces the value of $k$ by $1$, thus requiring $4$ fewer CNOT gates for the implementation.

Clearly, This can be applied for the bottom qubit as well, noting that our MCX and MCSU2 implementations include MCZ-$\Delta$ gates which are implemented upside-down as well. Finally, for the MCSU2 case, one qubit out  of the top/bottom pairs might be the target qubit. In this case it is still beneficial to apply this procedure, since it will guarantee that $n_++n_-\geq 1$, which reduces the gate count as well.

The following provides the orientation of the CNOT gates for both a control and a target being swapped.

\[
\scalebox{0.7}{
\Qcircuit @C=1.0em @R=0.em @!R { 
	 	\nghost{} & \lstick{} & \ctrl{2} & \qw \\
	 	\nghost{} & \lstick{} & \qw  & \qw\\
	 	\nghost{} & \lstick{} & \ar @{.} [1,0] & \\
        \nghost{} & \lstick{} &  & \\
	 	\nghost{} & \lstick{} & \qw  & \qw\\
	 	\nghost{} & \lstick{} & \gate{\mathrm{R_{\hat{x}}(\lambda)}}\ar @{-} [-2,0] & \qw \\
 }
 \hspace{5mm}\raisebox{-9mm}{=}\hspace{0mm}
 \Qcircuit @C=1.0em @R=0.em @!R {
	 	\nghost{} & \lstick{} & \qswap & \qw & \qswap & \qw \\
	 	\nghost{} & \lstick{} & \qswap \qwx[-1] & \ctrl{1} & \qswap \qwx[-1] & \qw \\
	 	\nghost{} & \lstick{} &  & \ar @{.} [1,0] &  & \\
	 	\nghost{} & \lstick{} &  &  & & \\
	 	\nghost{} & \lstick{} & \qswap & \gate{\mathrm{R_{\hat{x}}(\lambda)}} \ar @{-} [-1,0] & \qswap &  \qw\\
	 	\nghost{} & \lstick{} & \qswap \qwx[-1] & \qw & \qswap \qwx[-1] & \qw \\
 }
 \hspace{5mm}\raisebox{-9mm}{=}\hspace{0mm}
 \Qcircuit @C=1.0em @R=0.em @!R { 
	 	\nghost{} & \lstick{} & \targ & \ctrl{1} & \qw & \ctrl{1} & \targ & \qw\\
	 	\nghost{} & \lstick{} & \ctrl{-1} & \targ & \ctrl{1} & \targ & \ctrl{-1} & \qw \\
	 	\nghost{} & \lstick{} &  &  & \ar @{.} [1,0] &  &  &  \\
	 	\nghost{} & \lstick{} &  &  &  &  &  &  \\
	 	\nghost{} & \lstick{} & \targ & \ctrl{1} & \gate{\mathrm{R_{\hat{x}}(\lambda)}} \ar @{-} [-1,0] & \ctrl{1} & \targ & \qw \\
	 	\nghost{} & \lstick{} & \ctrl{-1} & \targ & \qw & \targ & \ctrl{-1} &  \qw\\
 }}
\]

This process can be repeated until the top two qubits are neighbors, as well as the bottom two qubits. In addition to the gate count reductions, as can be seen, it allows to apply CNOT gates in parallel, thus reducing the depth. Moreover, if this is applied many times, a chain of partial SWAP gates is created, and the depth can be reduced even more, as was shown in \cite{zindorf_efficient_2024}. 

For a specific example -- for a MCX gate with one control (long range CNOT), the procedure will only stop when $k=3$, and the ancilla qubit is between the control and the target. This next-nearest-neighbor CNOT is then applied using $4$ CNOT gates (instead of $5$ if implemented using two more partial SWAPs). In this case we get a decomposition which resembles the long-range CNOT implementation described in \cite{kutin_computation_2007}, with a reduction (achievable for $k\geq 3$) of $1$ for both the CNOT depth and cost.

Finally, we show how to implement any MCSU2 gate using the optimal number of $5$ $R_{\hat{z}}$ gates. This is achieved without any increase to the other costs, in case $n_2=0$, and $n_+=n_-=0$, i.e., the target is at the bottom, and the qubit above it is a dirty ancilla.
The reduction is simply achieved since the $\Delta$ gates of our structure do not apply on the target qubit in this case. While this can be shown using our framework, it can be realized from the known identity provided in \cite{barenco_elementary_1995} which provides a singly-controlled $SU(2)$ gate using two CNOT gates and $5$ rotations about the $\hat{x}/\hat{y}/\hat{z}$ axes (all of which can be converted to $R_{\hat{z}}$ using H or S gates).
The decomposition in this case is as follows, using notation from \cite{barenco_elementary_1995}.

\[
\scalebox{0.7}{
\Qcircuit @C=0.5em @R=0.7em @!R { 
    \nghost{} & \lstick{} & \qwl & \ctrlpd{2}{C_1} & \qw \\
    \nghost{} & \lstick{} & \qw & \qw & \qw \\
    \nghost{} & \lstick{} & \qw & \gate{\mathrm{W}} & \qw\\
}
\hspace{5mm}\raisebox{-7mm}{=}\hspace{0mm}
\Qcircuit @C=0.5em @R=0.7em @!R { 
    \nghost{} & \lstick{} & \qwl & \qw & \ctrlpd{2}{C_1\;\;\;\;\;\;} & \multigate{1} {\mathrm{\Delta_1}} & \qw  & \qw  & \multigate{1} {\mathrm{\Delta_1^{\dagger}}} & \qw & \ctrlpd{2}{C_1}  & \qw  & \qw & \qw & \qw \\ 
    \nghost{} & \lstick{} & \qw & \qw & \qw & \ghost{\mathrm{\Delta_1}} & \qw & \qw    & \ghost{\mathrm{\Delta_1^{\dagger}}} & \qw & \qw   & \qw & \qw   & \qw  & \qw \\
    \nghost{} & \lstick{} & \qw & \gate{\mathrm{A}}  & \targ & \qw & \qw & \gate{\mathrm{B}}    & \qw & \qw & \targ &   \qw & \gate{\mathrm{C}} & \qw & \qw  
    \gategroup{1}{9}{3}{11}{0.3em}{--} 
    \gategroup{1}{5}{3}{6}{0.3em}{--} 
}
}
\qcref{mcsu2_barenco}
\]

The used MCX-$\Delta$ gates in this case are those achieved for $n_+=n_-=0$ in \eq{MCX_d_final_costs}, as an MCZ-$\Delta$ with two H gates. In ATA this structure can always be used, providing one dirty ancilla is available, at the costs reported in \tab{LNN_costs}. 
In LNN, unfortunately, relocating two qubits comes at a cost of $\sim 4k$, however, since one of the qubits is a dirty ancilla, the cost can be reduced to $\sim 2k+2n$ since the dirty ancilla can be freely chosen out of all unused qubits. To show this we can simply mark $d_1$ and $d_2$ as the number of qubits below the target and below the lowest non-control qubit respectively.
The number of swaps needed to locate both at the bottom is therefore $D\approx d_1+d_2$.
Since there are $\approx d_2$ control qubits below the lowest non-control, the number of controls above the highest non-control qubit cannot be larger than $\approx n-d_2$, and thus the number of qubits above that qubit is no larger than $d'_2\approx n-d_2$.
The number of qubits above the target is $d'_1\approx k-d_1$.
Therefore, the number of swaps required to get both qubits to the top is $D'\approx d'_1+d'_2=k+n-D$.
We can always choose the orientation (top/bottom) to minimize the number of swaps to $min(D',D)$. This value is maximized for $D'=D$,
and so, the worst case is $D\approx\frac{k+n}{2}$. The number of required SWAP gates is $2D$, each costs two CNOT gates after cancellations.

\section{Scaling depth reductions}\label{apx:MCZd_reduce_depth}

Given many possible implementations of a circuit, all with the same gate count, one should generally choose to implement the one with the smallest depth.
Therefore, while the primary goal of this paper is to minimize the cost of MC gates, we wish to address the circuit depth as well.  We provide depth reductions which can be applied to our structure without increasing the gate count. We focus on reducing the depth of the MCZ-$\Delta$ gates, as it is used for the construction of all other gates which we discussed, and is the only scaling part of these structures.

First, it is clear that a pair of CNOT gates from \qc{mcz_d_new}, in case $n_-=0$, can be replaced with a pair of SWAP gates, only changing the $\Delta^{j}$ gate to be equal to $\Delta^{j-1}$ as follows.

\[
\scalebox{0.75}{
\Qcircuit @C=0.5em @R=0.2em @!R { 
    \nghost{{Q}^{j-1}_1 :  } & \lstick{{Q}^{j-1}_1 :  } & \qwl & \ctrlp{3}{C^{j-1}_1} & \multigate{2} {\mathrm{\Delta^{j}}} & \qw & \qw\\
    \nghost{{q}_{j-2} :  } & \lstick{{q}_{j-2} :  } & \qw & \qw& \ghost {\mathrm{\Delta^{j}}} & \qw\\
    \nghost{{q_{j-1}} :  } & \lstick{{q_{j-1}} :  } & \qw & \qw & \ghost {\mathrm{\Delta^{j}}} & \qw\\
    \nghost{{q}_{j} :  } & \lstick{{q}_{j} :  } & \qw & \gate{\mathrm{Z}} & \qw & \qw\\
}
\hspace{5mm}\raisebox{-9mm}{=}\hspace{0mm}
\Qcircuit @C=0.5em @R=0.2em @!R { 
    \nghost{{Q}^{j-1}_1 :  } & \lstick{{Q}^{j-1}_1 :  } & \qwl & \qw  & \qw  & \qw  & \ctrlp{2}{C^{j-1}_1\;\;\;\;\;} & \multigate{1} {\mathrm{\Delta^{j-1}}}   & \qw   & \qw & \qw  & \qw \\ 
    \nghost{{q}_{j-2} :  } & \lstick{{q}_{j-2} :  } & \qw & \qw  & \qw & \qw  & \qw & \ghost{\mathrm{\Delta^{j-1}}}   & \qw & \qw  & \qw  & \qw \\
    \nghost{{q_{j-1}} :  } & \lstick{{q_{j-1}} :  } & \qw & \qw  & \qw & \targ   & \gate{\mathrm{Z}} & \qw  & \qw & \targ & \qw & \qw  \\
    \nghost{{q}_{j} :  } & \lstick{{q}_{j} :  } & \qw & \qw  & \qw & \ctrl{-1}   & \qw & \qw & \qw  & \ctrl{-1}   & \qw & \qw 
    \gategroup{2}{6}{4}{6}{1em}{-} 
    \gategroup{2}{10}{4}{10}{1em}{-} 
    \gategroup{1}{7}{3}{8}{0.3em}{--} 
}
  \hspace{5mm}\raisebox{-9mm}{$\Rightarrow$}\hspace{0mm}
  \Qcircuit @C=0.5em @R=0.2em @!R { 
    \nghost{{Q}^{j-1}_1 :  } & \lstick{{Q}^{j-1}_1 :  } & \qwl & \qw  & \qw  & \qw & \qw  & \ctrlp{2}{C^{j-1}_1\;\;\;\;\;} & \multigate{1} {\mathrm{\Delta^{j-1}}}    & \qw & \qw & \qw  & \qw \\ 
    \nghost{{q}_{j-2} :  } & \lstick{{q}_{j-2} :  } & \qw & \qw & \qw  & \qw & \qw  & \qw & \ghost{\mathrm{\Delta^{j-1}}}       & \qw & \qw  & \qw  & \qw \\
    \nghost{{q_{j-1}} :  } & \lstick{{q_{j-1}} :  } & \qw & \qw  & \qw &\qswap & \qw & \gate{\mathrm{Z}} & \qw & \qw & \qswap & \qw & \qw  \\
    \nghost{{q}_{j} :  } & \lstick{{q}_{j} :  } & \qw & \qw  & \qw & \qswap \qwx[-1] & \qw & \qw & \qw & \qw  & \qswap \qwx[-1]   & \qw & \qw 
    \gategroup{2}{6}{4}{6}{1em}{-} 
    \gategroup{2}{11}{4}{11}{1em}{-} 
    \gategroup{1}{8}{3}{9}{0.3em}{--} 
}
 }
 \qcref{add_C_mcz_swap}
\]

The same holds for the pair of CNOT gates used for $\{Z\}^3_{C,Q}$ and $\{Z\}^4_{C,Q}$, as follows.

\[
\scalebox{0.75}{
\Qcircuit @C=0.5em @R=0em @!R { 
	 	\nghost{} & \lstick{} & \gate{\mathrm{H}} & \targ & \gate{\mathrm{H}}& \qw \\
	 	\nghost{} & \lstick{} & \targ & \ctrl{-1} & \targ & \qw \\
	 	\nghost{} & \lstick{}  & \ctrl{-1} & \qw &  \ctrl{-1} & \qw \\
 }
  \hspace{5mm}\raisebox{-5mm}{$\Rightarrow$}\hspace{0mm}
 \Qcircuit @C=0.5em @R=0em @!R { 
	 	\nghost{} & \lstick{} & \gate{\mathrm{H}} & \targ & \gate{\mathrm{H}}& \qw \\
	 	\nghost{} & \lstick{} & \qswap & \ctrl{-1} & \qswap & \qw \\
	 	\nghost{} & \lstick{}  & \qswap \qwx[-1] & \qw &  \qswap \qwx[-1] & \qw \\
 }
 }
 \hspace{5mm}\raisebox{-8mm}{;}\hspace{0mm}
 \scalebox{0.75}{
 \Qcircuit @C=1.0em @R=0.0em @!R { 
	 	\nghost{} & \lstick{} & \qw & \ctrl{1} & \qw & \qw \\
	 	\nghost{} & \lstick{} & \qw & \ctrl{1} & \qw & \qw \\
	 	\nghost{} & \lstick{} & \targ & \gate{\mathrm{iZ}} & \targ & \qw \\
	 	\nghost{} & \lstick{}  & \ctrl{-1} & \qw &  \ctrl{-1} & \qw \\
 }
 \hspace{5mm}\raisebox{-7mm}{$\Rightarrow$}\hspace{0mm}
 \Qcircuit @C=1.0em @R=0.0em @!R { 
	 	\nghost{} & \lstick{} & \qw & \ctrl{1} & \qw & \qw \\
	 	\nghost{} & \lstick{} & \qw & \ctrl{1} & \qw & \qw \\
	 	\nghost{} & \lstick{} & \qswap & \gate{\mathrm{iZ}} & \qswap & \qw \\
	 	\nghost{} & \lstick{}  & \qswap \qwx[-1] & \qw &  \qswap \qwx[-1] & \qw \\
 }
 }
\]
 Since in general, using a SWAP instead of a CNOT increases the gate count, we will only choose to do so in case some gate cancellation can be applied in order to return to the original count. For example, the $m_4$ box in \qc{depth_reduce_circ} requires two CNOT gates, and the same applies if the SWAP gate is replaced by a CNOT.

Morever, the orientation of a pair of relative phase Toffoli gates can be chosen to be implemented as  \qc{rftoff_ccv_cv}.a ("upwards"), or its upside-down version ("downwards"). The following options can be used to adjust \qc{mcz_d_new}, in case $n_-=1$. 

\[
\scalebox{0.75}{
\Qcircuit @C=0.5em @R=0.2em @!R { 
    \nghost{{Q}^{j-1}_1 :  } & \lstick{{Q}^{j-1}_1 :  } & \qwl & \ctrlp{1}{C^{j-1}_1} & \multigate{2} {\mathrm{\Delta^{j}}} & \qw & \qw\\
    \nghost{{q}_{j-2} :  } & \lstick{{q}_{j-2} :  } & \qw & \ctrl{2}& \ghost {\mathrm{\Delta^{j}}} & \qw\\
    \nghost{{q_{j-1}} :  } & \lstick{{q_{j-1}} :  } & \qw & \qw & \ghost {\mathrm{\Delta^{j}}} & \qw\\
    \nghost{{q}_{j} :  } & \lstick{{q}_{j} :  } & \qw & \gate{\mathrm{Z}} & \qw & \qw\\
}
\hspace{5mm}\raisebox{-9mm}{=}\hspace{0mm}
\Qcircuit @C=0.5em @R=0.2em @!R { 
    \nghost{{Q}^{j-1}_1 :  } & \lstick{{Q}^{j-1}_1 :  } & \qwl & \qw  & \qw  & \qw  & \ctrlp{2}{C^{j-1}_1\;\;\;\;\;} & \multigate{1} {\mathrm{\Delta^{j-1}}}   & \qw   & \qw & \qw  & \qw \\ 
    \nghost{{q}_{j-2} :  } & \lstick{{q}_{j-2} :  } & \qw & \qw  & \qw & \ctrlt{1}  & \qw & \ghost{\mathrm{\Delta^{j-1}}}   & \qw & \ctrlt{1} & \qw  & \qw \\
    \nghost{{q_{j-1}} :  } & \lstick{{q_{j-1}} :  } & \qw & \qw  & \qw & \targ   & \gate{\mathrm{Z}} & \qw  & \qw & \targ & \qw & \qw  \\
    \nghost{{q}_{j} :  } & \lstick{{q}_{j} :  } & \qw & \qw  & \qw & \ctrl{-1}   & \qw & \qw & \qw  & \ctrl{-1}   & \qw & \qw 
    \gategroup{2}{6}{4}{6}{0.8em}{-} 
    \gategroup{2}{10}{4}{10}{0.8em}{-} 
    \gategroup{1}{7}{3}{8}{0.3em}{--} 
}
\hspace{5mm}\raisebox{-9mm}{=}\hspace{0mm}
\Qcircuit @C=0.5em @R=0.2em @!R { 
    \nghost{{Q}^{j-1}_1 :  } & \lstick{{Q}^{j-1}_1 :  } & \qwl & \qw  & \qw  & \qw  & \ctrlp{2}{C^{j-1}_1\;\;\;\;\;} & \multigate{1} {\mathrm{\Delta^{j-1}}}   & \qw   & \qw & \qw  & \qw \\ 
    \nghost{{q}_{j-2} :  } & \lstick{{q}_{j-2} :  } & \qw & \qw  & \qw & \ctrl{1}  & \qw & \ghost{\mathrm{\Delta^{j-1}}}   & \qw & \ctrl{1} & \qw  & \qw \\
    \nghost{{q_{j-1}} :  } & \lstick{{q_{j-1}} :  } & \qw & \qw  & \qw & \targ   & \gate{\mathrm{Z}} & \qw  & \qw & \targ & \qw & \qw  \\
    \nghost{{q}_{j} :  } & \lstick{{q}_{j} :  } & \qw & \qw  & \qw & \ctrlt{-1}   & \qw & \qw & \qw  & \ctrlt{-1}   & \qw & \qw 
    \gategroup{2}{6}{4}{6}{0.8em}{-} 
    \gategroup{2}{10}{4}{10}{0.8em}{-} 
    \gategroup{1}{7}{3}{8}{0.3em}{--} 
}
 }
 \qcref{upwards_downwards_toff}
\]

For visual convenience, we use the definition of the relative-phase Toffoli from \qc{ccpi_two_OG}.a with $\theta=\pi/2$ as follows.
\[
\scalebox{0.75}{
\Qcircuit @C=1.0em @R=0.7em @!R { 
	 	\nghost{} & \lstick{} & \ctrlt{1} & \qw \\
	 	\nghost{} & \lstick{} & \targ & \qw \\
	 	\nghost{} & \lstick{} & \ctrl{-1} & \qw \\
 }
\hspace{5mm}\raisebox{-6mm}{=}\hspace{0mm}
\Qcircuit @C=1.0em @R=0.2em @!R { 
	 	\nghost{} & \lstick{} & \qw & \ctrl{1} & \qw & \ctrl{1} & \qw & \qw \\
	 	\nghost{} & \lstick{} & \gate{\mathrm{H}} & \rpigatesub{S} & \targ & \rpigatesub{S} & \gate{\mathrm{H}} & \qw \\
	 	\nghost{} & \lstick{} & \qw & \qw & \ctrl{-1} & \qw & \qw & \qw \\
 }
}
\qcref{rptoff_pis}
\]


 A depth reduction can be achieved when a gate is added to an existing layer of gates of its type, such that all are applied in the same time-slot. We can simply calculate the total depth reduction for a given circuit, and subtract it from the circuit's cost (which is equivalent to the depth in case all gates are applied sequentially) in order to achieve the depth of the circuit. For example, one immediate depth reduction of two H gates is achieved for any MCZ-$\Delta$ gate with $n\geq 2$ controls, if \qc{cz_cz_d} is used at the center of the MCZ-$\Delta$ V-chain. In this case, the used H gates can be applied in parallel with other H gates in the full decomposition. Alternatively, if \qc{cciz_deco} is used, there is an immediate depth reduction of two T gates and one CNOT.
 
 We refer to the gates applied before the central $\{Z\}^{j_0}_{C,Q}$ ($j_0\in\{3,4\}$) gate in the MCZ-$\Delta$ V-chain decomposition, in addition to the leftmost CNOT in $\{Z\}^{j_0}_{C,Q}$ as the "first chain" which can be written as $\MCO{X}{j_0}{j_0-1}\left(\prod_{j=j_0+1}^{k_1+2}\MCO{X}{\{q_{j},c^j_-\}}{q_{j-1}}\right)$ from \lem{VZ_chain}. Due to the symmetry of the V-chain, the total depth reduction will be twice the reduction achieved for the first chain. We can choose the orientation (upwards/downwards) of each Toffoli, and strategically replace pairs of CNOT with SWAPs in order to minimize the depth. 
 
 We focus on four cases which include an upward Toffoli, each case is defined according to an additional gate applied before (to the left) or after (to the right) of the upward Toffoli:
 \begin{itemize} 
 \item {\em case 1}: A SWAP or a CNOT is to the left. 
 \item {\em case 2}: A downward Toffoli is to the left. 
 \item {\em case 3}: A downward Toffoli is to the right.
 \item {\em case 4}: A SWAP to the right.
 \end{itemize}
 In \qc{depth_reduce_circ}, we use the five leftmost gates in \qc{V_chain_Z_examples} to demonstrate all of these cases, while using the freedom to choose the orientations of the Toffoli gates and to replace CNOT gates with SWAPs. We simply decompose the Toffoli gates using \qc{rptoff_pis} and mark the gates that correspond to {\em case $\gamma\in [1,4]$} as $m_\gamma$. We define the vector $M_\gamma=(M_\gamma^{CX},M_\gamma^{T},M_\gamma^{H})$, such that $M_\gamma^{g}$ holds the depth reduction  of the gate $g\in\{'CX','T','H'\}$ achieved in {\em case $\gamma$}.

\[
\scalebox{0.75}{
\Qcircuit @C=1.0em @R=0.73em @!R { 
	 	\nghost{} & \lstick{} & \qw & \qw & \qw & \qw & \ctrl{1} & \qw \\
	 	\nghost{} & \lstick{} & \qw & \qw & \qw & \ctrlt{1} & \targ & \qw \\
	 	\nghost{} & \lstick{} & \qw & \qw & \qw & \targ & \ctrlt{-1} & \qw \\
	 	\nghost{} & \lstick{} & \qw & \ctrlt{1} & \qswap & \ctrl{-1} & \qw & \qw \\
	 	\nghost{} & \lstick{} & \ctrl{1} & \targ & \qswap \qwx[-1] & \qw & \qw & \qw \\
	 	\nghost{} & \lstick{} & \targ & \ctrl{-1} & \qw & \qw & \qw & \qw \\
	 	\nghost{} & \lstick{} & \ctrlt{-1} & \qw & \qw & \qw & \qw & \qw \\
 }
  \hspace{5mm}\raisebox{-18mm}{=}\hspace{0mm}
\Qcircuit @C=0.2em @R=0.2em @!R { 
	 	\nghost{} & \lstick{}  & \mbox{\; \; \; $m_1$} & & \qw  & \qw & \qw & \qw & \qw & \qw & \qw & \qw & \qw & \qw  & \qw & \qw & \mbox{$m_3$} &  & \ctrl{1} & \qw & \qw & \qw \\
	 	\nghost{} & \lstick{} & \qw & \ctrl{1} & \qw & \qw & \qw & \qw & \qw & \qw & \qw & \qw & \qw & \qw  & \qw & \ctrl{1} & \gate{\mathrm{H}} & \rpigatesub{S} & \targ & \rpigatesub{S} & \gate{\mathrm{H}} & \qw \\
	 	\nghost{} & \lstick{} & \gate{\mathrm{H}} & \rpigatesub{S} & \qw & \qw & \qw & \mbox{$m_2$} &  & \qw & \qw & \qw & \qw & \qw  & \targ & \rpigatesub{S} & \gate{\mathrm{H}} & \ctrl{-1} & \qw & \ctrl{-1} & \qw & \qw \\
	 	\nghost{} & \lstick{} & \qw & \qw & \qw & \qw & \qw & \ctrl{1} & \qw & \qw & \ctrl{1} & \qw & \qw & \qswap & \ctrl{-1} & \qw & \qw & \qw & \qw & \qw & \qw & \qw \\
	 	\nghost{} & \lstick{} & \qw & \qw & \ctrl{1} & \gate{\mathrm{H}} & \qw & \rpigatesub{S} & \qw & \targ & \rpigatesub{S} & \qw & \gate{\mathrm{H}} & \qswap \qwx[-1] & \qw & \qw & \qw & \qw & \qw & \qw & \qw & \qw \\
	 	\nghost{} & \lstick{} & \gate{\mathrm{H}} & \rpigatesub{S} & \targ & \qw & \qw & \rpigatesub{S} & \gate{\mathrm{H}} & \ctrl{-1} & \qw & \mbox{$m_4$} &  & \qw & \qw & \qw & \qw & \qw & \qw & \qw & \qw & \qw \\
	 	\nghost{} & \lstick{} & \qw & \ctrl{-1} & \qw & \qw & \qw & \ctrl{-1} & \qw & \qw & \qw & \qw & \qw & \qw & \qw & \qw & \qw & \qw & \qw & \qw & \qw & \qw 
   \gategroup{4}{8}{7}{8}{1.8em}{--} 
   \gategroup{4}{11}{5}{14}{0.6em}{--} 
   \gategroup{2}{2}{3}{4}{0.3em}{--} 
   \gategroup{2}{16}{3}{18}{0.3em}{--} 
 }
}
\qcref{depth_reduce_circ}
\]


In {\em case 1}, no gates are applied on the same qubits prior to the $m_1$ gates, and therefore, these can be applied in parallel with gates which are located to their left. The depth reduction $M_{1}=(1,2,1)$ is therefore achieved for {\em case 1} if it does not include the leftmost Toffoli gate.

As in {\em case 2}, four $T$ gates are applied in depth 2, and two CNOT gates in depth 1 as $m_2$, we simply get 
$M_2=(1,2,0)$.
Using the following identities, we get 
$M_3=(0,2,1)$ and $M_4=(0,1,0)$.

\[
\scalebox{0.75}{
 \Qcircuit @C=0.5em @R=0.2em @!R { 
	 	\nghost{} & \lstick{} & \ctrl{1} & \gate{\mathrm{H}} & \rpigatesub{S} & \qw \\
	 	\nghost{} & \lstick{} & \rpigatesub{S} & \gate{\mathrm{H}} & \ctrl{-1} & \qw \\
 }
 \hspace{5mm}\raisebox{-3mm}{=}\hspace{0mm}
\Qcircuit @C=0.5em @R=0.2em @!R { 
	 	\nghost{} & \lstick{} & \targ & \gate{\mathrm{T}} & \gate{\mathrm{H}} & \gate{\mathrm{T}} & \ctrl{1} & \qw\\
	 	\nghost{} & \lstick{} & \ctrl{-1} & \gate{\mathrm{T^\dagger}} & \gate{\mathrm{H}} & \gate{\mathrm{T^\dagger}} & \targ & \qw \\
 }
  \hspace{7mm}\raisebox{-6mm}{\text{  and }}\hspace{2mm}
\Qcircuit @C=0.5em @R=0.2em @!R { 
	 	\nghost{} & \lstick{} & \ctrl{1} & \qw & \qswap & \qw \\
	 	\nghost{} & \lstick{} & \rpigatesub{S} & \gate{\mathrm{H}} & \qswap \qwx[-1] & \qw \\
 }
 \hspace{5mm}\raisebox{-3mm}{=}\hspace{0mm}
\Qcircuit @C=0.5em @R=0.2em @!R { 
	 	\nghost{} & \lstick{} & \targ & \gate{\mathrm{T}} & \ctrl{1} & \gate{\mathrm{H}} & \qw \\
	 	\nghost{} & \lstick{} & \ctrl{-1} & \gate{\mathrm{T^\dagger}} & \targ & \qw & \qw \\
 }
}
\]


We will now provide a simple algorithm, which determines which CNOT gates should be replaced with a SWAP, and the orientation of each Toffoli gate, in order to maximize the depth reduction.

The first chain consists of $d$ sets of sequential Toffoli gates of various sizes $n_{i\in\{1..d\}}$ such that $\sum_{i=1}^d n_i=n'$ and $d\in [1,n']$, with $n':=n-\abs{C^{j_0}}$ as the total number of Toffoli gates in the first chain. A Toffoli set $i$ and the set $i+1$ to its left are separated by at least one CNOT gate, so by the definition of the first chain, we can guarantee that the rightmost Toffoli of any set will have a CNOT gate to its right, and for $i\not = d$ there will be a CNOT to the left of the leftmost Toffoli.

The following algorithm is designed to maximize the achievable depth reduction, prioritizing {\em case 1}, followed by {\em cases 2}\&{\em 3} and finally {\em case 4}.
We define $\alpha_i\in \{0,1\}$ such that the leftmost Toffoli in set $i$ is upwards if $\alpha_i=1$ and downwards otherwise.

\begin{enumerate}
    \item Choose $\alpha_d$, and set $\alpha_i=1$ for any $i<d$. Define the orientation of the leftmost Toffoli of each set accordingly.
    \item A Toffoli, without a defined orientation, located to the right of a Toffoli with a defined orientation will have the opposite orientation.
    \item Each CNOT gate which is applied directly to the right of an upwards Toffoli is replaced with a SWAP gate. 
\end{enumerate}

We are looking for the depth reduction obtained for any scenario. \\We define $D_i(n_i) = (D^{CX}_i(n_i),D^{T}_i(n_i),D^{H}_i(n_i))$ to hold the depth reduction of each gate type, achieved for a set of $n_i$ sequential Toffoli gates. The overall depth reduction will therefore be:

\[
2\left(D_d(n_d)+\sum_{i=1}^{d-1} D_i(n_i)\right)
\]

In a chain of $n_i$ sequential Toffoli gates, if the leftmost Toffoli is upwards, we get $\ceil{\frac{n_i-1}{2}}$  and $\floor{\frac{n_i-1}{2}}$ {\em case 3} and {\em case 2} reductions, respectively, or {\em case 2} and {\em case 3} respectively if the leftmost Toffoli is downwards. 

In addition to these reductions, {\em case 1} is applied for all $i<d$, noting that for these sets it is guaranteed that there is a SWAP/CNOT gate to the left of an upwards Toffoli, and there is always a $m_1$ gate applied on other qubits beforehand. 

Finally, since any set $i$ of Toffoli gates, in which the rightmost Toffoli is upwards, has one SWAP gate to its right, {\em case 4} occurs if $\alpha_i=\beta_i$ with $\beta_i=(n_i \mod 2)$.

We can therefore write

\[
 D_i(n_i)=
 \begin{cases}
     f(n_d,\alpha_d)M_2+f(n_d,1-\alpha_d)M_3 + g(\alpha_d,\beta_d)M_4 & ,i=d\\     
    f(n_i,1)M_2+f(n_i,0)M_3 + \beta_iM_4 + M_1& ,i<d
 \end{cases}
 \]
with
 \[
 g(\alpha_i,\beta_i)=
 \begin{cases}
     \beta_i & \alpha_i=1\\
     1-\beta_i & \alpha_i=0
 \end{cases}
  \text{ ; }
f(n_i,\alpha_i)
 =
 \begin{cases}
     \ceil{\frac{n_i-1}{2}} & \alpha_i=0\\
\floor{\frac{n_i-1}{2}} & \alpha_i=1     
 \end{cases}
 =
 \frac{n_i-\beta_i-2\alpha_i(1-\beta_i)}{2}.
 \]


We can choose $\alpha_d=\beta_d$, and noting that $f(n_i,\alpha)+f(n_i,1-\alpha)=n_i-1$, we achieve:
\begin{equation}\label{eq:depth_reductions_i}
 D_i(n_i)=
 \begin{cases}
     (\tfrac{1}{2},2,\tfrac{1}{2})n_d+
(-\tfrac{1}{2},0,\tfrac{1}{2})\beta_d-
(0,1,1) & ,i=d\\     
    (\tfrac{1}{2},2,\tfrac{1}{2})n_i
+ (\tfrac{1}{2},1,-\tfrac{1}{2})\beta_i+(0,0,1)& ,i<d
 \end{cases}
\end{equation}
We consider the case in which $d=n'$ as the best case, where $n_i=1$ and therefore $\beta_i=1$ for any $i$. We get $D_d(1)=(0,1,0)$ and $D_{i<d}(1)=(1,3,1)$.
The total depth reduction in the best case is therefore $2((0,1,0)+(n'-1)(1,3,1))=(2n'-2,6n'-4,2n'-2)$.

We define a "useful" dirty quasi-ancilla as a qubit $q_j$ satisfying $q_j\not\in C$ and $q_{j-1}\in C\setminus (C^{j_0}\cup c_n)$, i.e. any non-control qubit located directly below a control qubit, as long as the control qubit is not the one closest to the target of the MCZ gate in $\{ Z \}^{k}_{C,Q}$, and is also not control of the central $\{ Z \}^{j_0}_{C,Q}$ from \lem{VZ_chain}.
By the definition of $d$, and of the first chain of the MCZ-$\Delta$ V-chain, we get that the total number of these useful ancilla is $n_\chi = d-1$. In the best case described above, the maximal value of $n_\chi = n'-1$ is required and therefore the resulting reduction can only be achieved for $k\geq 2n+1-\abs{C^{j_0}}$. 

We take this opportunity to mention that the number of $M_1$ reductions is equivalent to $n_\chi$. This specific reduction is unique, as the corresponding $m_1$ gates can completely commute outside of the V-chain structure, and may allow for cost reductions as well as depth in some cases, as we discuss in \apx{useful_ancill}.

Now, we wish to find the depth reductions in the worst case.
We can write $D_d(n_d)
\geq 
an_d+
b_d$ and $D_{i<d}(n_i)\geq 
an_i
+b$ with $a=(\tfrac{1}{2},2,\tfrac{1}{2})$, $b=(0,0,\tfrac{1}{2})$ and $b_d=-(\tfrac{1}{2},1,1)$, given by choosing the worst values of $\beta_i$ for each gate type in \eq{depth_reductions_i}.
The depth reduction in this case can be rewritten as
\[
2\left(an_d+b_d+\sum_{i=1}^{d-1} (an_i+b)\right)
=
2(an'+db+(b_d-b))
=
(n'-1,4n'-2,n'+d-3).
\]
The worst case is achieved for $d=1$, and is equal to
$(n'-1,4n'-2,n'-2)$. 

\section{Cost reductions in favorable qubit orderings}\label{apx:useful_ancill}

While we report the best known cost upper bound for MCX and MCSU2 in LNN connectivity, without any assumption regarding the choice of the control, target, and dirty ancilla qubits, in practice, when these gates are used as part of a quantum circuit, specific choices are made.
In this case, one should implement these gates in the lowest possible gate count, rather than using the upper bound cost. 
Here we present a method which naturally arises from our structure, and allows to reduce the cost significantly in case the location of the control/target/ancilla are favorable. While this cannot be applied in every case, it allows to provide a lower bound for our method, and in fact provides large reductions for randomly chosen qubit orderings. 

Similarly to the ATA case \cite{maslov_advantages_2016,zindorf_efficient_2024}, dirty ancilla qubits can be used to reduce the cost of MC gates implementations.
We discussed depth reductions which can always be achieved in \apx{MCZd_reduce_depth}, and extra reductions achievable in case 'useful' dirty ancilla are available. 
 The total number of dirty ancilla qubits can be written as $n_\chi=k-n+O(1)$ both for MCSU2 and for MCX.
In LNN connectivity, the cost generally increases with $k$, and therefore it may not be beneficial to add more unused qubits to the circuit in every case.
However, if the dirty ancillas are placed in favorable locations, they allow to provide reductions, which may cancel the added cost.
As mentioned in \apx{MCZd_reduce_depth}, a dirty ancilla qubit becomes useful when placed near a control qubit, on the side closer to the target of the MCZ-$\Delta$ gate, as it increases the number of $m_1$ boxes which can be commuted out of the V-chain structure. 
When used for the construction of the MCX/SU2, due to the structure used, there is also an inversed MCZ-$\Delta$ gate applied on the same qubits, such that if no other gates are applied between these MCZ-$\Delta$ gate, on the same qubits as the $m_1$ boxes, two such $m_1$ boxes simply cancel out.

Moreover, since the target of the gates inside the $m_1$ box is always the useful ancilla qubit, these can commute with the entire MC gate, which provides two additional such cancellations. Therefore, each such ancilla reduces the cost by four $m_1$ boxes, which results in a cost reduction of $(4,8,4,0)$. The maximal number of useful dirty ancilla qubits which can allow such cancellations is bounded by the fact that these must be neighbouring below a controls qubit, and therefore the maximal number of such cancellations is bounded by $n+O(1)$, resulting in a maximal reduction of $(4n,8n,4n,0)+O(1)$, if  $k\geq 2n+O(1)$, and a reduction of $(4(k-n),8(k-n),4(k-n),0) +O(1)$ otherwise.

The best gate count which we can acheve therefore scales as $(4k+4n,8n,4n,0)+O(1)$ if $k\geq 2n+O(1)$, and as $(12n,24n-8k,12n-4k,0)+O(1)$, which is always an improvement as $k\geq n+O(1)$. 

Interestigly, the CNOT gate count for $k\leq 2n+O(1)$ scales the same as the best known implementations in ATA without ancilla qubits \cite{zindorf_efficient_2024,khattar_rise_2024}. For $k\geq 2n+O(1)$, the CNOT gate count scales the same as a single long-range CNOT gate applied on a circuit of size $k+n$. 

In case the dirty ancilla qubit is neighboring a control, albeit on the wrong side, SWAP gates can be applied on the MC gate before starting the decomposition process in order to move the ancilla into place, as long as the cost of applying these SWAPs is not larger than the resulting reduction. In fact, it is clear that no more than four CNOT gates are required to SWAP the dirty ancilla with the control, and therefore, the CNOT count does not increase, and the cost of T gates is reduced. Moreover, we find that in practice, by applying simple identities, two additional CNOT gates can be cancelled in this case, making this qubit swap beneficial for the CNOT count as well in many cases. Since the worst case is given when ancilla reductions are not available, these do not improve the upper bound, however, it increases the chance to achieve such cancellations in most practical cases, and reduces the average gate count over random samples of qubit arrangements.

\section{Small MC$\mathbf{\Pi_{\bar{x}}-\Delta}$}\label{apx:pi_smalls}

We now provide a few lemmas which we refer to in \sec{mcxpidelta}. We have mentioned that two Hadamard gates can be used to apply the  transformation $\bar{z}\rightarrow\bar{x}$ from \lem{MCpi_transform}. This was used to show the correctness of cases \circled{4} and \circled{5} in \tab{ccpix_table}, as well as for the decomposition of case \circled{6} in \qc{ccpi_two_OG}.a.

\begin{lemma}\label{lem:MCpi_transform} 
$\MCO{\omega\Pi^{\theta}_{\bar{x}}}{C}{q_t}=\sqO{H}{q_t}\MCO{\omega\Pi^{\theta}_{\bar{z}}}{C}{q_t}\sqO{H}{q_t}$ with $\omega := e^{i\psi}$
for any angles $\theta,\psi$.
\end{lemma}
\begin{proof}
    From \lem{MC_transform}, 
    $\MCO{\omega\rv{2\theta}{\hat{x}}}{C}{q_t}=\sqO{\rpi{\hat{v}_{H}}}{q_t}\MCO{\omega\rv{2\theta}{\hat{z}}}{C}{q_t}\sqO{\rpi{\hat{v}_{H}}}{q_t}$. As mentioned in \sec{notation}, the vector $\hat{v}_H$ is located in the middle between $\hat{x},\hat{z}$ such that
$\hat{z} = \hat{R}_{\hat{v}_H}(\pi)\hat{x}$, and $H=\rpi{\hat{v}_{H}}$. \lem{2_rpi} can be used to rewrite the equation using only $\Pi$ gates as  
$\MCO{\omega\Pi^{\theta}_{\bar{x}}}{C}{q_t}\MCO{Z}{C}{q_t}=\sqO{H}{q_t}\MCO{\omega\Pi^{\theta}_{\bar{z}}}{C}{q_t}\MCO{X}{C}{q_t}\sqO{H}{q_t}$. We finally apply $\MCO{X}{C}{q_t}\sqO{H}{q_t} = \sqO{H}{q_t}\MCO{Z}{C}{q_t}$.
\end{proof}

To show the correctness of cases \circled{2} and \circled{3}, we have mentioned that a MC$\Pi_{\bar{x}}$-$\Delta$ gate can be implemented using two MCH gates and a single $R_{\hat{z}}$ rotation as \lem{CHH_delta}.

\begin{lemma}\label{lem:CHH_delta}
    $\sqO{\Delta}{\{C,q_t\}}\MCO{\Pi^{\theta}_{\bar{x}}}{C}{q_t} = \MCO{H}{C}{q_t}\sqO{R^\dagger_{\hat{z}}(2\theta)}{q_t}\MCO{H}{C}{q_t}$ with $\sqO{\Delta}{\{C,q_t\}} = \MCO{R_{\hat{z}}(2\theta)}{C}{q_t}\sqO{R^\dagger_{\hat{z}}(2\theta)}{q_t}\MCO{Z}{C}{q_t}$
\end{lemma}
\begin{proof}
    The circuit $\MCO{H}{C}{q_t}\sqO{R^\dagger_{\hat{z}}(2\theta)}{q_t}\MCO{H}{C}{q_t}$ applies $R^\dagger_{\hat{x}}(2\theta)$ on the target if the control set is in state $\ket{11..1}$, and $R^\dagger_{\hat{z}}(2\theta)$ otherwise. Therefore, it is equivalent to  $(\MCO{R_{\hat{z}}(2\theta)}{C}{q_t}\sqO{R^\dagger_{\hat{z}}(2\theta)}{q_t})\MCO{R^\dagger_{\hat{x}}(2\theta)}{C}{q_t}$, such that the part in brackets applies $R^\dagger_{\hat{z}}(2\theta)$ on the target iff the control set is not in state $\ket{11..1}$. Finally, from \lem{2_rpi}, $\MCO{R^\dagger_{\hat{x}}(2\theta)}{C}{q_t} = \MCO{Z}{C}{q_t}\MCO{\Pi^{\theta}_{\bar{x}}}{C}{q_t}$. 
\end{proof}

We have shown a decomposition of the Hermitian gates which are used in case \circled{3}. These are defined as CCH gates up-to a relative phase and one CNOT gate as \qc{mch_circ_circ_deltas} (the same as \qc{mch_circ_circ}.b, providing the equivalent inverted version of this Hermitian gate).
\[
\scalebox{0.7}{
\Qcircuit @C=0.5em @R=0.6em @!R { 
    \nghost{} & \lstick{}  & \ctrltt{1} & \qw\\
    \nghost{} & \lstick{}  & \gate{\mathrm{H}} & \qw\\
    \nghost{} & \lstick{}  & \ctrl{-1} & \qw\\
}
\hspace{5mm}\raisebox{-7mm}{:=}\hspace{0mm}
\Qcircuit @C=0.5em @R=0.6em @!R { 
    \nghost{} & \lstick{} & \qw & \qw  & \ctrl{1} & \multigate{2} {\mathrm{\Delta''}}  & \qw\\
    \nghost{} & \lstick{} & \qw & \targ  & \gate{\mathrm{H}} & \ghost {\mathrm{\Delta''}}   & \qw\\
    \nghost{} & \lstick{} & \qw & \ctrl{-1} & \ctrl{-1} & \ghost {\mathrm{\Delta''}}  & \qw
}
\hspace{5mm}\raisebox{-7mm}{:=}\hspace{0mm}
\Qcircuit @C=0.5em @R=0.6em @!R { 
    \nghost{} & \lstick{} & \qw  & \multigate{2} {\mathrm{{\Delta''}^\dagger}} & \ctrl{1} & \qw & \qw\\
    \nghost{} & \lstick{} & \qw & \ghost {\mathrm{{\Delta''}^\dagger}}`  & \gate{\mathrm{H}} & \targ  & \qw\\
    \nghost{} & \lstick{} & \qw & \ghost {\mathrm{{\Delta''}^\dagger}} & \ctrl{-1} & \ctrl{-1}  & \qw
}
}
\qcref{mch_circ_circ_deltas}
\]

We provided a decomposition of these gates in \qc{mch_circ_circ_decom}.a using \lem{dhx_vsv}.

\begin{lemma}\label{lem:dhx_vsv}
    $\sqO{\Delta''}{\{c_1,c_2,q_t\}}\MCO{H}{\{c_1,c_2\}}{q_t}\MCO{X}{c_2}{q_t} = \MCO{\Pi_V}{c_1}{q_t}\MCO{\Pi_S}{c_2}{q_t}\MCO{\Pi_V}{c_1}{q_t}$ with\\ $\sqO{\Delta''}{\{c_1,c_2,q_t\}}=\MCO{R^\dagger_{\hat{z}}(\tfrac{\pi}{2})}{\{c_1,c_2\}}{q_t}\MCO{R_{\hat{z}}(\tfrac{\pi}{2})}{c_2}{q_t}\MCO{-Z}{\{c_1,c_2\}}{q_t}$
\end{lemma}
\begin{proof}
We show that the following only applies a relative phase \\$\sqO{\Delta''}{\{c_1,c_2,q_t\}}=\MCO{\Pi_V}{c_1}{q_t}\MCO{\Pi_S}{c_2}{q_t}\MCO{\Pi_V}{c_1}{q_t}\MCO{X}{c_2}{q_t}\MCO{H}{\{c_1,c_2\}}{q_t}$.
From \lem{ccrv_reg}, and noting that $\hat{v}_V\perp\hat{x}$ we get \\$\MCO{\Pi_V}{c_1}{q_t}\MCO{X}{c_2}{q_t}=\MCO{X}{c_2}{q_t}\MCO{\Pi_V}{c_1}{q_t}\MCO{-I}{\{c_1,c_2\}}{q_t}$. Then, from \lem{2_rpi} we get $\MCO{\Pi_S}{c_2}{q_t}\MCO{X}{c_2}{q_t} = \MCO{R_{\hat{z}}(\tfrac{\pi}{2})}{c_2}{q_t}$.
So far we have $\sqO{\Delta''}{\{c_1,c_2,q_t\}}=\MCO{\Pi_V}{c_1}{q_t}\MCO{R_{\hat{z}}(\tfrac{\pi}{2})}{c_2}{q_t}\MCO{\Pi_V}{c_1}{q_t}\MCO{-H}{\{c_1,c_2\}}{q_t}$. The following can be verified by checking each of the four options of $c_1,c_2$ being in state $\ket{0}$ or $\ket{1}$:  $\MCO{\Pi_V}{c_1}{q_t}\MCO{R_{\hat{z}}(\tfrac{\pi}{2})}{c_2}{q_t}\MCO{\Pi_V}{c_1}{q_t} = (\MCO{R^\dagger_{\hat{z}}(\tfrac{\pi}{2})}{\{c_1,c_2\}}{q_t}\MCO{R_{\hat{z}}(\tfrac{\pi}{2})}{c_2}{q_t})\MCO{R^\dagger_{\hat{y}}(\tfrac{\pi}{2})}{\{c_1,c_2\}}{q_t}$, using the definition of $\Pi_V$ and \lem{MC_transform}. Finally, from \lem{2_rpi} we get\\ $\MCO{R^\dagger_{\hat{y}}(\tfrac{\pi}{2})}{\{c_1,c_2\}}{q_t}\MCO{-H}{\{c_1,c_2\}}{q_t}=\MCO{-Z}{\{c_1,c_2\}}{q_t}$.
\end{proof}

From the definition in \qc{mch_circ_circ_deltas}, it is clear that the relative phase gates $\Delta',{\Delta'}^\dagger$ cancel out when used in case \circled{3} as these commute with the $R_{\hat{z}}$. 
Therefore, without the added CNOT gates, this circuit implements $\sqO{\Delta}{\{C,q_t\}}\MCO{\Pi^{\theta}_{\bar{x}}}{C}{q_t}$, with the $\Delta$ gate given by \lem{CHH_delta}. Then we can consider the effect of the pair of CNOTs on this gate. When a CNOT gate is commuted with the $\Delta$ gate, it simply transforms it to another relative-phase gate, and when the CNOT is commuted with the $\MCO{\Pi^{\theta}_{\bar{x}}}{C}{q_t}$ gate, as both are applied on the same target, it can be realized from \lem{ccrv_reg} that a CCSU2 gate is added, and since $\hat{x}\perp\hat{v}^{\theta}_{\bar{x}}$, this added gate applies a $-I$ on the target, which is equivalent to a CZ gate applied on the control qubits (\qc{mcx_to_mcz_to_2pi}.a) - only adjusting the relative phase as well.


Finally, we provided the following gate (the same as \qc{mch_circ_circ}.c, providing the equivalent inverted version of this Hermitian gate, along with its decomposition from \qc{ccipi_neww} using the notation in \qc{rftoff_ccv_cv}) that can be used to replace the boxed CC$\Pi_{\bar{x}}$-$\Delta$ gates in \qc{mcpix_delta_new}. A CNOT gate is simply removed from case \circled{5}, as it can commute with the MCZ-$\Delta$ gate, only changing the relative phase.
\[
\scalebox{0.7}{
{
 \Qcircuit @C=0.2em @R=0.37em @!R { 
	 	 \lstick{}  & \ctrltt{1}  & \qw \\
	 	 \lstick{}  & \gate{\mathrm{\Pi_{\bar{x}}^{\theta}}}  & \qw \\
            \lstick{}  & \ctrl{-1} & \qw \\
   }
 \hspace{2mm}\raisebox{-8mm}{:=}\hspace{2mm}
  \Qcircuit @C=0.2em @R=0.37em @!R { 
	 	 \lstick{}  & \ctrl{1}  & \qw  & \qw  \\
	 	 \lstick{}  & \gate{\mathrm{i\Pi_{\bar{x}}^{\theta}}} & \targ  & \qw \\
            \lstick{}  & \ctrl{-1}  & \ctrl{-1} & \qw \\
   }   
 \hspace{2mm}\raisebox{-8mm}{=}\hspace{2mm}
  \Qcircuit @C=0.2em @R=0.37em @!R { 
	 	 \lstick{} & \qw & \ctrl{1}    & \qw  \\
	 	 \lstick{} & \targ & \gate{\mathrm{-i\Pi_{\bar{x}}^{\theta}}}   & \qw \\
            \lstick{}  & \ctrl{-1}  & \ctrl{-1} & \qw \\
   }   
 \hspace{2mm}\raisebox{-8mm}{=}\hspace{2mm}
   \Qcircuit @C=0.2em @R=0.37em @!R{ 
	 	\lstick{} & \qw & \qw & \ctrlt{1} & \qw & \qw  & \qw \\
	 	\lstick{} & \gate{\mathrm{H}} & \gate{\mathrm{R_z^\dagger(\theta)}} & \targ & \gate{\mathrm{R_z(\theta)}} & \gate{\mathrm{H}}   & \qw \\
        \lstick{} & \qw & \qw & \ctrl{-1} & \qw & \qw & \qw \\
 }
}
}
\]

As can be seen, the additional CNOT and $i$ phase intrduced in this case allow to replace the Toffoli gate, which would require a cost of 8 CNOT and 7 T gates \cite{nemkov_efficient_2023,gwinner_benchmarking_2021,nakanishi_quantum-gate_2021,nakanishi_decompositions_2024,cruz_shallow_2023,zindorf_all_2025}, with its relative phase counterpart which only costs 3 CNOT and 4 T gates.

\end{document}